\documentclass[aps,pra,reprint,groupedaddress]{revtex4-1}
\usepackage{url}
\usepackage[dvips]{graphicx}
\usepackage{amsthm}
\usepackage{amsmath}
\usepackage{physics}
\usepackage{here}
\usepackage{amssymb}
\usepackage{dcolumn}% Align table columns on decimal point
\usepackage{bm}% bold math
\usepackage{subfigure}
\newtheorem{definition}{Definition}
\newtheorem{theorem}{Theorem}
\newtheorem{lemma}{Lemma}
\newtheorem{proposition}{Proposition}
\newtheorem{remark}{Remark}
\newtheorem{Ass}{Assumption}

\allowdisplaybreaks[0]

\begin{document}
\title{Optimal performance of generalized heat engines with finite-size baths of arbitrary multiple conserved quantities beyond i.i.d.~scaling}
%\title{Analysis of non-i.i.d. finite-size effects on generalized heat engines with multiple conserved quantities}

\author{Kosuke Ito$^{1}$~and~Masahito  Hayashi$^{1,2}$\\
\textit{${}^1$Graduate School of Mathematics, Nagoya University, Furocho, Chikusa-ku, Nagoya 464-8602, Japan \\${}^2$
Centre for Quantum Technologies, National University of Singapore, Singapore 117543}}
%\textit{${}^1$Graduate School of Mathematics, Nagoya University, Japan, \\${}^2$
%Centre for Quantum Technologies, National University of Singapore, Singapore}}

\begin{abstract} 
 In quantum thermodynamics, effects of finiteness of the baths have been less considered. In particular, there is no general theory which focuses on finiteness of the baths of multiple conserved quantities. Then, we investigate how the optimal performance of generalized heat engines with multiple conserved quantities alters in response to the size of the baths. In the context of general theories of quantum thermodynamics, the size of the baths has been given in terms of the number of identical copies of a system, which does not cover even such a natural scaling as the volume. In consideration of the asymptotic extensivity, we deal with a generic scaling of the baths to naturally include the volume scaling. Based on it, we derive a bound for the performance of generalized heat engines reflecting finite-size effects of the baths, which we call fine-grained generalized Carnot bound. We also construct a protocol to achieve the optimal performance of the engine given by this bound.
 Finally, applying the obtained general theory, we deal with simple examples of generalized heat engines.
 As for an example of non-i.i.d.~scaling and multiple conserved quantities,
 we investigate a heat engine with two baths composed of an ideal gas exchanging particles, where the volume scaling is applied. The result implies that the mass of the particle explicitly affects the performance of this engine with finite-size baths.
\end{abstract}

\maketitle

\section{Introduction}
\subsection{Motivation}
Thermodynamics has succeeded in revealing the
universal principles of nature
since its origin by Carnot \cite{Carnot:1824aa}.
Carnot efficiency is given only by the temperatures of heat baths
independently of other details of the systems.
Coarse-grained perspective of extremely enormous systems enables
such descriptions by a few number of quantities.
On the other hand,
%it is also well known that fluctuations are coming relevant as we focus on finer scale structures.
it is ubiquitous in physics that effective theories alter in accordance with the scale.
Researchers are now working on various scales of thermodynamics from microscopic to macroscopic.
Recent explosion of studies on resource theories of quantum thermodynamics has worked out
fine-grained thermodynamic laws of small systems
\cite{PhysRevLett.111.250404,Horodecki:2013aa,1367-2630-17-8-085004,1367-2630-18-1-011002,sai-qthermo2015,1751-8121-49-14-143001,Janzing:2000ab,Brandao:2015aa}.
Moreover, quantum thermodynamics of multiple conserved quantities
including non-commutative observables
has also been actively studied \cite{e15114956,PhysRevLett.118.060602,Guryanova:2016aa,1367-2630-19-4-043008,Vaccaro1770,Yunger-Halpern:2016aa,PhysRevE.93.022126} recently.
%Thermodynamics not only of energy has long history as well.
A primary system with multiple conserved quantities in thermodynamics is a system which exchanges the energy and the particle number with reservoirs (heat baths and particle baths), whose
thermal state is described by the grand canonical ensemble.
%From the principle of maximum entropy,
Jaynes \cite{PhysRev.106.620,PhysRev.108.171} further generalized
thermodynamics for arbitrary multiple conserved quantities.
%Quantum thermodynamics of multiple quantities are also based on generalized notion of bath of each arbitrary quantity from such earlier results.

%which is not only theoretically important but also became practical issue
%by outstanding progress in nanotechnology.
%thermodynamics intro~so-called resource theory introduction
\begin{table*}
  \caption{Regimes treated in conventional quantum thermodynamics. We fill in all the rest of regimes.}
  \label{table:conventional}
  \begin{center}
    \begin{tabular}{c||c|c|c|c}
      \hline
     & \multicolumn{2}{c|}{i.i.d.~scaling} & \multicolumn{2}{c}{Generic scaling} \\
      \hline conserved quantities
     & Thermodynamic limit & Finite-size effects & Thermodynamic limit & Finite-size effects \\
     [-8pt] &&&&\\ \hline \hline
      only energy & many (e.g.\cite{PhysRevLett.111.250404,Horodecki:2013aa,1367-2630-17-8-085004,Brandao:2015aa}) & \cite{PhysRevE.96.012128} & other approaches \cite{1601.00487,1611.06614}  & none \\ \hline
      multiple & \cite{Guryanova:2016aa,PhysRevE.93.022126} & none & none & none \\ \hline
    \end{tabular}
  \end{center}
\end{table*}
Although many researches \cite{PhysRevLett.111.250404,Horodecki:2013aa,1367-2630-17-8-085004,1367-2630-18-1-011002,sai-qthermo2015,1751-8121-49-14-143001,Janzing:2000ab,Brandao:2015aa,e15114956,PhysRevLett.118.060602,Guryanova:2016aa,1367-2630-19-4-043008,Vaccaro1770,Yunger-Halpern:2016aa,PhysRevE.93.022126} of quantum thermodynamics focused on the finiteness of the working substance of thermal machines,
less studies has been done on the finite-size effects of the heat baths.
Heat baths are treated as unboundedly available resources by the majority of conventional researches.
%In representative researches  of quantum thermodynamics, ``small'' just refers to
%the finiteness of the system interacting with the surrounding baths.
%On the other hand, finite-size effects of the baths have been less focused, since they treated baths as unboundedly available resources.
As pointed out by \cite{PhysRevE.96.012128,1605.06092}, 
the baths should be treated as finite resources when the size of the baths is restricted during the thermodynamic process. For example, when the source and sink are given as mesoscopic systems,
such a formulation is desired.
%In particular, such a treatment is desired for using mesoscopic systems as source and sink.
Very recently, this topic attracts increasing attentions \cite{PhysRevE.96.012128,1607.01302,1605.06092,1506.02322,1367-2630-16-10-103011}.
%%%%%%%%%
%It is very convenient to clarify the finite-size effect more simply. %%
In particular, Tajima and Hayashi \cite{PhysRevE.96.012128} derived the asymptotic expansion of the optimal efficiency of heat engines with respect to the system-size $n$, the number of identical copies of the baths.
In this expansion, since the first leading term expresses the optimal efficiency with thermodynamic limit, the second leading term expresses the finite-size effect appearing in the optimal efficiency.
%That is, the second leading term characterizes how the system-size $n$ is reflected in the difference between the thermodynamic limit and the finite-size situation.
Although this type of argument is not common in quantum thermodynamics, it became very common in recent years in quantum and classical information theory \cite{Strassen:2009aa,Polyanskiy:2010aa,Polyanskiy:2008aa,Hayashi:2008aa,hayashi09:_infor_spect_approac_secon_order}, which is often called second order asymptotics.
We can expect that the second leading term has similar importance in quantum thermodynamics.

%The power of such asymptotic approach is to precisely reveal what quantities are how relevant in what order of scales,
%a quantitative systematic analysis of the effective theories.
%%%%%%%%%%%%%%%
%They are based on the generalization
%such that the system exchanges each arbitrary quantity with its corresponding bath like particle bath, angular momentum bath, etc.
%The nature of equilibrium state of such system, generalized thermal state, itself is closely studied in \cite{1367-2630-19-4-043008,Yunger-Halpern:2016aa}.
%%%%%%%%%%%%%%%%
%Nevertheless, finiteness of such generalized baths is still untouched.
%We investigate the dependence of the optimal performance of generalized heat engines with arbitrary multiple conserved quantities on the baths' scale, giving as well the simple protocol to achieve the optimal performance.
%^Especially, our result reveals that the second leading order already contains quantum effects from non-commutativity.

Although the paper~\cite{PhysRevE.96.012128} was a first step to quantitative analysis of scale dependency in quantum thermodynamics,
their analysis with finite-size baths is limited to
%just the energy transfer, i.e.,
the case when the energy is extracted from two heat bathes with different temperatures.
In fact, there is no research on {\it finite-size baths} of {\it multiple conserved quantities} in quantum thermodynamics yet (Table \ref{table:conventional}).
In an ordinary heat engine, only the energy transfer is involved.
%Usually, an ordinary heat engine refers to a machine to transfer just energy.
%in the way as the previous paragraph.
In contrast, when a thermal machine transfers multiple conserved quantities, we call it a generalized heat engine.

%From not only theoretical but also practical interest,
%finiteness of reservoirs seems rather to be more important in consideration of arbitrary multiple conserved quantities not only energy.
%Practically, it may be convenient to use not only energy for experiments of small systems.
Many interesting systems involving multiple conserved quantities,
e.g. electric batteries, biological processes, chemical reactions, etc,
are possibly affected by finiteness of the baths.
%are needed to be formulated as finite-size baths.
%Basically, experimental setups for mesoscopic systems involve particle transport together with energy.
To investigate the finite-size effects of generalized baths, we study how the optimal performance of generalized heat engines alters in response to the scale.
For this purpose, we improve the second order asymptotics for multiple conserved quantities.
That is, in the sense of second order asymptotics, we investigate the dependence of the performance of generalized heat engines on the baths' scale.
We also give a simple protocol to achieve the optimal performance.

%Can Landauer's principle by angular momentum be touched?
%Though many situations where finiteness of reservoirs is relevant can be considered %%%
%Thus, finiteness of reservoirs must be more important when multiple conserved quantities not only energy are exchanged,
%%finiteness must be relevant very in multiple conserved quantities-situations
%Is it needed to touch with generalized Gibbs ensemble? But then, it is completely the same context as Guryanova
Next, we revisit `scaling' in quantum thermodynamics.
Most of the existing researches on quantum thermodynamics employ
the identical and identically distribution (i.i.d.)-based scaling, where
the baths are scaled by the number $n$ of identical copies of the system.
In general,
the scaling of systems in nature is not necessarily given as the i.i.d.-scaling
but rather in a more generic form, like the volume of the container including the gas,
as has originally been treated in thermodynamics and statistical mechanics.
%given by continuous values in general, like the volume of the container,
%as well as such discrete the number as number of particles.
%Nevertheless, 
%treat the size of the baths as the number of identical copies of a system.
%Note that even if the size is given by such discrete numbers as the number of particles,
%it does not mean that the system is described by identical copies.
%%%%%%%%%%%%
%Although i.i.d.-scaling may be applicable for some of artificial systems,
%such artificial scaling can not even capture scaling by the volume of the container.
%%%%%%%%%%%
%.
Thus, the i.i.d.-scaling is quite constrained in general.
In particular, to treat the change of the number of particles such as particle transport and chemical reactions,
it is natural to use the scaling in terms of the `volume' of the system.
%but not the number of identical copies of some system.
%%%%%%%%%%
%In particular, when particle transport is involved, i.i.d.-scaling seems to be quite constrained.
To extend the applicability of quantum thermodynamics to a wide range of natural objects,
we establish a more general formulation of scaling beyond the i.i.d-structure.
Especially, we achieve it in consideration of the asymptotic extensivity (recently, Tajima et al.~\cite{1601.00487,1611.06614} independently took other approaches to non-i.i.d.~based on the large deviation property to treat non-i.i.d.~Gibbs states in thermodynamic limit).
Based on such a generic scaling, we construct a protocol for a generalized heat engine
under such a generalized scaling, which is novel even in thermodynamic limit
(Table \ref{table:conventional}).
As a typical example, we deal with a heat engine with two baths composed of ideal gas exchanging particles where the size of the baths is given by the volume.
Applying our general theory, we calculate the finite-size effects on the optimal performance
of this canonical example of a generalized heat engine.

Then, our results are roughly made up of two aspects:
extension of the scaling of the baths to a generic manner,
and generalization of the finite-size reservoir thermodynamics to multiple quantum conserved quantities,
in terms of this generic scaling,
which fills in untouched regimes (Table \ref{table:conventional}).
%finite-size effect may includes woods et al and so on

\subsection{Overview}
\begin{figure*}
\centering
 \includegraphics[clip ,width=7.0in]{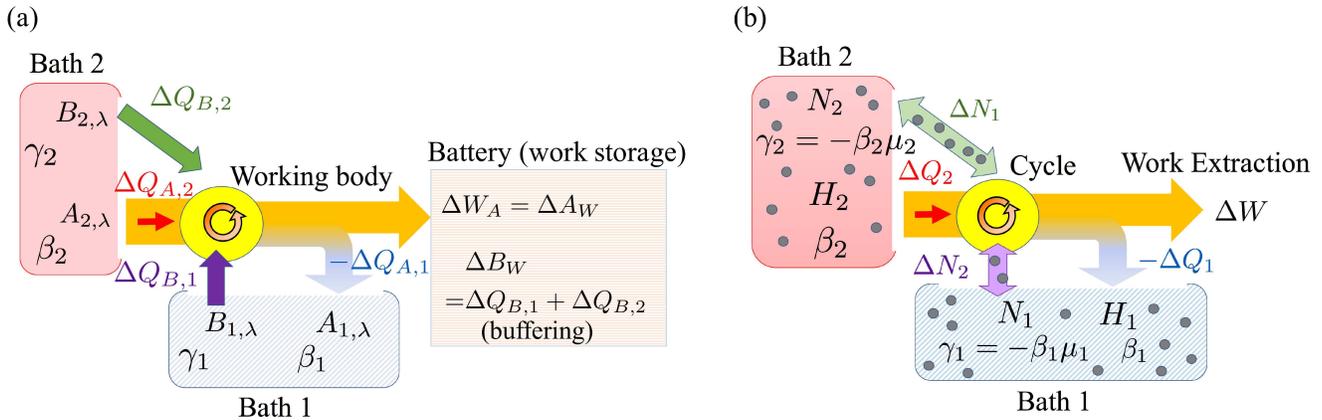}
 \caption{(a) Schematic picture of the model of generalized heat engines with multiple conserved quantities.
 Each bath has two kinds of quantities $A_{i,\lambda}$ and $B_{i,\lambda}$, which are, for example, energy, particle number, $x$-component of angular momentum, etc.
 We assume that generalized inverse temperatures $\beta_i$ and $\gamma_i$ conjugate to $A_{i,\lambda}$ and $B_{i,\lambda}$ respectively are assigned to each bath.
 GCB implies that in extraction of Quantity $A$ through the supply $\Delta Q_{A,2}$ from Bath $2$,
 ``waste heat'' $-\Delta Q_{A,1}$ to Bath $1$ can be reduced by the supplies $\Delta Q_{B,1}$ and $\Delta Q_{B,2}$ of quantity $B$ from the two baths.
 The finite-size effects in FGCB reflect the canonical correlations of the baths' quantities.
 Of course, we can consider a generalized heat pomp in the same way by running it oppositely. (b) Schematic picture of a heat engine with particle transport, which is a primary example of a generalized heat engines. We treat it in detail in Sec.~\ref{sub_particles}.
 This engine extracts work through exchanging heat energy and identical particles with two baths.
 In this case, Quantity $A$ is the energy, and $B$ is the number of identical particles.
 Each bath $i$ has corresponding Hamiltonian $H_{i,\lambda}$ and the particle number $N_{i,\lambda}$.
 Ordinary inverse temperature $\beta_i$ conjugate to $H_{i,\lambda}$ is assigned to each bath.
 The generalized inverse temperature $\gamma_i$ conjugate to the particle number $N_{i,\lambda}$ is given as $\gamma_i=-\beta_i\mu_i$, where $\mu_i$ is the chemical potential of Bath $i$.
 Waste heat can be reduced by compensating it by particle transport.
 Finite-size effects in the optimal performance reflect the canonical correlation of $H_{i,\lambda}$ and $N_{i,\lambda}$, as well as the fluctuation of $H_{i,\lambda}$ and $N_{i,\lambda}$.
}
\label{figure0gen}
\end{figure*}
In this paper,
%based on the quantum statistical mechanical formulation,
we explicitly reveal the effects of the finiteness of the baths on
the optimal performance of a generalized heat engine with multiple conserved quantities, even when they are not necessarily mutually commutative.
Especially, we treat finiteness of the baths by the generic scaling parameter $\lambda$ which
 can be discrete or even continuous.
Instead of assuming the i.i.d.~form scaling of the baths, we just impose the asymptotic extensivity on appropriate
quantities with respect to the scaling parameter $\lambda$.
The deviation from extensivity, because of the finiteness,
may reflect the effects of the interactions and the boundary.
Of course, the i.i.d.-scaling is also covered
since the extensivity is trivially satisfied.
%%%%%%%%%%%%%%%%%%%%%%
%%%%%%%%%%%%%%%%%%
%ref corresponding statement in the main text

%%%%%%%%%%%%%%%%%%
%Guryanova et al. \cite{Guryanova:2016aa} showed the second law of thermodynamics with multiple conserved quantities.
%When we focus on just two baths with two kinds of conserved quantities $A$ and $B$ for simplicity,
% the second law immediately implies
%%%%%%%%%%%%%%%%%%%%%%%
First of all, we focus on the bound on the performance (Sec.~\ref{Sub_opt}).
%%implicit
To this end, we have to impose appropriate constraints on allowed operations.
We have two ways to describe the battery system storing the extracted quantity: implicitly or explicitly \cite{Guryanova:2016aa}.
Implicit-battery formulation just focuses on the operations on the bath $\mathcal{H}_{\mathrm{Baths}}$ and
working body $\mathcal{H}_C$, and describes the extracted quantity as the difference between their quantities before and after the operation so that the battery storing it is implicitly given outside of them.
%Implicit-battery formulation just focuses on the operations on the baths $\mathcal{H}_{\mathrm{Baths}}$ and working body $\mathcal{H}_C$ so that the extracted quantity is given as the difference between their quantities before and after the operation
%and the battery storing it is implicitly given outside of them.
Explicit-battery formulation includes the battery system as an explicit quantum system $\mathcal{H}_W$ so that we explicitly treat the whole system $\mathcal{H}_{\mathrm{Baths}}\otimes \mathcal{H}_C \otimes \mathcal{H}_W$.
Since the implicit-battery formulation describes a part of the whole dynamics,
an operation in the implicit-battery formulation corresponds to many operations in the explicit-battery formulations in general.
In the derivation of an upper bound of the performance,
as weak as possible constraints are preferable for wide
applicability.
That is, a necessary condition for permissible operation is sufficient to impose.
In this sense, we derive a bound under the appropriate implicit-battery formulation (Sec.~\ref{sub_exp_bat}).
%%
%In the derivation of an upper bound of the performance,
%as weaker as possible constraints are preferable for wide applicability.
%That is, a necessary condition for permissible operation is sufficient to impose.
%In this sense, since an operation in the implicit-battery
%formulation corresponds to many operations in the explicit-battery formulations in general, 
%we derive a bound under the appropriate implicit-battery formulation (Sec.~\ref{sub_exp_bat}).
%we just seek the appropriate necessary condition for the implicit operation to be quantum thermodynamically permissible.
%
The recent paper \cite{Guryanova:2016aa} describes an operation in the implicit-battery formulation 
as a unitary operation on $\mathcal{H}_{\mathrm{Baths}}\otimes\mathcal{H}_C$,
corresponding to the dynamics of the system driven by semi-classical external field.
%Operations in implicit-battery formulation refer to unitary operations on $\mathcal{H}_{\mathrm{Baths}}\otimes\mathcal{H}_C$ in \cite{Guryanova:2016aa},
%corresponding to the dynamics of the system driven by a semi-classical external field.
%
However, since the reduced dynamics of $\mathcal{H}_{\mathrm{Baths}}\otimes\mathcal{H}_C$ tracing out the battery $\mathcal{H}_W$ is not unitary without approximation in general,
we include wider class of operations as operations with implicit battery: unital completely positive and trace preserving (CPTP)-maps \cite{PhysRevA.95.032132}.
The unitalness is equivalent to non-decreasing of von Neumann entropy,
which is analogous to adiabaticity with the battery in macroscopic thermodynamics.
Furthermore,
as the reduced dynamics from the operations with an explicit battery,
the unitalness indeed follows from translational symmetry of the battery \cite{Guryanova:2016aa,Skrzypczyk:2014aa,PhysRevA.95.032132},
which is imposed
to guarantee that no hidden heat-like transfer cheatingly improves the performance.
%translation invariance of the battery implies the unitalness in explicit operations \cite{Guryanova:2016aa,Skrzypczyk:2014aa,PhysRevA.95.032132} including our formulation detailed in Sec.~\ref{Sub_achievability}.
%guarantees that only displacements caused by the operations are relevant.
Since we consider the working body executing a cycle, we also impose the cyclicity with respect to $\mathcal{H}_C$.

Let us consider a generalized heat engine with two baths, namely Baths $1$ and $2$, of two kinds of conserved quantities, namely Quantities $A$ and $B$, for simplicity (Fig.~\ref{figure0gen} (a)).
For our formulation, the role of Quantities $A$ and $B$ are essentially the same.
Thus, we focus on the bound on the extraction $\Delta W_A$ of Quantity $A$ without loss of generality.
We can choose $A$ and $B$ as arbitrary conserved quantities.
For example, one may choose the energy as Quantity $A$ to focus on the work extraction, or one may choose the particle number as Quantity $A$ to focus on the extraction of the number of particles.
Heat engine with particle transport (Fig \ref{figure0gen} (b)) is
a canonical example of the generalized heat engine.
Our objective is the upper bound on
the extraction $\Delta W_A$ of Quantity $A$ by a cyclic process where the generalized heat $\Delta Q_{A,i}$ of Quantity $A$ and $\Delta Q_{B,i}$ of quantity $B$ are absorbed from Bath $i$ (Fig.~\ref{figure0gen} (a)).
%the ratio of the maximum work extraction to a given endothermic heat from the hot bath
%is given by the Carnot efficiency.
Under the implicit-battery formulation,
 the second law for multiple conserved quantities \cite{Guryanova:2016aa} immediately implies
 the following upper bound
for the extraction $\Delta W_A$:
\begin{align}
 \Delta W_A\leq \left(1-\frac{\beta_{2}}{\beta_{1}}\right)\Delta Q_{A,2}
 -\sum_{i=1}^{2}\frac{\gamma_{i}}{\beta_{1}}\Delta Q_{B,i},\label{intro_GCB}
 \end{align}
where the baths are initially in the generalized thermal state at
the respective generalized inverse temperatures $\beta_{i}$ and $\gamma_{i}$ corresponding to Quantities $A$ and $B$ of Bath $i$
(For the definitions of 
generalized thermal state and generalized inverse temperature,
see Definition \ref{def_thermal}).
\begin{remark}\label{remark1}
This bound does not include $\Delta Q_{A,1}$ since $\Delta W_A$ is determined if we fix both $\Delta Q_{A,1}$ and $\Delta Q_{A,2}$.
Rather $\Delta Q_{A,1}$ is constrained when the others $\Delta Q_{A,2}$, $\Delta Q_{B,1}$ and $\Delta Q_{B,2}$ are given.
The bound \eqref{intro_GCB} is obtained through this constraint.
This situation is similar to the ordinary Carnot bound, where
the upper bound for the work extraction $\Delta W$ is given in response to the endothermic heat $\Delta Q_h$ from the hot bath:
\begin{align}
 \Delta W\leq \left(1-\frac{\beta_h}{\beta_c}\right)\Delta Q_h,
\end{align}
 where $\beta_h$ and $\beta_c$ are the inverse temperatures of the hot and the cold baths, respectively.
\end{remark}
We call the bound \eqref{intro_GCB} the generalized Carnot bound (GCB)
since this is a straightforward generalization of the Carnot bound, which
has a similar structure depending only on the generalized inverse temperatures.
However,
because of finite-size effects, this bound is never achievable unless thermodynamic limit is taken.

Throughout the paper, we assume that the generalized heat is small enough relative to the scale (see \eqref{heat_scale}) because the baths' state should be unchanged in thermodynamic limit.
Then, by incorporating finite-size effects into GCB, our first main result is
the following inequality, which we call fine-grained generalized Carnot bound (FGCB) (Sec.~\ref{sec_FGCB} Theorem \ref{Thm_FGCB}):
\begin{align*}
 \Delta W_A
 \lesssim & \left(1-\frac{\beta_{2}}{\beta_{1}}\right)\Delta Q_{A,2}
 -\sum_{i=1}^{2}\frac{\gamma_{i}}{\beta_{1}}\Delta Q_{B,i}\nonumber\\
  &-C_{AA}\frac{\Delta Q_{A,2}^2}{\lambda}
  -\sum_{i=1}^{2}C_{AB}^{i}\frac{\Delta Q_{A,2}\Delta Q_{B,i}}{\lambda}\nonumber\\
 &-\sum_{i,j=1}^2C_{BB}^{i,j}\frac{\Delta Q_{B,i}\Delta Q_{B,j}}{\lambda}.
\end{align*}
%FGCB reveals that the terms of the first order of the finiteness already reflect
The leading terms are the same as the GCB, corresponding to thermodynamic limit.
The next leading terms reflect the largest finite-size effect, which is indeed always negative,
so that FGCB does not exceed GCB.
This finite-size effect represents the decrease of the performance 
caused by the non-negligible disturbance to the state of the baths due to their
finiteness.
FGCB gives us the guideline for relieving such drawback.
The canonical correlations of the baths' observables are included in
the coefficients $C_{AA}, C^{i}_{AB}, C^{i,j}_{BB}$ ((\ref{Caa})-(\ref{Cbb}))
of the second leading terms.
Thus, correlations reflecting non-commutativity of the conserved quantities explicitly affects the performance, quite differently from thermodynamic limit described by GCB.
Especially, correlation between Bath $1$ and $2$ is also reflected.
We should consider such correlation structures of the baths to design the engine with finite-size baths.
%Thus, finer structures than the temperatures of the baths,
%the correlations and fluctuations of the baths' quantities, are found to be relevant in this regime.
%, where it is an advantage of this asymptotic analysis to reveal such precise order in first time.

%Quantum effects occurs as the canonical correlations between observables of the baths, which reduces ordinary correlations when the observables are mutually commutative,
We consider the heat engine with baths composed of an ideal gas exchanging particles whose size is given by the volume of the container in Sec.~\ref{sub_particles}, as a physical example (Fig.~\ref{figure0gen} (b)).
Although it is so famous canonical example,
this is the first time to explicitly calculate the finite-size effects on the optimal performance of this kind of heat engine.
%%%%%%%%%%%%%%
%We also estimate the precise order with which $\lesssim$ holds, reflecting boundary effects.
%%%%%%%%%%%%%%%%%
% Guryanova et al work (and Tajima MH relative entropy free entropy difference new perspective %%for finite)
% immediately implies the thermodynamic bound on extraction of a quantity when the other quantities are supplied
% we call generalized Carnot bound (GCB)
%further finiteness result fine grained Carnot bound (FGCB) as follows (Theorem )
%including non-commutative case
%statistical mechanically derive, free from empirical results from macroscopic Thermodynamics.

It is also important to show how the bound can be achieved.
Hence, we construct a protocol to achieve the FGCB (Sec.~\ref{Sub_achievability}).
In construction of the protocol,
we should carefully avoid any hidden heat source which may cheatingly improve apparent performance of the engine,
because the definition of work-like transfer of each quantity is ambiguous in quantum thermodynamics.
Thus, we have to explicitly treat the battery to show the achievability of optimal performance in FGCB.
For an explicit treatment, in addition to the conservation laws and cyclicity of the working body, we should restrict the battery to really work just as a battery but not as a `cold reservoir'.
A reasonable condition is the 'no-cheating' condition \cite{Skrzypczyk:2014aa,Guryanova:2016aa,Masanes:2017aa}, which restricts the protocol to be independent of the state of the battery.
In this way, it is guaranteed that the battery itself can not be used cheatingly as an entropy sink.
In this sense, any exchange of the quantities with the battery does not improve the performance cheatingly as hidden heat-like transfer.
%extraction does not include heat-like amount to cheatingly use the battery's 
We consider a realization of the battery with continuous spectrum to satisfy this no-cheating condition, and finally
construct the protocol with the explicit battery.
%this part may be appropriate for the conclusion (discussion)
Furthermore, there are two types of conservation laws, the strict and average conservation laws \cite{Guryanova:2016aa,PhysRevX.5.021001}.
The strict conservation requires that each quantity commutes with the dynamics, while
the average conservation requires only the conservation of its average value.
When the observables representing the conserved quantities are commutative, 
we construct a protocol satisfying the strict conservation.
However, for the non-commuting case, it is not easy to construct such a protocol.
Instead of this requirement, we construct a protocol satisfying just the average conservation law 
%regardless of the states of both the battery and the working body, 
as in \cite{Skrzypczyk:2014aa,Guryanova:2016aa}.
As pointed out in \cite{PhysRevX.5.021001},
coherence may be indefinitely needed to realize a protocol satisfying just the average conservation.
%Also pointed out by \cite{PhysRevX.5.021001}
However, it is also pointed out in \cite{PhysRevX.5.021001} that
considering resource of coherence appropriately \cite{Aberg:2014aa},
we have a possibility to transform
a protocol satisfying the average conservation law
%with an explicit battery 
to a protocol satisfying the strict conservation law.

%%%%%%%%%%%%%%%%%%%%%%%
%One of the crucial drawbacks of the weak conservation is dependence of the protocol on the state.
%However, our protocol gets through this point in the sense that it is independent of both the states of the working body and the battery, which
%only depends on the temperatures of baths and the amounts of the quantities to supply.
%%%%%%%%%%%%%%%%%%%
%Weak conservation law is sufficient for our regime.

FGCB is ``formally'' attained by the final thermal state at the ideal final inverse temperature defined by (\ref{tl_1})-(\ref{tl_4}) in Sec.~\ref{sec_FGCB}.
However, this final thermal state is not realizable from the initial thermal state by any protocol in general.
Instead, our optimal protocol makes the final state very close to the thermal state at the ideal final inverse temperature.
To show that our protocol indeed achieves FGCB (Theorems \ref{achieve_thm}, \ref{achieve_thm_non}), we impose additional assumptions (Assumption \ref{assump2} and (\ref{Q_order})).
Assumption \ref{assump2} is a stronger version of the asymptotic extensivity which guarantees small enough deviation from the extensivity.
The condition (\ref{Q_order}) requires large enough generalized heat.
Finally, under these assumptions, we show that our protocol achieves the equality in FGCB asymptotically by
making use of information geometric structure of thermal states.
A similar idea was given for an ordinary heat engine in \cite{HTMHpre}.
%We also verify that the fluctuation of the work in our protocol can be negligible even for this finite-size regime.
%%%%%%%%%%%%%%%%%%%%%%%%
%We focus on the average values of the quantities similarly to \cite{Guryanova:2016aa,Skrzypczyk:2014aa} because it makes connection to standard thermodynamics clear.
%Another formulation is so-called single shot thermodynamics \cite{Horodecki:2013aa,1506.02322}.
%Though it gives important new features of
%work extraction in small systems,
%it is quite different from work extraction to macroscopic systems
%at the point that its battery is a wit, which is a two level system with a predetermined energy gap.
%%%%%%%%%%%%%%%%%%%%

%We apply the general theory to some examples.
\subsection{Organization}
This paper is organized as follows.
In Sec.~\ref{sec_setup}, we present the setup for our analysis.
At first, we introduce the generalized heat engine and the generalized thermal state in Sec.~\ref{Sub_setup}.
Next, we bring in a scaling of the baths based on the asymptotic extensivity in Sec.~\ref{sub_extensivity} beyond the identical and independent distributions.
Sec.~\ref{Sub_opt} is devoted to show our first main result fine-grained generalized Carnot bound (FGCB).
The implicit-battery formulation is introduced to deal with the bound on the optimal performance in Sec.~\ref{sub_imp_exp}.
In Sec.~\ref{sec_GCB}, we review the second law of thermodynamics with multiple conserved quantities, and introduce the generalized Carnot bound (GCB).
FGCB is shown in Sec.~\ref{sec_FGCB}.
We construct the optimal protocol to show the tightness of FGCB in Sec.~\ref{Sub_achievability}.
Firstly, we construct an operation with implicit battery in Sec.~\ref{sub_cons_imp}.
Then, in Sec.~\ref{sub_cons_exp}, we extend the implicit-battery protocol to the explicit-battery formulation which is introduced in Sec.~\ref{sub_exp_bat}.
Next, we verify the optimality of the protocol in Sec.~\ref{sub_ach_FGCB}.
From Sec.~\ref{sub_exp_bat} to \ref{sub_ach_FGCB}, we consider commutative quantities.
Then, we extend the construction to non-commutative quantities in Sec.~\ref{sub_noncomm} under the average conservation laws.
We apply the above general theory to some examples in Sec.~\ref{sec_example}.
%We apply the above general theory to some examples in Sec.~\ref{sec_example}.
%We discuss a bit detailed relationship of this result with prior researches in the context of quantum thermodynamics in Sec.~\ref{sec_discussion}.
Finally, the conclusion is in Sec.~\ref{sec_conclusion}.
%fluctuation
%We finally verify that our protocol really achieves FGCB in the finite-size regime.
%
%It is important to show that
%the bound can be achieved, and how to do it
%quasi-static
%Guryanova protocol  step wise
%one-cycle protocol, indep of systems
%implicit explicit
%average
%catalyst

%%%%%%%%%%%%%
%The paper is organized as follows:
%To begin with, we give the setup of the generalized heat engine in Sec.~\ref{sec_setup}.

\section{Setup}\label{sec_setup}
%We consider the general heat engine which extracts multiple quantum conserved quantities by using multiple baths concerning these quantities.
%%%%%%%%%%%%%%%%%%%%%%%%%%%
%% of these quantities.
%%, where baths has a common characteristic scale.
%%%%%%%%%%%%%%%%%%%%
%Our object is $K$ kinds of quantities conserved among the baths and the engine and the battery, which do not necessarily commute each other.
%Specifically, we focus on how the scale of the baths affects the work extraction in detail, not completely taking thermodynamic limit, but asymptotic sense in the scale.
%%%%%%%%%%%%%%%%%%%%%%%%%%%%%%%%%%%%%%%%
%Guryanova et al. \cite{} showed the constraint on the extraction of multiple quantities.
%They showed that
%For well behaved baths, we show the achievability by explicitly constructing the dynamics focusing on the baths which achieves the optimal work extraction in asymptotic sense.
%\subsection{Setup for heat engine with multiple conserved quantities}
\subsection{Heat engine with generalized thermal baths}\label{Sub_setup}
We consider a generalized heat engine to extract arbitrary quantities composed of multiple baths and a working body as Fig.~\ref{figure0gen}.
We denote the system composed of all the baths by $\mathcal{H}_{\mathrm{Baths}}$.
The working body is supposed to execute the cyclic process, which is
denoted by $\mathcal{H}_C$.
In addition,
we denote the battery system to store the extracted quantities by $\mathcal{H}_W$.
All these Hilbert spaces depend on the scale parameter $\lambda$, though we abbreviate the notation.

%The battery system $\mathcal{H}_W$ is also included in explicit treatment of the battery (cf. Sec.~\ref{sub_imp_exp}, \ref{sub_exp_bat}).
The system $\mathcal{H}_{\mathrm{Baths}}$ consists of two generalized baths, Baths $1$ and $2$, each of which exchanges two conserved quantities (Quantities $A$ and $B$) with the working body and the battery.
We set both numbers of the conserved quantities and the baths as two since our results are essentially the same for general multiple baths and quantities.
It is straightforward to generalize our results to the case of arbitrary number of the baths with arbitrarily many conserved quantities.
Especially, for only one bath with two quantities ($m=1, K=2$), it is sufficient to omit one of the baths (see an example in Appendix~\ref{sub_spin}).
%Here, we define every single type of quantity in the sense that the conservation law holds, e.g.
For example, each conserved quantity $A$ or $B$ may stand for
energy, particle number, $x$-component of the angular momentum, etc.
We denote Quantities $A$ and $B$ of Bath $i$ $(i=1,2)$ with the scale $\lambda$ by
$A_{i,\lambda}$ and $B_{i,\lambda}$ respectively.
Then, $\mathcal{H}_{\mathrm{Baths}}$ has the observables $X_{j,\lambda}$ $(j=1,2,3,4)$,
where $X_{1,\lambda}=A_{1,\lambda}$, $X_{2,\lambda}=A_{2,\lambda}$, $X_{3,\lambda}=B_{1,\lambda}$, $X_{4,\lambda}=B_{2,\lambda}$.
In general, we do not assume commutativity of $X_{j,\lambda}$'s.
Especially, quantities from the different baths (e.g. $A_{1,\lambda}$ and $A_{2,\lambda}$) can be correlated.

% each of which corresponds to Type $i$ quantity of Bath $j$, where the index $\lambda$ denotes the scale parameter.
For simplicity, we assume that the dimension $d_\lambda > 4$ of the baths' Hilbert space $\mathcal{H}_{\mathrm{Baths}}$ is finite but depending on the scale $\lambda$.
In addition, we assume that $X_{j,\lambda}\  (j=1,2,3,4)$ and the identity $I$ are linearly independent as real vectors.
Otherwise, the relation $X_{j,\lambda}=\sum_{k\neq j} a_{k}X_{k,\lambda} + a I$ holds for a $j$ with some real numbers $a_{k}$ and $a$,
which implies that $X_{j,\lambda}$ is a redundant quantity since it is just a linear combination of the other quantities plus a constant $a$.
Thus, we assume this linear independence.
%We do not assume the commutativity for any of them.
%
%Moreover, we suppose that all the baths' scales are characterized by one dimensionless scaling parameter $\lambda$.
%
%\footnote{Strictly speaking, $\mathcal{H}_{\mathrm{Baths}}$ and $A_{B,j}^{(i)}$ should be indexed by $\lambda$ as $\mathcal{H}_{\mathrm{Baths},\lambda}$, $A_{B,j,\lambda}^{(i)}$, but we omitted for simplicity of notation here, while we refine the notation when we focus on the scale dependency.}
%Reflecting this, we assume some behavior of the observables of $\mathcal{H}_{\mathrm{Baths}}$ as we mention later.
Note that our scaling of the baths is different from the conventional one where the baths consist of many identical copies of a system.
We just assume the asymptotic extensivity of the baths' quantities with respect to this generic scale parameter $\lambda$, which can even be continuous, as we discuss in detail in the next subsection.
%%%%%as is mentioned in the next subsection.
%Note that our scaling of the bath is not necessarily conventional i.i.d. form \cite{}.
%Each $i$-th kind of quantity is exchanged among these systems through conservation of
%%%%%%%%%%%%%%%%%%%%%%%%%%%%%%%%%
%The dimension of the baths' system is assumed finite but depends on the scale $\lambda$ in general.
%%%%%%%%%%%%%%%is it possible for the system to be infinite-dimensional?????
%%%%%%%%%%%%%%%%%%%%%%%%%%
Suppose that the initial state of the baths is the generalized thermal state with the associated generalized inverse temperatures $\theta^j$ conjugate to $X_{j,\lambda}$ $(j=1,2,3,4)$.
We also denote the generalized inverse temperatures associated with $A_{i,\lambda}$ and $B_{i,\lambda}$ by $\beta_{i}$ and $\gamma_{i}$ respectively to emphasize which quantity and bath correspond to each generalized inverse temperature.
%In the following, `inverse temperature' refers to a generalized inverse temperature.
 The generalized thermal state and generalized inverse temperature are defined as follows:
  \begin{definition}[Generalized thermal state \cite{PhysRev.108.171,Guryanova:2016aa,Yunger-Halpern:2016aa}]\label{def_thermal}
  Let $\mathcal Z(\bm{\theta}):=\tr e^{-\sum_{j=1}^{4}\theta^{j}X_{j,\lambda}}$ be the generalized partition function with $\bm{\theta}=(\theta^1,\theta^2,\theta^3,\theta^4)$,
  the generalized thermal state at a generalized inverse temperature $\bm{\theta}$ is
  \begin{align}
   \tau_{\bm{\theta}}^{(\lambda)}&:=
   \frac{e^{-\sum_{j=1}^{4}\theta^{j}X_{j,\lambda}}}{\mathcal Z(\bm{\theta})}\nonumber\\
   &=\frac{e^{\sum_{i=1}^{2}(-\beta_{i} A_{i,\lambda}-\gamma_{i} B_{i,\lambda})}}{\mathcal Z(\bm{\theta})}.
 %\frac{e^{-\sum_{i=1}^{K}\sum_{j=1}^{m}\beta^{(i)}_jX_{j,\lambda}^{(i)}}}{\tr e^{-\sum_{i=1}^{K}\sum_{j=1}^{m}\beta^{(i)}_jX_{j,\lambda}^{(i)}}}.
  \end{align}
   As a function of the inverse temperature coordinate $\bm{\theta}$, we define the generalized free entropy (also known as the Massieu potential) $\phi_{\lambda}(\bm{\theta}):=\log \mathcal Z(\bm{\theta})$ of the thermal state.
  \end{definition}
  The ordinary grand canonical state is a typical example of the generalized thermal state.
  It is the thermal state of the system exchanging the particles as well as the energy with the large reservoir.
  The particle number and energy of the total system are conserved.
  In this case, observables are $X_{1,\lambda}=H_{1,\lambda}$, $X_{2,\lambda}=H_{2,\lambda}$, $X_{3,\lambda}=\mathcal N_{1,\lambda}$, and $X_{4,\lambda}=\mathcal N_{2,\lambda}$, where $H_{j,\lambda}$ and $\mathcal N_{j,\lambda}$ are the Hamiltonian and the particle number operator of Bath $j$, respectively.
 The generalized inverse temperatures are composed of the inverse temperature $\beta_i$ and the chemical potential $\mu_i$ of each Bath $i$ as
 $\theta^i=\beta_i$, $\theta^{i+2}=-\beta_{i}\mu_{i}$ $(i=1,2)$.
  In the same way as the grand canonical state, the state given by Definition \ref{def_thermal} was shown to be the thermal state of the system exchanging non-commuting charges with a large reservoir \cite{Yunger-Halpern:2016aa}.
  In this sense, we consider a small part of the large reservoir as our finite-size bath.
% Reasonableness of this state as the thermal state even for non-commutative quantities was
% shown in several ways \cite{PhysRev.108.171,Guryanova:2016aa,Yunger-Halpern:2016aa}.
%%%%%%%%%%
  %Jayens Winter Popescu

% though we just focus on the commutative case for the achievability of the bound.
%We fix the initial temperature $\beta_0=(\beta_{0j}^{(i)})_{(i,j)=(1,1)}^{(m,K)}$.
 % which is recently shown to be the equilibrium state for multiple quantum conserved quantities \cite{}.
%battery
%We consider the scale dependency of
%%%%%%%%%%%%%%%%%%%%%%%%%%%
%We consider the unitary operation $U$ on $\mathcal{H}_{\mathrm{Baths}}\otimes\mathcal{H}_C\otimes\mathcal{H}_W$ conserving $\sum_{j=1}^m X_{j,\lambda}^{(i)}+A_W^{(i)}$ for each $i$, where
%$A_W^{(i)}$ is the corresponding observable of $\mathcal{H}_W$.
%Under such conservation, the baths and battery exchanges each quantity through the cycle system $\mathcal{H}_C$ as described in Fig.~\ref{figure1}.
%The state of $\mathcal{H}_C$ itself remains unchanged, while it just works to implement the dynamics.

We regard the average value as the extracted amount of each quantity in the same way as \cite{Guryanova:2016aa,Skrzypczyk:2014aa}.
Another formulation is so-called single-shot thermodynamics \cite{Horodecki:2013aa,1506.02322}.
This formulation of deterministic work
is quite different from work extraction to macroscopic systems
at the point that its battery is a wit, which is a two level system with a predetermined energy gap.
We rather focus on non-deterministic transfer of the quantity.
%, and verify that its fluctuation is permissibly small to be deterministic in our regime, similarly as in the macroscopic statistical mechanics.

To derive the universal optimal performance,
we consider general dynamics of the generalized heat engine (Fig.~\ref{figure0gen} (a))
where the conservation law among the total system including the battery for every conserved quantity.
However, there are two kinds of conservation laws, the strict and average conservation laws \cite{Guryanova:2016aa,PhysRevX.5.021001}.
The strict conservation requires that each quantity commutes with the dynamics, while
the average conservation requires only the conservation of its average value.
It is important to distinguish them when we consider the allowed operations on the whole system $\mathcal{H}_{\mathrm{Baths}}\otimes\mathcal{H}_C\otimes\mathcal{H}_W$ in Sec.~\ref{sub_exp_bat} and \ref{sub_noncomm}.
Since the average conservation follows from the strict conservation, under both conservation laws, the average value of each quantity is exchanged between the baths and the battery through the cyclic process by the working body.
Hence, the sums $- \sum_{i=1}\Delta A_{i,\lambda}$ and $- \sum_{i=1}\Delta B_{i,\lambda}$ of the differences in the average values of Quantities $A$ and $B$ are respectively stored in the battery.
%change in average values as transfer
%\end{Def}
%, that is, we consider the average
%The schematic picture of our formulation is described in Fig.~\ref{figure1}.
%%%%%%%%%%%%%%%%%%%
%The details of the operations are introduced in Sec.~\ref{sub_imp_exp} and Sec.~\ref{sub_exp_bat}.
%%%%%%%%%%%%%%%%%
%The precise statement of these restriction on the operation is in Sec.~\ref{}.

%For each $i$-th kind, the battery $\mathcal{H}_W$ has corresponding observable $A^{(i)}$.
%For example, each kind of quantity corresponds to energy, particle number or angular momentum, etc, as we show an explicit example in the end of this subsection.
%the common characteristic scale $\lambda$, e.g. volume, length and lattice size.
%engine
%The engine part $\mathcal{H}_C$ is supposed to be cyclic under the process, that is, we use $\mathcal{H}_C$ just as a catalyst.
%%%%%%%%%%%%%%%%%%%%%%%%%%%%%%%%%%%%%%figure
%More precisely,

\subsection{Extensivity of baths}\label{sub_extensivity}
We consider the behavior of the heat engine when $\lambda$ grows large under the fixed initial inverse temperature $\bm{\theta}=\bm{\theta}_0$,
which generalizes the consideration of grand-canonical type ensemble.
%For convenience, we use the coordinate ${\bm \theta}$ composed of the inverse temperatures defined as
%$\bm{\theta}:=(\theta^1,\theta^2,\dots,\theta^{mK}):=(\beta^{(t(1))}_{b(1)},\beta^{(t(2))}_{b(2)},\dots,\beta^{(t(mK))}_{b(mK)})$, where $t(i):=\lceil \frac{i}{m}\rceil$, $b(i):=i \mod m$.
%.
The free entropy $\phi_{\lambda}$ is almost the same as the free energy, but rather more natural for dealing with multiple conserved quantities \cite{Guryanova:2016aa}.
It is the generating function of the physical quantities.
The first derivatives are the expectation value
\begin{align}
 \eta_{\lambda,j}(\bm{\theta})
 :=-\pdv{\phi_{\lambda}}{\theta^j} (\bm{\theta})
 = \tr X_{j,\lambda} \tau_{\bm{\theta}}^{(\lambda)}.
\end{align}
This still holds for non-commutative quantities.
As common in information geometry \cite{Amari:2000aa},
$\eta_{\lambda,i}(\bm{\theta})$
can be regarded as a component of the dual coordinate of $\bm{\theta}$ composed of the expectation values (see Appendix \ref{app_pyth})
\begin{align}
 \bm{\eta}_{\lambda}(\bm{\theta})&:=(\eta_{\lambda,1}(\bm{\theta}), \eta_{\lambda,2}(\bm{\theta}), \eta_{\lambda,3}(\bm{\theta}), \eta_{\lambda,4}(\bm{\theta}))\nonumber\\
 &=(\tr A_{1,\lambda}\tau_{\bm{\theta}}^{(\lambda)}, \tr A_{2,\lambda}\tau_{\bm{\theta}}^{(\lambda)}, \tr B_{1,\lambda}\tau_{\bm{\theta}}^{(\lambda)}, \tr B_{2,\lambda}\tau_{\bm{\theta}}^{(\lambda)}).\label{dual_coord}
\end{align}
The second derivatives form the Fisher information matrix composed of the canonical correlation
\begin{align}
 &J_{\lambda,ij}(\bm{\theta})\nonumber\\
 :=&\pdv{\phi_{\lambda}}{\theta^i}{\theta^j} (\bm{\theta})\nonumber\\
 =&\int_{0}^{1}ds\;\tr\left[\left(\tau_{\bm{\theta}}^{(\lambda)}\right)^{1-s}X_{i,\lambda}\left(\tau_{\bm{\theta}}^{(\lambda)}\right)^s X_{j,\lambda}\right]\nonumber\\
 &-\eta_{\lambda,i}(\bm{\theta})\eta_{\lambda,j}(\bm{\theta}).\label{canonical_cor}
\end{align}
The canonical correlation reduces to the covariance for commutative observables.
In the same way, the third derivatives correspond to the skewness.
These statistical quantities are expected to be extensive in thermodynamics.
Thus, it is natural to assume that the free entropy and its derivatives are asymptotically extensive.
%where we use geometric convention of upper index since we treat $\beta^i$ as the coordinate in the information geometric analysis.
More precisely, we impose the following:
  \begin{Ass}\label{assumption}
   There exists
   %a constant $\alpha <1$ characterizing the order of the deviation of the free entropy $\phi_{\lambda}$ from its extensivity a
   an asymptotic density $\phi(\bm{\theta})$ of the free entropy $\phi_{\lambda}(\bm{\theta})$ as a smooth function satisfying the following condition.
  %There exists a smooth function $\phi(\bm{\theta})$ corresponding to the asymptotic density of the free entropy satisfying the following asymptotic.
 As $\lambda\rightarrow\infty$, the free entropy $\phi_{\lambda}$ asymptotically satisfies
 \begin{align}
  \phi_{\lambda}(\bm{\theta})=\lambda\phi(\bm{\theta})+o(\lambda),\label{extensive_1}
 \end{align}
   uniformly on a neighborhood of $\bm{\theta}_0$.
   %   , where a constant $\alpha <1$ characterizes the order of the deviation of the free entropy from its extensivity.
   %%%%%%%%%%%%%%%%%%%%%%%%%%
  %Here, we assume that $\phi(\bm{\theta})$ is smooth, and $\alpha <1$ is a constant which reflects the boundary effects.
 Moreover, up to the third-order partial derivatives of $\phi_{\lambda}$ satisfies the similar condition uniformly on a neighborhood of $\bm{\theta}_0$:
 \begin{align}
  &\left(\pdv{\theta^{i_1}}\right)^{l_1}\left(\pdv{\theta^{i_2}}\right)^{l_2}\left(\pdv{\theta^{i_3}}\right)^{l_3}\phi_{\lambda}(\bm{\theta})\nonumber\\
  =&
  \left(\pdv{\theta^{i_1}}\right)^{l_1}\left(\pdv{\theta^{i_2}}\right)^{l_2}\left(\pdv{\theta^{i_3}}\right)^{l_3}
  \lambda\phi(\bm{\theta})
  +o(\lambda)\label{extensive_2}
 \end{align}
 for all integers $l_1,l_2,l_3$ with $0< l_1+l_2+l_3\leq 3$, $i_1,i_2,i_3\in\{1,2,3,4\}$.
 %and $j_1,j_2,j_3=1,\dots,m$
 %where $\beta^{i}$ is labeled as $\beta^{i}=\beta^{(t(i))}_{b(i)}$ $(i=1,\dots,mK)$.
 The matrix $(\pdv{\phi}{\theta^i}{\theta^{j}} (\bm{\theta}))_{ij}$ is assumed to be full rank.
\end{Ass}
%Note that  by the following information geometric structures.
This asymptotic extensivity is widely expected as long as the system is thermodynamic in large scale since the free entropy should be an extensive quantity as is the case with the free energy.
%Its derivatives, the statistical moments of the physical quantities, are also expected to be extensive.
Especially, the asymptotic extensivity of the free energy was rigorously proved for Hamiltonians with short range interaction \cite{ruelle1969rigorous}, though any similar theorem is not known for the general multiple conserved quantities.
The asymptotic extensivity \eqref{extensive_2} of the derivatives of the free entropy was not generically proved even for Hamiltonian in \cite{ruelle1969rigorous}.
However, its validity is expected for usual systems since the derivatives correspond to extensive quantities in thermodynamic limit, such as the expectation values, fluctuations, and the statistical moments of the extensive quantities \footnote{Especially, explicit assurance of the asymptotic extensivity of the third-derivatives in Assumption \ref{assumption}
is needed for our main analysis in Sec.~\ref{sec_FGCB}, where
we use it to
verify the order of residual terms in fine-grained generalized Carnot bound to be negligible.
}.
In fact, Assumption \ref{assumption} is verified for some examples in Sec.~\ref{sec_example}.
It trivially holds for the i.i.d.~scaling.
A simple example of non-i.i.d.~scaling with the asymptotic extensivity is a spin chain (Sec.~\ref{ising}).
Furthermore, it is also satisfied by an ideal gas in the container (Sec.~\ref{sub_particles}), where the volume is the scaling parameter.

In \eqref{extensive_1}, $\phi(\bm{\theta})$ stands for the asymptotic density of the free entropy in the sense that $\phi(\bm{\theta})=\lim_{\lambda\rightarrow\infty}\phi_{\lambda}(\bm{\theta})/\lambda$.
The first derivatives
$
 \eta_{i}(\bm{\theta}):= -\pdv{\phi}{\theta^i} (\bm{\theta})
 $
 and the second derivatives
 $
 g_{ij}(\bm{\theta}):=\pdv{\phi}{\theta^i}{\theta^j} (\bm{\theta})
 $
of $\phi(\bm{\theta})$ also correspond to the asymptotic densities of the expectation values and canonical correlations, respectively, as seen from
an another expressions for them:
\begin{align}
 \eta_{\lambda,i}(\bm{\theta})
 =& \lambda \eta_i(\bm{\theta}) + o(\lambda),\label{dual_coord0}\\
 J_{\lambda,ij}(\bm{\theta})
 =&\lambda g_{ij}(\bm{\theta})+o(\lambda).\label{asymp_correlation}
\end{align}
Thus, Assumption \ref{assumption} coincides with the existence of the asymptotic density of each extensive quantity, in other words.

%%%%%%%%%%%%%%%%%%%%%%%%%%%%%%%%%%%%%%%%%%%%%%%%%%%%%figure
\begin{figure}[!t]
\centering
\includegraphics[clip ,width=3.2in]{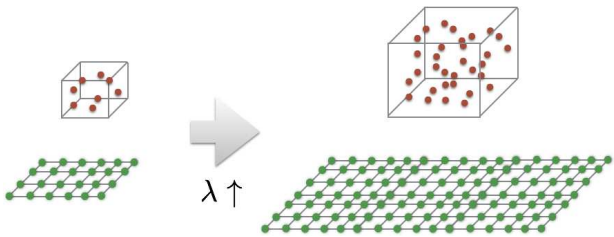}
\caption{Homothetic scaling of the baths. Even if the dimensions of the systems are different from each other, it is sufficient to put $\lambda$ as a dimensionless scaling parameter by defining the unit size of each system.}
\label{figure3}
\end{figure}
%%%%%%%%%%%%%%%%%%%%%%%%%%%%%%%%%%%%%%%%%%%%%%%%%%%%%%%%%f
%By Assumption \ref{assumption}, we implicitly define the meaning of the generic scale parameter $\lambda$ we treat.
%That is,
Here, we consider just one scaling parameter, but not as many parameters as the baths.
That is, we fix the ratio between the sizes of the baths Fig.~\ref{figure3}.
%since the free entropy $\phi_\lambda$ of all the baths asymptotically behaves as $\lambda\phi$, it is implied that all the baths are asymptotically homothetically scaled by $\lambda$ (Fig.~\ref{figure3}).
%In this way, 
Note that this scaling is applicable even if baths contain different dimensional systems or systems with different measures of their sizes by defining the unit size of each system.
For example, consider the case where one bath is a two-dimensional system and the other is of three-dimensional, whose sizes are scaled by their area $S$ and volume $V$ respectively.
Then, defining the unit area $s_0$ and volume $v_0$,
we consider the scaling $S=\lambda s_0$ and $V=\lambda v_0$ by the dimensionless parameter $\lambda$.
In this case, the difference in their dimensionality is putted on their `ratio' $v_0/s_0$ whose dimension is the length.

In order to apply our analysis to a generalized heat engine,
all we have to check is the existence of the scaling $\lambda$ satisfying this property.

%It seems that the deviation $\order{\lambda^{\alpha}}$ from the extensivity reflects the effects of the internal interactions inside the baths, and the effects caused by their boundaries.
%To evaluate
%its order $\order{\lambda^{\alpha}}$, it is enough to deal with the largest such contribution.

The last statement guarantees independence of the observables.
More precisely, the expectation values of the observables can take any combinations under sufficiently large $\lambda$,
since the Fisher information matrix $(\pdv{\phi_{\lambda}}{\theta^i}{\theta^{j}} (\bm{\theta}))_{ij}$ is the same as the Jacobian matrix of the transformation of the variable from $\bm \theta$ to the expectation value.
\section{Generalized Carnot bound for generalized work extraction with finite-size effects}\label{Sub_opt}

\subsection{Operations with implicit battery}\label{sub_imp_exp}
\begin{figure}[!t]
\centering
\includegraphics[clip ,width=3.2in]{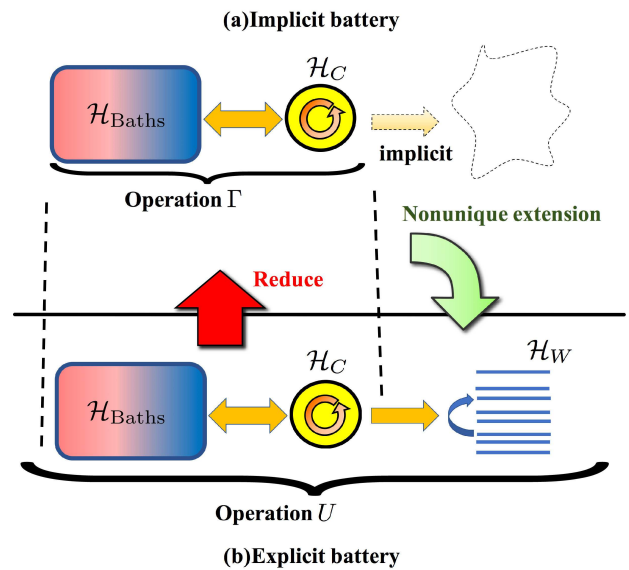}
 \caption{Schematic of (a) implicit-battery formulation and (b) explicit-battery formulation. An explicit-battery formulation is reduced to the corresponding implicit-battery formulation by tracing out $\mathcal{H}_W$, while extension from an implicit-battery formulation to some explicit-battery formulation is not unique.
}
\label{figure_impexp}
\end{figure}
There are two formulations of operations,
the implicit-battery and the explicit-battery formualtions
%which implicitly or explicitly describes the battery system
(Fig.~\ref{figure_impexp}).
%, as is pointed out also in \cite{Guryanova:2016aa}.
The former focuses on the operations only on $\mathcal{H}_{\mathrm{Baths}}\otimes \mathcal{H}_C$, so that the extracted amount of each quantity is stored in the implicitly existing battery outside of $\mathcal{H}_{\mathrm{Baths}}\otimes \mathcal{H}_C$.
The latter explicitly treats the operations on the whole system $\mathcal{H}_{\mathrm{Baths}}\otimes\mathcal{H}_C\otimes\mathcal{H}_W$ under the conditions mentioned in Sec.~\ref{Sub_setup}.
%Obviously, physically relevant formulation is the latter.
%%this expression is not so correct. Implicit operation can be also relevant, once it is checked appropriately
Macroscopic thermodynamics usually employs the implicit-battery formulation since the work is clearly defined and the functionality of the battery system is obvious.
%more non-trivial example of implicit
However,
in quantum thermodynamics,
the definition of the work-like transfer of each quantity itself is ambiguous,
and it is non-trivial to verify that there is no heat-like transfer with the battery,
even in consideration of thermodynamic limit \cite{1302.2811,Skrzypczyk:2014aa,Guryanova:2016aa}.
%to verify that it is really work
Thus, an implicit-battery operation
has no clear meaning as a thermodynamic process,
unless it is extended to an operation in an appropriate explicit-battery formulation.
Such an extension to the explicit-battery formulation is nonunique in general.
%especially for finite-size quantum thermodynamics.
On the other hand, applicability of the upper bound on the extraction becomes wider as we derive it under as weak conditions as possible.
Therefore, in our derivation of the FGCB, we focus on the implicit-battery operations on $\mathcal{H}_{\mathrm{Baths}}\otimes\mathcal{H}_C$ which satisfy appropriate necessary conditions for being extended to an explicit-battery operation.
We consider a concrete explicit formulation in Sec.~\ref{sub_exp_bat}
to construct the operation to achieve FGCB.

One way of the implicit-battery formulation is to restrict the operations to be unitary \cite{Guryanova:2016aa}.
%However, it is quite restrictive since the reduced operations on $\mathcal{H}_{\mathrm{Baths}}\otimes\mathcal{H}_C$ are CPTP-maps in general, where state transitions are not guaranteed to be described by some unitary operation.
%The reduced dynamics is not even guaranteed to be described by energy non-conserving unitary
However, this restriction does not work because the state transitions are not guaranteed to be described by some unitary operation due to the interaction with the battery \cite{PhysRevA.95.032132}.
Indeed, the reduced operations on $\mathcal{H}_{\mathrm{Baths}} \otimes \mathcal{H}_C$ are written as completely positive and trace preserving (CPTP)-maps in general.
Thus, we impose the unitalness $\Gamma(I_{\mathrm{Baths},C})=I_{\mathrm{Baths},C}$ on implicit-battery operations $\Gamma$ on $\mathcal{H}_{\mathrm{Baths}}\otimes\mathcal{H}_C$, which is equivalent to non-decreasing of von Neumann entropy.
Here, $I_{\mathrm{Baths},C}$ is the identity of $\mathcal{H}_{\mathrm{Baths}}\otimes \mathcal{H}_C$.
The unitalness is reasonable as a necessary condition
because not only non-decreasing of von Neumann entropy corresponds to adiabaticity in the macroscopic thermodynamics 
but also the unitalness is
actually derived from another reasonable condition on explicit-battery formulations.
%, where $I_{\mathrm{Baths},C}$ is the identity of $\mathcal{H}_{\mathrm{Baths}}\otimes \mathcal{H}_C$.
%It is reasonable since it , which corresponds to adiabaticity in the standard thermodynamics.
We impose that the global unitary operations in the explicit-battery formulation commute with all the translation operators on the battery (Sec.~\ref{sub_exp_bat}) as in \cite{Guryanova:2016aa,Skrzypczyk:2014aa,PhysRevA.95.032132}.
As a natural situation, we consider the case where we cannot control the initial state on the battery and observe only the translation of the battery \cite{Skrzypczyk:2014aa}.
In order that the generalized heat engine works properly, we need such translational symmetry for the battery.
%%%check32
%Such translational symmetry of the battery reflects the fact that only changes in the quantities of the battery are important, similarly to \cite{Skrzypczyk:2014aa}.
In fact, the translational symmetry of the battery implies the unitalness of the reduced dynamics \cite{Guryanova:2016aa,Skrzypczyk:2014aa,PhysRevA.95.032132}.
Hence, the unitalness is believed to be necessary in consideration of the performance of generalized heat engines.
The cyclicity $\tr_{\mathcal{H}_{\mathrm{Baths}}}\Gamma(\tau^{(\lambda)}_{\bm{\theta}}\otimes\rho_C)=\rho_C$ of the working body $\mathcal{H}_C$ is also required.
In this way, FGCB is applicable whenever the 'explicit' dynamics reduces to a unital channel on $\mathcal{H}_{\mathrm{Baths}}\otimes\mathcal{H}_C$ with the cyclicity.
Note that the cyclicity can depend on the initial state $\rho_C$ of $\mathcal{H}_C$, so that $\rho_C$ can be used as a catalyst to enlarge the class of possible operations on $\mathcal{H}_{\mathrm{Baths}}$.
%Thus, we derive the fine-grained generalized Carnot bound (\ref{FGCB}) under the `implicit' formulation
%in the interest of keeping generality.
%We restrict the implicit operation $\Gamma$ on $\mathcal{H}_{\mathrm{Baths}}\otimes\mathcal{H}_C$ to the unital CPTP map $\Gamma(I_{\mathrm{Baths},C})=I_{\mathrm{Baths},C}$, where $I_{\mathrm{Baths},C}$ is the identity of $\mathcal{H}_{\mathrm{Baths}}\otimes \mathcal{H}_C$.

In summary, we employ the following operations in the implicit-battery formulation here:
\begin{definition}[Operations in the implicit-battery formulation]\label{def_implicit}
 Allowed operations in the implicit-battery formulation are
 CPTP maps $\Gamma$ on $\mathcal{H}_{\mathrm{Baths}}\otimes\mathcal{H}_C$ which satisfies the following:
  \begin{enumerate}
  \renewcommand{\labelenumi}{A\arabic{enumi}.}
  \item Unitalness:\label{it_ic-1}
	\begin{align}
	 \Gamma(I_{\mathrm{Baths},C})=I_{\mathrm{Baths},C}.\label{ic-1}
	\end{align}
  \item Cyclicity of the engine:\label{it_ic-2}
	\begin{align}
	 \tr_{\mathcal{H}_{\mathrm{Baths}}}\Gamma(\tau^{(\lambda)}_{\bm{\theta}}\otimes\rho_C)=\rho_C.\label{ic-2}
	\end{align}
 \end{enumerate}
\end{definition}

\subsection{Second law and the generalized Carnot bound in thermodynamic limit}\label{sec_GCB}
Since operations are given as unital CPTP maps on $\mathcal{H}_{\mathrm{Baths}}\otimes\mathcal{H}_C$ with the cyclicity,
we have
\begin{align}
 S(\rho_{\mathrm{Baths}}')+S(\rho_C)&\geq S(\Gamma(\tau^{(\lambda)}_{\bm{\theta}}\otimes\rho_C))\nonumber\\
 &\geq S(\tau^{(\lambda)}_{\bm{\theta}}\otimes\rho_C)
 =S(\tau^{(\lambda)}_{\bm{\theta}})+S(\rho_C)
\end{align}
from the subadditivity of the von Neumann entropy.
Therefore, the von Neumann entropy $S(\rho_{\mathrm{Baths}}')$ of
the final state $\rho_{\mathrm{Baths}}':=\tr_{\mathcal{H}_C} \Gamma(\tau_{\bm{\theta}_0}^{\lambda}\otimes\rho_C)$ of the bath system satisfies
\begin{align}
 S(\rho_{\mathrm{Baths}}')\geq S(\tau_{\bm{\theta}_0}^{\lambda}).\label{entropy_inc1}
\end{align}
%entropy non-decreasing property $S(\rho)\geq S(\tau^{(\lambda)}_{\beta_0})$ implies that
Thus, the relation $\Delta S:=S(\rho_{\mathrm{Baths}}')- S(\tau_{\bm{\theta}_0}^{\lambda})=\sum_{i=1}^{2}(\beta_{i}\Delta A_{i,\lambda}+\gamma_{i}\Delta B_{i,\lambda})-D(\rho_{\mathrm{Baths}}'\| \tau^{(\lambda)}_{\bm{\theta}_0})$ yields
the following second law of thermodynamics \cite{Guryanova:2016aa}:
\begin{align}
 \sum_{i=1}^{2}(\beta_{i}\Delta A_{i,\lambda}+\gamma_{i}\Delta B_{i,\lambda})\geq D(\rho_{\mathrm{Baths}}'\| \tau^{(\lambda)}_{\bm{\theta}_0}),\label{2nd_law}
\end{align}
where $\Delta A_{i,\lambda}:=\tr A_{i,\lambda}(\rho_{\mathrm{Baths}}'-\tau_{\bm{\theta}_0}^{(\lambda)})$ and $\Delta B_{i,\lambda}:=\tr B_{i,\lambda}(\rho_{\mathrm{Baths}}'-\tau_{\bm{\theta}_0}^{(\lambda)})$ are the amounts of difference of $A_{i,\lambda}$ and $B_{i,\lambda}$ respectively, and $D(\rho\| \sigma):= \tr \rho(\log\rho-\log \sigma)$ is the relative entropy between states $\rho$ and $\sigma$.

As pointed out in \cite{Guryanova:2016aa}, the relation (\ref{2nd_law}) implies the trade-off relation between the amounts of extraction of quantities, instead of a constraint for each single quantity.
Now, as a natural extension of the formulation of the ordinary Carnot bound, we
formulate the generalized Carnot bound (GCB) for the extraction of Quantity $A$, without loss of generality.
%Then, let $-\Delta A_{2,\lambda}$ and $-\Delta B_{i,\lambda}$ $(i=1,2)$ be given as $\Delta Q_A$ and $\Delta Q_{B,i}$ respectively.
We call $\Delta Q_{A,i}:= -\Delta A_{i,\lambda}$ and $\Delta Q_{B,i}:=-\Delta B_{i,\lambda}$ the generalized heat.
The extraction of $A$ is defined as $\Delta W_A :=\Delta Q_{A,1} + \Delta Q_{A,2}$ in the implicit-battery formulation.
This definition is based on the conservation of the average value of the quantity.
Though there are the strict and average conservation laws on the whole system in the explicit-battery formulation,
the conservation of the average values is satisfied for both cases as mentioned in Sec.~\ref{Sub_setup}.
%%check11
Thus, both the strict and average conservation laws meet this definition of the work extraction in the implicit-battery formulation.
The sum
$\Delta B_{1,\lambda}+\Delta B_{2,\lambda}$ of the differences in quantity $B$ does not have to vanish.
The amount $\Delta Q_{B,1}+\Delta Q_{B,2}$ is stored as the gain or lose of the average value of quantity $B$ of the battery in the same way, which may be regarded as the extraction of the other quantity or a `buffer' to extract Quantity $A$.
%Especially, we consider the strict conservation for commutative quantities

%=-\Delta A_{1,\lambda}+\Delta Q_A$.
Then, the relations (\ref{2nd_law}) and $D(\rho_{\mathrm{Baths}}'\| \tau^{(\lambda)}_{\bm{\theta}_0})\geq 0$ imply the following GCB:
\begin{align}
 \Delta W_A\leq \left(1-\frac{\beta_{2}}{\beta_{1}}\right)\Delta Q_{A,2}
 -\sum_{i=1}^{2}\frac{\gamma_{i}}{\beta_{1}}\Delta Q_{B,i}\label{GCB}
\end{align}
in response to the generalized heat $\Delta Q_{A,2}$ and $\Delta Q_{B,i}$,
where we set $\beta_{1}>0$.
In the following, we just focus on $\beta_{1}>0$ regime.
When $\beta_{1}<0$ is true, the opposite inequality holds.
As with the ordinary Carnot bound, this GCB is given only by generalized inverse temperatures.
%Since $D(\rho_{\mathrm{Baths}}'\| \tau^{(\lambda)}_{\bm{\theta}_0})$ corresponds to the difference in the free entropy, it can converge to $0$ in the thermodynamic limit, hence the equality in (\ref{GCB}) can be achieved in the limit, as is shown in \cite{Guryanova:2016aa}.
As mentioned in Remark \ref{remark1}, this bound does not include $\Delta Q_{A,1}$ since it is rather constrained if the other generalized heats are given.

The equality in (\ref{GCB}) is achieved if and only if $\Delta S$ and $D(\rho_{\mathrm{Baths}}'\| \tau^{(\lambda)}_{\bm{\theta}_0})$ vanish simultaneously.
In the thermodynamic limit with i.i.d.~baths, an achievable protocol was shown for commutative quantities and non-commutative quantities \cite{Guryanova:2016aa}.
However, this is possible only in the thermodynamic limit.
When finite-size effects are taken into account, $D(\rho_{\mathrm{Baths}}'\| \tau^{(\lambda)}_{\bm{\theta}_0})$ cannot vanish.
Hence, with finite size baths, GCB is never achieved.
To derive a tight bound with finite-size effects,
we have to consider how we can make $D(\rho_{\mathrm{Baths}}'\| \tau^{(\lambda)}_{\bm{\theta}_0})$ small under the scale $\lambda$ as in \cite{PhysRevE.96.012128}.
We derive the fine-grained GCB in the next subsection.
%Then, we derive tight constraint in asymptotic sense for scaling.
%Mathematically rigorous verification of the statements in this section is in Appendix \ref{}.

\subsection{Fine-grained generalized Carnot bound}\label{sec_FGCB}
We fix the generalized heat $\Delta Q_{A,2,\lambda}$ of Quantity $A$ from Bath $2$ and $\Delta Q_{B,i,\lambda}$ of quantity $B$ from Bath $i$ taking their scale dependence in account.
We focus on the regime with
\begin{align}
 \Delta Q_{A,2,\lambda}=o(\lambda), \Delta Q_{B,i,\lambda}=o(\lambda)\; (i=1,2).\label{heat_scale}
\end{align}
These relations reflect the fact that the system $\mathcal{H}_{\mathrm{Baths}}$ is used as just like baths
in the sense that the `final inverse temperature' converges to the initial one in thermodynamic limit,
because
the order of the generalized heats
$\Delta Q_{A,2,\lambda}$, $\Delta Q_{B,i,\lambda}$
become smaller than
the order $\order{\lambda}$ of $\pdv{\eta_{\lambda,2}}{\beta_2} (\bm{\theta}_0)$ and $\pdv{\eta_{\lambda,i+2}}{\gamma_i} (\bm{\theta}_0)$, which correspond to the respective `heat capacities'.
Note that since the resultant state is not necessarily a generalized thermal state,
the final inverse temperature is not necessarily well-defined.
%so that it does not necessarily have generalized inverse temperature.
For a generic state $\rho$, we assign 
the `effective' inverse temperature $\tilde{\bm{\theta}}_{\lambda}(\rho)$,
which is defined as the generalized inverse temperature of the thermal state
%Here, we consider the final inverse temperature as the `effective' inverse temperature $\tilde{\bm{\theta}}_{\lambda}(\rho)$ defined for an arbitrary state $\rho$ as the generalized inverse temperature of the thermal state
$\tau_{\tilde{\bm{\theta}}_{\lambda}(\rho)}^{(\lambda)}$ sharing the same expectation values:
\begin{align}
 \tr A_{i,\lambda}\tau_{\tilde{\bm{\theta}}_{\lambda}(\rho)}^{(\lambda)}
 =&\tr A_{i,\lambda}\rho\\
 \tr B_{i,\lambda}\tau_{\tilde{\bm{\theta}}_{\lambda}(\rho)}^{(\lambda)}
 =&\tr B_{i,\lambda}\rho \quad (i=1,2).
\end{align}
Oppositely, if \eqref{heat_scale} does not hold, the generalized inverse temperatures of the baths change even in thermodynamic limit, so that GCB is not achievable.
This is quite different from what the bath is in thermodynamics.

%More precisely,
%for any states not necessarily thermal, the effective inverse temperature is assigned %
%Especially,
%any final state with 
%converges to ${\bm \theta}_0$ in this regime.
Then, the following fine-grained GCB holds:
\begin{theorem}[Fine-grained generalized Carnot bound (FGCB)]\label{Thm_FGCB}
 Let the generalized heats $\Delta Q_{A,2,\lambda}$,$\Delta Q_{B,1,\lambda}$ and $ \Delta Q_{B,2,\lambda}$ satisfy (\ref{heat_scale}).
 Then, we have
 \begin{align}
 \Delta W_A
 \leq & \left(1-\frac{\beta_{2}}{\beta_{1}}\right)\Delta Q_{A,2,\lambda}
 -\sum_{i=1}^{2}\frac{\gamma_{i}}{\beta_{1}}\Delta Q_{B,i,\lambda}\nonumber\\
  &
  -C_{AA}\frac{\Delta Q_{A,2,\lambda}^2}{\lambda}
  -\sum_{i=1}^{2}C_{AB}^{i}\frac{\Delta Q_{A,2,\lambda}\Delta Q_{B,i,\lambda}}{\lambda}\nonumber\\
  &-\sum_{i,j=1}^2C_{BB}^{i,j}\frac{\Delta Q_{B,i,\lambda}\Delta Q_{B,j,\lambda}}{\lambda}\nonumber\\
  &+o\left(\frac{\|\vb{Q}_{\lambda}\|^2}{\lambda}\right)\nonumber\\
  =:&\Delta W^{\mathrm{opt}}_{A,\lambda}(\vb{Q}_{\lambda})+o\left(\frac{\|\vb{Q}_{\lambda}\|^2}{\lambda}\right),\label{FGCB}
\end{align}
 where we define $\vb{Q}_{\lambda}:=(\Delta Q_{A,2,\lambda},\Delta Q_{B,1,\lambda}, \Delta Q_{B,2,\lambda})$ and its norm $\|\vb{Q}_{\lambda}\|^2:=\beta_0^2\Delta Q_{A,2,\lambda}^2+\gamma_0^2\Delta Q_{B,1,\lambda}^2+\gamma_0^2\Delta Q_{B,2,\lambda}^2$ with the unit generalized inverse temperatures $\beta_0$ and $\gamma_0$ to adjust the physical dimension, and the second order coefficients are given as follows:
  \begin{align}
   C_{AA}
   =&\frac{1}{2\beta_1}\left[
   g^{22}(\bm{\theta}_0)
   +
   \left(\frac{\beta_2}{\beta_1}\right)^2 g^{11}(\bm{\theta}_0)
   -2\frac{\beta_{2}}{\beta_{1}}g^{12}(\bm{\theta}_0)
 \right],\label{Caa}\\
   C_{AB}^{i}
   =&\frac{1}{\beta_1}\left[
   g^{2(i+2)}(\bm{\theta}_0)
   +\frac{\beta_{2}\gamma_{i}}{\beta_{1}^2}g^{11}(\bm{\theta}_0)
   \right.\nonumber\\
   &\left.-\frac{\beta_{2}}{\beta_{1}}g^{1(i+2)}(\bm{\theta}_0)
   -\frac{\gamma_{i}}{\beta_{1}}g^{12}(\bm{\theta}_0)
 \right],\label{Cab}\\
   C_{BB}^{ij}
   =&\frac{1}{2\beta_1}\left[
   g^{(i+2)(j+2)}(\bm{\theta}_0)
   +\frac{\gamma_{i}\gamma_{j}}{\beta_{1}^2}g^{11}(\bm{\theta}_0)
   \right.\nonumber\\
   &\left.
   -\frac{\gamma_{j}}{\beta_{1}}g^{1(i+2)}(\bm{\theta}_0)
   -\frac{\gamma_{i}}{\beta_{1}}g^{1(j+2)}(\bm{\theta}_0)
  \right],\label{Cbb}
 \end{align}
% \begin{align}
%   C_{AA}
%  =&\frac{1}{2\beta_1}\left[
%   \left(\frac{\beta_{2}}{\beta_{1}}\right)^2 g^{11}(\bm{\theta}_0)
%   +g^{22}(\bm{\theta}_0)
%   -2\frac{\beta_{2}}{\beta_{1}}g^{12}(\bm{\theta}_0)
% \right],\label{Caa}\\
%   C_{AB}^{i}
%  =&\frac{1}{2\beta_1}\left[
%   \frac{\beta_{2}}{\beta_1}\frac{\gamma_{i}}{\beta_{1}}g^{11}(\bm{\theta}_0)
%   +\frac{g^{2(i+2)}(\bm{\theta}_0)}{\beta_{1}}\right.\nonumber\\
%   &\left.-\frac{g^{1(i+2)}(\bm{\theta}_0)\beta_{2}}{\beta_{1}^2}
%   -\frac{g^{12}(\bm{\theta}_0)\gamma_{i}}{\beta_{1}^2}
% \right],\label{Cab}\\
%   C_{BB}^{ij}
%  =&\frac{1}{2}\left[
%   \frac{g^{11}(\bm{\theta}_0)\gamma_{i}\gamma_{j}}{(\beta_{1})^3}
%   +\frac{g^{(i+2)(j+2)}(\bm{\theta}_0)}{\beta_{1}}\right.\nonumber\\
%   &\left.-\frac{g^{1(i+2)}(\bm{\theta}_0)\gamma_{j}}{\beta_{1}^2}
%   -\frac{g^{1(j+2)}(\bm{\theta}_0)\gamma_{i}}{\beta_{1}^2}
%  \right],\label{Cbb}
% \end{align}
 where $(g^{ij}(\bm{\theta}_0))_{ij}$ is the inverse matrix of the asymptotic density of the canonical correlations $(g_{ij}(\bm{\theta}_0))_{ij}$ defined by (\ref{canonical_cor}), (\ref{asymp_correlation})
 \footnote{Physical dimension of \eqref{FGCB} is consistent since each generalized inverse temperature has inverse dimension of its conjugate quantity.}.
\end{theorem}
%In the following until Sec. \ref{Sub_achievability}, we employ $\beta_0=\gamma_0=1$ unit, for simplicity.
The quantity $\Delta W^{\mathrm{opt}}_{A,\lambda}(\vb{Q}_{\lambda})$ gives an upper bound on the extraction of Quantity $A$ including the finite-size effects of $\frac{\|\vb{Q}_{\lambda}\|^2}{\lambda}$-order with the generalized heat $\vb{Q}_{\lambda}$.
Note that
$\Delta W_A$ may not be proper `work-like' transfer of Quantity $A$, but rather possibly includes `heat-like' transfer.
Nevertheless, since any proper work-like transfer, namely $\Delta W_A'$, is included in the total transfer $\Delta W_A$,
we have $\Delta W_A'\leq \Delta W_A$.
Thus, FGCB \eqref{FGCB} is still true upper bound even for proper work-like extraction of Quantity $A$.
The achievability of the bound is more delicate in this sense.
In Sec.~\ref{Sub_achievability}, including the battery system explicitly, we carefully construct an operation to achieve FGCB by avoiding hidden extra reservoir inside the battery.

The three terms of the second order $\frac{\|\vb{Q}_{\lambda}\|^2}{\lambda}$ in (\ref{FGCB}) express the finite-size effects, which are indeed always negative, so that FGCB does not exceed GCB.
Remember that $g_{ij}(\bm{\theta}_0)$ is the asymptotic density of the Fisher information, and the elements of Fisher information are the canonical correlations of the baths' quantities reflecting their non-commutativity.
Thus, the second order terms reflect the effects of the fluctuation and correlation of the baths through $g^{ij}(\bm{\theta}_0)$.
Therefore, the Fisher information,
which is finer structure than just the temperatures of the baths is relevant in FGCB, differently from GCB (\ref{GCB}) and the ordinary Carnot bound.
Especially, correlations between the different baths are also taken into account.
This result implies that we should consider the correlations of the conserved quantities of the baths to design better engine with finite-size baths.

Let us examine how to obtain better performance of generalized heat engines through interpreting the coefficients of the finite-size effect.
FGCB is a direct consequence of the entropy increasing law due to the unitalness of the dynamics, as with GCB.
As will be shown in the proof of Theorem \ref{Thm_FGCB} in later, for FGCB the second order term
\begin{align}
 -\frac{1}{2}\sum_{i,j=1}^4 g^{ij}(\bm{\theta}_0)\frac{\Delta X_{i,\lambda}\Delta X_{j,\lambda}}{\lambda}\label{ent22}
\end{align}
 in the entropy change is taken into account, where $\Delta X_{i,\lambda}$ is the variation in the expectation value of $X_{i,\lambda}$.
The optimal performance is given when the entropy change vanishes.
Since this negative definite term \eqref{ent22} should be canceled, degradation of the optimal performance is caused.
That is, the optimal performance approaches GCB as the baths get closer to the ideal baths in the sense that their state is unchanged through the operation.
Conversely, variation of the state of the baths causes degradation of the optimal performance.
In fact, the second order term \eqref{ent22} is expressed by the variation $\Delta \bm{\theta}$ in the inverse temperature as
 \begin{align}
  &-\frac{1}{2}\sum_{i,j=1}^4 g^{ij}(\bm{\theta}_0)\frac{\Delta X_{i,\lambda}\Delta X_{j,\lambda}}{\lambda}\nonumber\\
 =&-\frac{1}{2}\sum_{i=1}^4 \Delta \theta^i \Delta X_{i,\lambda} + o\left(\frac{\Delta X_{i,\lambda}\Delta X_{j,\lambda}}{\lambda}\right)
 \end{align}
 because
 \begin{align}
  \Delta \theta^i = \sum_{j=1}^{4} g^{ij}(\bm{\theta}_0)\frac{\Delta X_{j,\lambda}}{\lambda}
  + o\left(\frac{\Delta X_{j,\lambda}}{\lambda}\right).
 \end{align}
 Hence, the finite-size effect in FGCB reflects the linear response of the inverse temperature to the variation in $X_{j,\lambda}$ up to the order $\Delta X_{j,\lambda}/\lambda$.
 Indeed, the coefficients (\ref{Caa})-(\ref{Cbb}) are given in terms of the coefficient matrix $g^{ij}(\bm{\theta}_0)$.
The coefficients (\ref{Caa})-(\ref{Cbb})
give the effect of the response of the inverse temperature on the optimal performance in a concrete form.
From the above perspective, the smaller the response becomes, the better performance is achieved.
Note that $g^{ij}$ depends not only on the inverse temperature but also other parameters $\vb{x}$ in general (see examples in Sec.~\ref{sec_example}) as $g^{ij}(\bm{\theta}; \vb{x})$.
Especially, when $(g^{ij}(\bm{\theta}^{(1)}_0;\vb{x}_1))\leq (g^{ij}(\bm{\theta}^{(2)}_0;\vb{x}_2))$ holds as the matrix inequality, $(g^{ij}(\bm{\theta}^{(1)}_0;\vb{x}_1))$ gives better performance.

% not only temperatures but also the Fisher information
\begin{proof}[Proof of Theorem \ref{Thm_FGCB}.]
The von Neumann entropy of the thermal states can be seen as a function of the expectation values by the following Legendre transformation:
\begin{align}
 S_{\lambda}(\bm{\eta})
 :=\min_{\tilde{\bm{\theta}}}\left[\sum_{i=1}^{4}\tilde{\theta}^i\eta_{i}+\phi_{\lambda}(\tilde{\bm{\theta}})\right].
\end{align}
 In this expression, $S_\lambda$ is a function of the variable $\bm{\eta}=(\eta_1,\eta_2,\eta_3,\eta_4)$.
 For an inverse temperature $\bm{\theta}$, the function $S_{\lambda}$ is actually related with the von Neumann entropy of the thermal state $\tau^{(\lambda)}_{\bm{\theta}}$ by the relation
 \begin{align}
  S_{\lambda}(\bm{\eta}_{\lambda}(\bm{\theta}))=S(\tau^{(\lambda)}_{\bm{\theta}}),
 \end{align}
 where $\bm{\eta}_{\lambda}(\bm{\theta})$ is the dual coordinate of $\bm{\theta}$ composed of the expectation values defined by \eqref{dual_coord}.
 In this sense, the variable $\bm{\eta}$ expresses the expectation values.
 %check
 If the final state $\rho_{\mathrm{Baths}}'$ of the baths satisfies
 $\Delta A_{i,\lambda}=\tr A_{i,\lambda}(\rho_{\mathrm{Baths}}'-\tau_{\bm{\theta}_0}^{(\lambda)})=o(\lambda)$ and $\Delta B_{i,\lambda}=\tr B_{i,\lambda}(\rho_{\mathrm{Baths}}'-\tau_{\bm{\theta}_0}^{(\lambda)})=o(\lambda)$,
 the effective inverse temperature $\bm{\theta}':=\tilde{\bm{\theta}}_{\lambda} (\rho_{\mathrm{Baths}}')$ of $\rho_{\mathrm{Baths}}'$ exists for sufficiently large $\lambda$.
 %revise the expression so that effective inverse temprature is used effectively
 That is, $\bm{\theta}'$ satisfies
 \begin{align}
  \Delta \eta_{\lambda,i}&:= \eta_{\lambda,i}(\bm{\theta}')-\eta_{\lambda,i}(\bm{\theta}_0)=\Delta A_{i,\lambda}, \nonumber\\
  \Delta \eta_{\lambda,i+2}&:=\eta_{\lambda,i+2}(\bm{\theta}')-\eta_{\lambda,i+2}(\bm{\theta}_0)=\Delta B_{i,\lambda}\quad (i=1,2).
 \end{align}
 Because the thermal state has the maximum entropy among the states with the same expectation values \cite{PhysRev.106.620,PhysRev.108.171,Guryanova:2016aa,Yunger-Halpern:2016aa}:
  \begin{align}
  S(\tau^{(\lambda)}_{\bm{\theta}})=\max_{\rho}\{S(\rho)|\tilde{\bm{\theta}}_{\lambda}(\rho)=\bm{\theta}\},
  \end{align}
the Taylor expansion of $S_{\lambda}$ around $\bm{\eta}_{\lambda}(\bm{\theta}_0)$ yields
\begin{align}
 &S(\tau^{(\lambda)}_{\bm{\theta}'})-S(\tau^{(\lambda)}_{\bm{\theta}_0})\nonumber\\
 =&S_{\lambda}(\bm{\eta}_{\lambda}(\bm{\theta}'))-S_{\lambda}(\bm{\eta}_{\lambda}(\bm{\theta}_0))\nonumber\\
 =&\sum_{i=1}^{4}\theta_0^{i}\Delta \eta_{\lambda,i}-\frac{1}{2}\sum_{i,j=1}^4 J_{\lambda}^{ij}(\bm{\theta}_0)\Delta \eta_{\lambda,i}\Delta \eta_{\lambda,j}
 +\order{\frac{\|\bm{\Delta \eta}_{\lambda}\|^3}{\lambda^2}}\nonumber\\
 =&\sum_{i=1}^{4}\theta_0^{i}\Delta \eta_{\lambda,i}
 -\frac{1}{2}\sum_{i,j=1}^4 g^{ij}(\bm{\theta}_0)\frac{\Delta \eta_{\lambda,i}\Delta \eta_{\lambda,j}}{\lambda}
 +o\left(\frac{\|\bm{\Delta \eta}_{\lambda}\|^2}{\lambda}\right)\nonumber\\
 \geq &
 S(\rho_{\mathrm{Baths}}')-S(\tau^{(\lambda)}_{\bm{\theta}_0})\geq 0,
 \label{region}
\end{align}
 %by using Assumption \ref{assumption}
 where
 %, where $\|\Delta A_{\lambda}\|$ is the $2$-norm $\sqrt{\sum_{i=0}^{m}(\Delta A_{i,\lambda})^2}$ of the vector $\Delta A_{\lambda}:=(\Delta A_{i,\lambda})_{i=0}^{m}$.
 $J_{\lambda}^{ij}(\bm{\theta}_0)$ is the $(i,j)$-element of the inverse matrices of $(J_{\lambda,ij}(\bm{\theta}_0))$,
 and $\|\bm{\Delta \eta}_{\lambda}\|=\sqrt{\sum_{i=1}^{4}(\Delta \eta_{\lambda,i})^2}$.
 The last inequality follows from the increasing of the entropy \eqref{entropy_inc1}.
 We used the relation $\pdv{S_{\lambda}}{\eta_i} (\bm{\eta}_{\lambda}(\bm{\theta}))=\theta^i$ to evaluate the coefficients of the Taylor expansion.
 We carried out the estimation of the third order derivatives based on Assumption \ref{assumption}
 in the derivation of the order of the residual term
 in the second equality.
The third equality follows from the estimation $J_{\lambda}^{ij}(\bm{\theta}_0)=\lambda^{-1}g^{ij}(\bm{\theta}_0)+o(\lambda^{-1})$, which holds uniformly on the neighborhood of the initial temperature.
%This is the general constraint on the generalized work extraction with the size effect for $\lambda$, which is valid regardless of commutativity of the observables.
%tight second law
%The second and later terms are contribution from $\Delta F^{(\lambda)}_{\beta_0}$.
 %Furthermore, we obtain the generalization of the Carnot efficiency bound with size effect by solving (\ref{region}) asymptotically with respect to $\Delta A_{0,\lambda}$ without loss of generality as follows:
 When $\Delta \eta_{\lambda,2}=\Delta A_{2,\lambda}=-\Delta Q_{A,2,\lambda}, \Delta \eta_{\lambda, i+2}=\Delta B_{i,\lambda}=-\Delta Q_{B,i,\lambda}$ $(i=1,2)$ are given,
 by solving the equantion $S_{\lambda}(\bm{\eta}_{\lambda}(\bm{\theta}'))-S_{\lambda}(\bm{\eta}_{\lambda}(\bm{\theta}_0))=0$ asymptotically with respect to $\Delta \eta_{\lambda,1}$, we obtain
 an upper bound for the possible value of $\Delta A_{1,\lambda}=\Delta \eta_{\lambda,1}$ as:
\begin{align}
 &-\beta_{1}\Delta A_{1,\lambda}\nonumber\\
 \leq & -\beta_{2}\Delta Q_{A,2,\lambda}
 -\sum_{i=1}^{2}\gamma_{i}\Delta Q_{B,i,\lambda}\nonumber\\
  &
  -\beta_{1}C_{AA}\frac{\Delta Q_{A,2,\lambda}^2}{\lambda}
  -\beta_{1}\sum_{i=1}^{2}C_{AB}^{i}\frac{\Delta Q_{A,2,\lambda}\Delta Q_{B,i,\lambda}}{\lambda}\nonumber\\
  &-\beta_{1}\sum_{i,j=1}^2C_{BB}^{i,j}\frac{\Delta Q_{B,i,\lambda}\Delta Q_{B,j,\lambda}}{\lambda}+o\left(\frac{\|\vb{Q}_{\lambda}\|^2}{\lambda}\right).
 \label{carnot}
\end{align}
% for \eqref{region} to be satisfied in the asymptotic sense.
%valid for the regime, lower or upper, solve perturbatively
%where $P$ is the projection onto the later than $0$-th elements such that $P\Delta A_{\lambda}=(\Delta A_{1,\lambda},\Delta A_{2,\lambda},\dots,\Delta A_{m,\lambda})$.
%Note that (\ref{carnot}) gives the upper bound on $-\Delta A_{\lambda,0}$ for $\beta_0^0\geq 0$, while it gives the lower bound for $\beta_0^0 < 0$.
%To see that (\ref{carnot}) is essentially generalized Carnot efficiency bound, let $A_{0,\lambda},\dots,A_{m_0,\lambda}$ be included in the same kind of the observables. Let $\beta_0^0\geq 0$, while $\beta_0^0<0$ case is similar except that the lower bound is derived. We focus on what amount of the work extraction $\Delta W=-\sum_{i=0}^{m_0}\Delta A_{i,\lambda}$ of this kind of quantity
%is possible when $\Delta A_{i,\lambda}\;(i=1,2,\dots,m)$ are given.
 Then,
 substituting (\ref{carnot}) to $\Delta W_A=-\Delta A_{1,\lambda}+\Delta Q_{A,2,\lambda}$,
 we obtain (\ref{FGCB}).
\end{proof}
%The first two terms correspond to the generalized Carnot type bound when thermodynamic limit is taken.
%Indeed, for the original bound, the upper bound on the work is given as $\Delta W \leq (1-\frac{\beta_H}{\beta_L})\Delta Q_H$ for given heat $\Delta Q_H$ from the hot bath.
%The later terms of order $O(\frac{\|P\Delta A_{\lambda}\|^2}{\lambda})$ is the first correction term reflecting finite-size effect.
%Remarkably, all these terms are described in terms of not only temperatures but also inverse of the canonical correlations, which reflect the effect of fluctuation.
%Thus, different from the thermodynamic limit regime, the structure of correlation and fluctuation of the baths other than temperatures is relevant for finite-size effect.
%Similar results were reported for i.i.d. case \cite{}, and for extremely small and large baths \cite{}.

%:=(\Delta Q_{A,2,\lambda}, \Delta Q_{B,\lambda,1}, \Delta Q_{B,\lambda, 2})
We define the ideal final inverse temperature $\bm{\theta}_\lambda=(\beta_{\lambda 1}, \beta_{\lambda 2}, \gamma_{\lambda 1}, \gamma_{\lambda 2})$ associated with
a vector $\vb{Q}_{\lambda}$ of the generalized heat as
\begin{align}
 S(\tau^{(\lambda)}_{\bm{\theta}_{\lambda}})&=S(\tau^{(\lambda)}_{\bm{\theta}_0})\label{tl_1}\\
 \beta_{\lambda 1}\beta_{1}&\geq 0\label{tl_2}\\
 \tr A_{2,\lambda}(\tau_{\bm{\theta}_0}^{(\lambda)}-\tau^{(\lambda)}_{\bm{\theta}_{\lambda}})&=\Delta Q_{A,2,\lambda}\label{tl_3}\\
 \tr B_{i,\lambda}(\tau_{\bm{\theta}_0}^{(\lambda)}-\tau^{(\lambda)}_{\bm{\theta}_{\lambda}})&=\Delta Q_{B,i,\lambda} \quad (i=1,2).\label{tl_4}
\end{align}
The equality in (\ref{FGCB}) is formally attained by the thermal state $\tau^{(\lambda)}_{\bm{\theta}_{\lambda}}$ at the ideal final inverse temperature $\bm{\theta}_\lambda=(\beta_{\lambda 1}, \beta_{\lambda 2}, \gamma_{\lambda 1}, \gamma_{\lambda 2})$ associated with
$\vb{Q}_{\lambda}$.
% $\eta_{\lambda,i}(\bm{\theta}_{\lambda})=\eta_{\lambda,i}(\bm{\theta}_0)+\Delta A^*_{i,\lambda}$.% for any given $\Delta A^*_{i,\lambda}$ $(i\neq 0)$ with $\Delta A^*_{i,\lambda}=o(\lambda)$.
%We denote also $\eta_{\lambda,0}(\beta_{\lambda})-\eta_{\lambda,0}(\bm{\theta}_0)$ by $\Delta A^*_{0,\lambda}$, which indeed satisfies
%\begin{align}
% &\Delta A^*_{0,\lambda}\nonumber\\
% =&-\frac{1}{\beta_0^0}\langle P\beta_0,P\Delta A_{\lambda}^*\rangle +\frac{1}{2\beta_0^0}\sum_{i,j=1}^m J_{\lambda}^{ij}(\bm{\theta}_0)\Delta A_{i,\lambda}^*\Delta A_{j,\lambda}^*
% +\frac{1}{2(\beta_0^0)^3}J^{00}_{\lambda}\langle P\beta_0,P\Delta A_{\lambda}^*\rangle^2
% \nonumber\\
% &+\frac{\langle P\beta_0,P\Delta A_{\lambda}^*\rangle}{(\beta_0^0)^2}\sum_{j=1}^{m}J^{0j}_{\lambda}\Delta A_{j,\lambda}^*
% +\order{\frac{\|P\Delta A_{\lambda}^*\|^3}{\lambda^2}}\nonumber\\ 
%=&O(\|P\Delta A_{\lambda}^*\|)=o(\lambda). 
%\end{align}
%check
However, this state $\tau^{(\lambda)}_{\bm{\theta}_{\lambda}}$ is not necessarily achievable from $\tau^{(\lambda)}_{\bm{\theta}_0}$ by operations.
In Sec.~\ref{Sub_achievability}, for commutative observables, we show that FGCB is achievable in the asymptotic sense by constructing the operation which maps $\tau^{(\lambda)}_{\bm{\theta}_0}$ close to $\tau^{(\lambda)}_{\bm{\theta}_{\lambda}}$
instead of exactly to $\tau^{(\lambda)}_{\bm{\theta}_{\lambda}}$.
%Thus, reflecting no cheating condition, we restrict $V$ to increasing the entropy of the baths, which is equivalent to that the local dynamics $\mathcal E:\tau_B\otimes\rho_C\mapsto\rho_{\mathrm{Baths}}'\otimes\rho_C$ is a unital CP-map.
%

\section{Achievability of FGCB by explicit construction of the protocol}\label{Sub_achievability}
%non-commutative case with average conservation may appear in appendix
In this section, we focus on the achievability of FGCB in `physical sense'.
That is, as mentioned in Sec.~\ref{sub_imp_exp},
%(it is not enough to construct an implicit-battery operation)
%to achieve the FGCB without loss of physical relevance
%because hidden heat-like transfer is possibly included in the extracted amount.
%Instead, we construct an operation with an explicit battery under appropriate conditions
%to show the achievability of FGCB.
it is not enough to construct an implicit-battery operation to achieve the FGCB even though it satisfies the unitalness
because the unitalness is only the necessarily condition for the existence of an operation in the explicit-formulation that has no hidden heat-like transfer in the extracted amount.

To verify that our protocol achieves FGCB, we assume a stronger extensivity than Assumption \ref{assumption}.
\begin{Ass}\label{assump2}
 The order of the deviation from the extensivity is sufficiently small so that there exists an $\alpha<\frac{1}{2}$ such that
  \begin{align}
   \phi_{\lambda}(\bm{\theta})&=\lambda\phi(\bm{\theta})+\order{\lambda^{\alpha}},\label{str_extensive_1}
  \end{align}
  \begin{align}
   &\left(\pdv{\theta^{i_1}}\right)^{l_1}\left(\pdv{\theta^{i_2}}\right)^{l_2}\left(\pdv{\theta^{i_3}}\right)^{l_3}\phi_{\lambda}(\bm{\theta})\nonumber\\
  =&
  \left(\pdv{\theta^{i_1}}\right)^{l_1}\left(\pdv{\theta^{i_2}}\right)^{l_2}\left(\pdv{\theta^{i_3}}\right)^{l_3}
  \lambda\phi(\bm{\theta})
  +\order{\lambda^{\alpha}}\label{str_extensive_2}
 \end{align}
 hold instead of \eqref{extensive_1} and \eqref{extensive_2}, where $0< l_1+l_2+l_3\leq 3$, $i_1,i_2,i_3\in\{1,2,3,4\}$.
 In addition, the matrix norms $\|A_{i,\lambda}\|, \|B_{i,\lambda}\|$ $(i=1,2)$ are of order $\order{\lambda}$:
 \begin{align}
  \|A_{i,\lambda}\|=\order{\lambda},\quad
  \|B_{i,\lambda}\|=\order{\lambda}.
  \label{mat_norm_0}
 \end{align}
\end{Ass}

At first, we just focus on the commutative quantities in Sec.~\ref{sub_cons_imp}-\ref{sub_ach_FGCB}.
We explicitly construct a protocol to achieve the equality in FGCB
in the asymptotic sense up to $o\left(\frac{\|\vb{Q}_{\lambda}\|^2}{\lambda}\right)$ under the strict conservation law.
The idea is to make the final state close to the thermal state with the ideal final inverse temperature $\bm{\theta}_{\lambda}$.
Finally in Sec. \ref{sub_noncomm}, we extend the construction to the case of non-commutative quantities under the average conservation law.
%, i.e. the maximal extraction of a quantity in the sense of the Carnot bound.

\subsection{Construction of the implicit-battery operation}\label{sub_cons_imp}
%To begin with,
%%%%%%%%%%%%%%%%%%%%%%%%%%%%%%%%%%%%%%%%%%%%%%%%%%%%%figure
\begin{figure}[!t]
\centering
\includegraphics[clip ,width=3.3in]{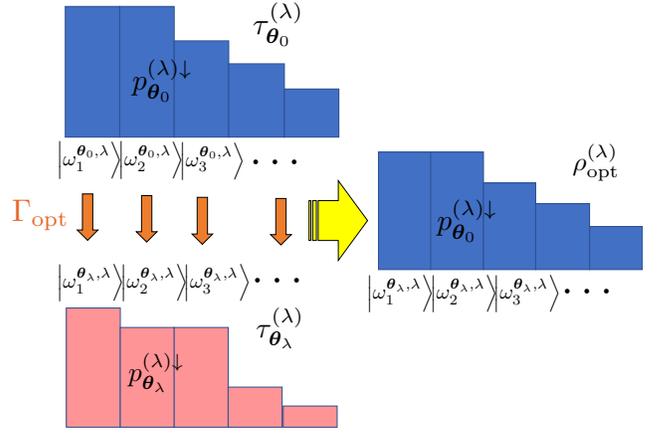}
 \caption{Schematic picture of the implicit operation $\Gamma_{\rm opt}$.
 An ordering of the eigenstates $\ket{\omega_1^{\vb*{\theta}_{0},\lambda}},\ket{\omega_2^{\vb*{\theta}_{0},\lambda}},\dots$ is such that
 $\ket{\omega_i^{\vb*{\theta}_{0},\lambda}}$ is mixed at the $i$-th largest probability
 $p^{(\lambda)}_{\bm{\theta}_0,\lambda}(i)$ in the initial state $\tau^{(\lambda)}_{\bm{\theta}_0,\lambda}$.
 Another ordering $\ket{\omega_1^{\vb*{\theta}_{\lambda},\lambda}},\ket{\omega_2^{\vb*{\theta}_{\lambda},\lambda}},\dots$ is such that
 $\ket{\omega_i^{\vb*{\theta}_{\lambda},\lambda}}$ is mixed at the $i$-th largest probability
 $p^{(\lambda)}_{\bm{\theta}_\lambda}(i)$ in the thermal state $\tau^{(\lambda)}_{\bm{\theta}_{\lambda}}$ at the the ideal final inverse temperature $\vb*{\theta}_{\lambda}$.
 $\Gamma_{\rm{opt}}$ maps each $\ket{\omega_i^{\vb*{\theta}_{0},\lambda}}$ to $\ket{\omega_i^{\vb*{\theta}_{\lambda},\lambda}}$.
 In the resultant state $\rho_{\rm opt}^{(\lambda)}$, $\ket{\omega_i^{\vb*{\theta}_{\lambda},\lambda}}$ is mixed at the probability $p_{\bm{\theta}_0}^{(\lambda)}(i)$ which was initially assigned to $\ket{\omega_i^{\vb*{\theta}_{0},\lambda}}$.}
\label{figure5}
\end{figure}
%%%%%%%%%%%%%%%%%%%%%%%%%%%%%%%%%%%%%%%%%%%%%%%%%%%%%%%%%f
To begin with, we construct an operation in the implicit-battery formulation to achieve the equality in FGCB for the case of commutative quantities.
For this purpose, we choose the simultaneous eigenstates $\ket{\omega}$ to diagonalize
$A_{1,\lambda},A_{2,\lambda}, B_{1,\lambda}, B_{2,\lambda}$,
so that
the respective eigenvalues $a_{i,\lambda}(\omega)$ and $b_{i,\lambda}(\omega)$ $(i=1,2)$ of $A_{i,\lambda}$ and $B_{i,\lambda}$
are labeled by $\omega$.
In our operation, we use the ordering of 
the simultaneous eigenstates $\ket{\omega}$
that depends on the inverse temperature coordinate $\bm{\theta}$ and 
the scale $\lambda$.
Given $\bm{\theta}$ and $\lambda$,
we diagonalize $ \tau_{\bm{\theta}}^{(\lambda)}$
as \begin{align}
 \tau_{\bm{\theta}}^{(\lambda)}
 =:\sum_{i\in\mathbb N_{d_{\lambda}}}
 p^{(\lambda)}_{\bm{\theta}}(i)\ketbra{\omega_i^{\vb*{\theta},\lambda}}{\omega_i^{\vb*{\theta},\lambda}}, \label{diag1}
   \end{align}
   where $\mathbb N_{d_{\lambda}}:=\{1,2,\dots,d_{\lambda}\}$.
In the equation \eqref{diag1},
we define the probability distribution $ p^{(\lambda)}_{\bm{\theta}}$
composed of the eigenvalues of $\tau_{\bm{\theta}}^{(\lambda)}$
in descending order
$p^{(\lambda)}_{\bm{\theta}}(1)\geq p^{(\lambda)}_{\bm{\theta}}(2)\geq \dots$
, and
%assign the positive integers $i$ to
accordingly label
the simultaneous eigenstates $\ket{\omega}$
by defining the state $\ket{\omega_i^{\vb*{\theta},\lambda}}$.
Although the ordering of the eigenstates is not unique because of the degeneracy, such multiplicity is totally irrelevant for our analysis.
Thus, it is sufficient to arbitrarily choose an ordering for the eigenstates with the same eigenvalues.
%We abbreviate $\ket{a_{1,\lambda}(\omega),a_{2,\lambda}(\omega), b_{1,\lambda}(\omega), b_{2,\lambda}(\omega)}$ as $\ket{\omega}$.
Given generalized heat amounts $\vb{Q}_{\lambda}=(\Delta Q_{A,2,\lambda},\Delta Q_{B,1,\lambda},\Delta Q_{B,2,\lambda})$,
we have defined the ideal final inverse temperature $\bm{\theta}_{\lambda}$ by the conditions (\ref{tl_1})-(\ref{tl_4}).
Since the respective $i$-th largest eigenvalues
$p^{(\lambda)}_{\bm{\theta}_0}(i)$
and
%corresponds the different eigenstate from that of
%has an ordering different from
$p^{(\lambda)}_{\bm{\theta}_{\lambda}}(i)$
of $\tau_{\bm{\theta}_0}^{(\lambda)}$ and $\tau_{\bm{\theta}_\lambda}^{(\lambda)}$
correspond to the different eigenstates from each other
in general,
the two states $\ket{\omega_i^{\vb*{\theta}_0,\lambda}}$ and $\ket{\omega_i^{\vb*{\theta}_{\lambda},\lambda}}$ are different.
Then,
we consider a unital CPTP map on $\mathcal{H}_{\mathrm{Baths}}$ which maps each eigenstate $\ket{\omega_i^{\vb*{\theta}_0,\lambda}}$ to $\ket{\omega_i^{\vb*{\theta}_{\lambda},\lambda}}$.
That is, we employ an operation $\Gamma_{\rm opt}^{\vb{Q}_{\lambda}}$ to transform
the initial state $\tau^{(\lambda)}_{\bm{\theta}_0}$ to the final state
\begin{align}
 \rho_{\rm opt}^{(\lambda)}:=&\sum_{i}\ket{\omega_i^{\vb*{\theta}_{\lambda},\lambda}}\ev{\tau^{(\lambda)}_{\bm{\theta}_0}}{\omega_i^{\vb*{\theta}_0,\lambda}}\bra{\omega_i^{\vb*{\theta}_{\lambda},\lambda}}
 \nonumber\\
 =&\sum_{i}
 p^{(\lambda)}_{\bm{\theta}_0}(i)\ketbra{\omega_i^{\vb*{\theta}_{\lambda},\lambda}}{\omega_i^{\vb*{\theta}_{\lambda},\lambda}},
\end{align}
i.e.
\begin{align}
 \Gamma_{\rm opt}^{\vb{Q}_{\lambda}}(\tau^{(\lambda)}_{\bm{\theta}_0})=\rho_{\rm opt}^{(\lambda)}.\label{GM1}
\end{align}
In the final state $\rho_{\rm opt}^{(\lambda)}$, the $i$-th largest probability $p_{\bm{\theta}_0}^{(\lambda)}(i)$
is assigned to
the eigenstate $\ket{\omega_i^{\vb*{\theta}_{\lambda},\lambda}}$ instead of the original eigenstate $\ket{\omega_i^{\vb*{\theta}_{0},\lambda}}$.
Since the operation $\Gamma_{\rm opt}^{\vb{Q}_{\lambda}}$ exchanges the eigenstates,
it satisfies the unital condition.
Therefore,
$\Gamma_{\rm opt}^{\vb{Q}_{\lambda}}\otimes \text{id}_C$ is an implicit-battery operation on $\mathcal{H}_{\mathrm{Baths}}\otimes\mathcal{H}_C$
and satisfies the unitalness. 
The cyclicity is also trivially satisfied.
Especially, we do not use any catalytic effects of $\mathcal{H}_C$ in this operation.
Fig.~\ref{figure5} is a schematic picture of $\Gamma_{\rm opt}^{\vb{Q}_{\lambda}}$.

In Sec.~\ref{sub_ach_FGCB}, we show that this final state $\rho_{\mathrm{opt}}^{(\lambda)}$ achieves the FGCB.
Although such a CPTP map $\Gamma_{\mathrm{opt}}^{\vb{Q}_{\lambda}}$ is not unique,
the proof of the achievability relies only on the final state $\rho_{\rm opt}^{(\lambda)}$.
%Since not all of $\Gamma_{\mathrm{opt}}^{\vb{Q}_{\lambda}}$ can be extended to an explicit-battery operation \cite{PhysRevA.95.032132},
%%CP1
However, since an implicit-battery operation is not necessarily extended to an explicit-battery operation,
before showing the achievability of FGCB, we have to
construct
an explicit-battery unitary operation to be reduced to a CPTP map
$\Gamma_{\mathrm{opt}}^{\vb{Q}_{\lambda}}$ satisfying \eqref{GM1}.
Then, under an explicit battery given in Sec.~\ref{sub_exp_bat}, we construct such a unitary operation in Sec.~\ref{sub_cons_exp}.

\subsection{Explicit battery}\label{sub_exp_bat}
To show the tightness of FGCB, we should construct a unitary operation on the whole system in the appropriate 'explicit' formulation as mentioned in Sec.~\ref{sub_imp_exp}.
To do so, we fix an explicit formulation by choosing an appropriate battery system $\mathcal{H}_W$ and reasonable constraints on the operations with explicit batteries as follows.

%We treat the battery system with discrete spectrum.
As in \cite{Guryanova:2016aa},
we assume that the battery system ${\mathcal H}_W$ is $\mathcal{H}_{W_a}\otimes\mathcal{H}_{W_b}=L^2(\mathbb{R})^{\otimes 2}$, where the components $\mathcal{H}_{W_a}=L^2(\mathbb{R})$ and $\mathcal{H}_{W_b}=L^2(\mathbb{R})$ of the tensor product correspond to the degree of freedom for Quantities $A$ and $B$, respectively.
Let the respective battery observables $A_W$ and $B_W$ of Quantities $A$ and $B$ be given as
$A_W=c_a \hat{x}_a$, $B_W=c_b\hat{x}_b$, where
$c_a$ and $c_b$ are the constants, $\hat{x}_a$ and $\hat{x}_b$ are the independent position operators.
We can also construct the battery system with discrete spectrum in the same way as \cite{Aberg:2014aa,PhysRevA.95.032132}.
%For continuous spectrum, if we set $A_W$ and $B_W$ as the 'weight position' the same as \cite{Guryanova:2016aa}\cite{Skrzypczyk:2014aa},
%we can similarly proceed the following analysis.
Note that such a bit unphysical doubly infinite spectrum of the battery is an idealization to focus on the theoretical limit to the performance of the engine,
%That is, we focus on the optimal performance with extremely large and symmetric battery to remove the effects of the
which is similar to that we do not care about the length of the string suspending the weight in thermodynamics.

%Then, we restrict the operations of the whole system $\mathcal{H}_{\mathrm{Baths}}\otimes\mathcal{H}_C\otimes\mathcal{H}_W$ to the unitary operations $U$ on $\mathcal{H}_{\mathrm{Baths}}\otimes\mathcal{H}_C\otimes\mathcal{H}_W$
%which satisfies:
To show that the FGCB is really achieved by properly work-like transportation of the quantity,
it is not enough to just impose
the conditions A\ref{it_ic-1}, A\ref{it_ic-2} on the operations on $\mathcal{H}_{\mathrm{Baths}}\otimes\mathcal{H}_C$ under which FGCB is verified.
Stronger conditions are needed on the dynamics of the whole system $\mathcal{H}_{\mathrm{Baths}}\otimes\mathcal{H}_C\otimes\mathcal{H}_W$ with the explicit battery $\mathcal{H}_W$ fixed above.
As reasonable constraints for our explicit-battery formulation,
we consider the following conditions~B\ref{item_strict_cons}-B\ref{item_noc2}
on a unitary operation $U$
on $\mathcal{H}_{\mathrm{Baths}}\otimes\mathcal{H}_C\otimes\mathcal{H}_W$
to be allowed as a dynamics of the generalized heat engine.
%per one cycle
In other words,
if a unitary $U$ satisfies the following conditions B\ref{item_strict_cons}-B\ref{item_noc2},
there exists a generalized heat engine to implement $U$ per unit cycle whose output work is
\begin{align}
 \Delta W_A=\tr A_W U \rho_0 U^{\dagger}- \tr A_W \rho_0 \label{workdef}
\end{align}
for the initial state $\rho_0=\tau_{\bm{\theta}_0}^{(\lambda)}\otimes \rho_C \otimes \rho_W$ of the total system.
%, in consideration of physical relevance.
In fact, as we mention in the last paragraph of this subsection,
such an allowed operation $U$ reduces to an operation on $\mathcal{H}_{\mathrm{Baths}}\otimes\mathcal{H}_C$ satisfying
the constraints A\ref{it_ic-1}, A\ref{it_ic-2} on the implicit-battery operation (Definition \ref{def_implicit}).
%where conservation law holds for each $i$ between the sum of all observables belong to $\mathcal A_i$ and corresponding observable $A_{W}^{(i)}$ of $\mathcal{H}_W$, where $\mathcal{H}_C$ just works to implement the process and unchanged through the process.
%More precisely,

 \begin{enumerate}
  \renewcommand{\labelenumi}{B\arabic{enumi}.}
 \item Strict conservation law:
      \begin{align}
	\hspace{6mm}\left[\sum_{j=1}^2A_{j,\lambda}+A_{W},U\right]=\left[\sum_{j=1}^2B_{j,\lambda}+B_{W},U\right]=0,\label{strict_conserve}
      \end{align}
       where $[O_1,O_2]$ denotes the commutator $O_1O_2-O_2O_1$ of two operators $O_1$ and $O_2$.
       \label{item_strict_cons}\\
 \item Cyclicity of the engine:\\
       There exists a state $\rho_C$ of $\mathcal{H}_C$ such that
       \begin{align} 
  \tr_{\mathcal{H}_{\mathrm{Baths}}\otimes\mathcal{H}_W}U(\tau^{(\lambda)}_{\bm{\theta}_0}\otimes\rho_C\otimes\rho_W)U^{\dagger}=\rho_C\label{catalytic}
       \end{align}
       holds for an arbitrary initial state $\rho_W$ of the battery.\label{item_cycle}\\
 \item Independence of the initial state of the battery (`no-cheating condition 1'):
      \begin{align}
       &\tr_{\mathcal{H}_W} U(\tau^{(\lambda)}_{\bm{\theta}_0}\otimes\rho_C\otimes\rho_{W,1})U^{\dagger}\nonumber\\
       =&\tr_{\mathcal{H}_W} U(\tau^{(\lambda)}_{\bm{\theta}_0}\otimes\rho_C\otimes\rho_{W,2})U^{\dagger}\label{noc_1}
      \end{align}
       for any states $\rho_{W,1},\rho_{W,2}$ of $\mathcal{H}_W$.\label{it_noc1}
 \item Translational symmetry (`no-cheating condition 2'):
       \begin{align}
	[\Delta_A^{\epsilon}, U]=[\Delta_B^{\epsilon}, U]=0, \label{no_cheating}
       \end{align}
       where we define
       the translation operators of $A_W$ as
\begin{align}
\Delta_A^{\epsilon}:=\exp (-i\epsilon \hat{p}_a)\label{eq-26}
\end{align}
 by the momentum operator $\hat{p}_a$ conjugate to $\hat{x}_a$.
       The translation operator $\Delta_B^{\epsilon}$ of $B_W$ is similarly defined.
       \label{item_noc2}
 \end{enumerate}
Unitarity is required to prohibit using any resource outside of $\mathcal{H}_{\mathrm{Baths}}\otimes\mathcal{H}_C\otimes\mathcal{H}_W$.
In addition to the conservation laws and cyclicity of the engine, we demand the no-cheating condition as in
\cite{1302.2811,Skrzypczyk:2014aa,Guryanova:2016aa}.
Independence of the initial state of the battery (Condition B\ref{it_noc1}) is to prevent ourselves from cheatingly using the battery as other than a battery, e.g. as like a `cold reservoir'.
That is, we guarantee that there is no hidden heat-like transfer of each quantity with the battery itself,
which is non-trivial to verify in quantum thermodynamics.
Indeed, if there is such a heat-like transfer, it must depend on the state of the battery.
As Condition B\ref{item_noc2}, translational symmetry is individually imposed since it is not shown to automatically follow from Condition B\ref{it_noc1}, and vice versa.
%%%%check23
Indeed, the translational symmetry of the battery is needed because it guarantees that the generalized heat engine works properly
even when we cannot control the initial state on the battery and can observe only the translation of the battery \cite{Skrzypczyk:2014aa}.
The relevance of this requirement can be found by considering the typical case where the battery is given as the `height' of the weight.
%The translational symmetry of the battery reflects the fact that only displacements in the eigenvalues of the battery's observables, which correspond to the `height' of the weight, are important \cite{Skrzypczyk:2014aa}.
%%%%%%%%
%This is needed to guarantee that only displacements caused by the operations are relevant.
Such extreme symmetries of the battery are sufficient to remove undesired effects from the battery.

Furthermore, the reduced dynamics $\Gamma(\tau^{(\lambda)}_{\bm{\theta}_0}\otimes\rho_C)=\tr_{\mathcal{H}_W} U(\tau^{(\lambda)}_{\bm{\theta}_0}\otimes\rho_C\otimes\rho_W)U^{\dagger}$ on $\mathcal{H}_{\mathrm{Baths}}\otimes\mathcal{H}_C$ is unital for arbitrary $\rho_W$ \cite{Guryanova:2016aa,Skrzypczyk:2014aa}.
Thus, an operation in this explicit-battery formulation indeed reduces to an implicit-battery operation defined in Definition \ref{def_implicit}.
Hence, for showing the achievability of FGCB, it is enough to construct an operation to achieve it under these constraints.
In the next subsection, we construct a global unitary operation to achieve FGCB under these conditions.

\subsection{Construction of the `explicit' operation}\label{sub_cons_exp}  
  Now, we construct a global unitary operation which is reduced to
  an operation $\Gamma_{\mathrm{opt}}^{\vb{Q}_{\lambda}}$ satisfying \eqref{GM1}.
  Using the translation operators $\Delta_A^{\epsilon}$, $\Delta_B^{\epsilon}$, and the state $\ket{\omega_i^{\vb*{\theta},\lambda}}$
defined in \eqref {eq-26}, and \eqref{diag1} respectively,
we define the unitary operator $U_{\rm opt}^{(\lambda)}(\vb{Q}_{\lambda})$ on $\mathcal{H}_{\mathrm{Baths}}\otimes\mathcal{H}_{W}$ depending on $\vb{Q}_{\lambda}=(\Delta Q_{A,2,\lambda},\Delta Q_{B,1,\lambda}, \Delta Q_{B,2,\lambda})$ as follows:
\begin{align}
 &U_{\rm opt}^{(\lambda)}(\vb{Q}_{\lambda})\nonumber\\
 :=&\sum_{i}
 \ketbra{\omega_i^{\vb*{\theta}_{\lambda},\lambda}}{\omega_i^{\vb*{\theta}_0,\lambda}}\nonumber\\
 &\otimes
 \Delta_{A}^{c_a^{-1}(a_{1,\lambda}(\omega_i^{\vb*{\theta}_{0},\lambda})+a_{2,\lambda}(\omega_i^{\vb*{\theta}_{0},\lambda})-a_{1,\lambda}(\omega_i^{\vb*{\theta}_{\lambda},\lambda})-a_{2,\lambda}(\omega_i^{\vb*{\theta}_{\lambda},\lambda}))}\nonumber\\
 &\otimes
 \Delta_{B}^{c_b^{-1}(b_{1,\lambda}(\omega_i^{\vb*{\theta}_{0},\lambda})+b_{2,\lambda}(\omega_i^{\vb*{\theta}_{0},\lambda})-b_{1,\lambda}(\omega_i^{\vb*{\theta}_{\lambda},\lambda})-b_{2,\lambda}(\omega_i^{\vb*{\theta}_{\lambda},\lambda}))}.
\end{align}
Note that $\vb{Q}_{\lambda}$ does not have the meaning of generalized heat
at this moment.
That is, the amount of the generalized heat of this protocol $U_{\rm opt}^{(\lambda)}(\vb{Q}_{\lambda})$ has not guaranteed to be $\vb{Q}_{\lambda}$.
Instead, $\vb{Q}_{\lambda}$ should be regarded just as a variable, though
it will turn out that it indeed asymptotically corresponds to the generalized heat of this protocol.

%Note that $\bm{\theta}_{\lambda}$ depends on $\vb{Q}_{\lambda}$.
%$U_{\rm opt}^{(\lambda)}(\vb{Q}_{\lambda})$ is well defined unitary operation by the condition 2. and 3., and satisfies the conservation law (\ref{strict_conserve}).
For an arbitrary fixed initial state $\rho_W$ of the battery,
we define the reduced dynamics $\Gamma_{\mathrm{opt}}^{\vb{Q}_{\lambda}}$
as an implicit-battery protocol
by $\Gamma_{\mathrm{opt}}^{\vb{Q}_{\lambda}}(\rho)= \tr_{\mathcal{H}_W} U_{\mathrm{opt}}^{(\lambda)}(\vb{Q}_{\lambda})(\rho \otimes\rho_W )U_{\mathrm{opt}}^{(\lambda)}(\vb{Q}_{\lambda})^{\dagger}$.
The condition $\Gamma_{\mathrm{opt}}^{\vb{Q}_{\lambda}}(\tau_{\bm{\theta}_0}^{(\lambda)})=\rho_{\rm opt}^{(\lambda)}$ is satisfied regardless of the state $\rho_W$.
Thus, once the unitary operator $U_{\mathrm{opt}}^{(\lambda)}(\vb{Q}_{\lambda})$ satisfies the conditions \eqref{strict_conserve}-\eqref{no_cheating},
we find that the reduced dynamics $\Gamma_{\mathrm{opt}}^{\vb{Q}_{\lambda}}$ is the desired implicit-battery operation
satisfying the property given in Sec.~\ref{sub_cons_imp}.

In fact,
the unitary operator $U_{\rm opt}^{(\lambda)}(\vb{Q}_{\lambda})$ satisfies the strict conservation laws (\ref{strict_conserve})
since
%it is a natural extension of our
%implicit-battery protocol to incorporate the conservation
%laws since
the battery part of the operation
absorbs
%describes the transition of each quantity of the battery which coincides with
the transition of the corresponding quantity
of the baths.
No-cheating condition (\ref{noc_1}), (\ref{no_cheating}) is also easily verified.
The global unitary operation on $\mathcal{H}_{\mathrm{Baths}}\otimes\mathcal{H}_W\otimes\mathcal{H}_C$ is $U_{\rm opt}^{(\lambda)}(\vb{Q}_{\lambda})\otimes I_C$, which obviously satisfies the cyclicity (\ref{catalytic}).
%Further, we do not use any catalytic power of $\mathcal{H}_C$ in this operation.
Therefore, this global unitary operation satisfies all the conditions of the explicit formulation.
Thus, the final state $\rho_{\mathrm{opt}}^{(\lambda)}$ is verified to be attained by an allowed operation.
%it is sufficient to show that $\rho_{\rm opt}^{(\lambda)}:=\mathcal E_{\rm opt}(\tau_{\bm{\theta}_0}^{(\lambda)})$ asymptotically achieves the equality of (\ref{carnot}).

In the next subsection, we show that this final state $\rho_{\mathrm{opt}}^{(\lambda)}$ really achieves the equality in FGCB in the asymptotic sense.

\subsection{Achievement of the equality in FGCB}\label{sub_ach_FGCB}
%To show the tightness of the FGCB, we
Our goal is to show that our constructed protocol $U_{\rm{opt}}^{(\lambda)}(\vb{Q}_{\lambda})\otimes I_C$ achieves the maximum extraction $\Delta W^{\mathrm{opt}}_{A,\lambda}(\vb{Q}_{\lambda})$ except for $o\left(\frac{\|\vb{Q}_{\lambda}\|^2}{\lambda}\right)$ order of error terms.
The work output $\Delta W_A$ of a generalized heat engine implementing an allowed unitary operation $U$ is defined by
\eqref{workdef}
with the initial state $\rho_0=\tau_{\bm{\theta}_0}^{(\lambda)}\otimes \rho_C \otimes \rho_W$.
%For any allowed unitary operation $U$,
The strict conservation law B\ref{item_strict_cons} of Quantity $A$ implies the conservation of the sum of the average values of them:
\begin{align}
 \tr (A_{1,\lambda}+A_{2,\lambda} + A_W)U\rho_0 U^{\dagger}
 =\tr (A_{1,\lambda}+A_{2,\lambda} + A_W)\rho_0.
\end{align}
Hence, if the final state of the baths is $\rho_{\mathrm{Baths}}'$,
the work output is given by
\begin{align}
 \Delta W_A= \tr (A_{1,\lambda}+A_{2,\lambda})(\tau_{\bm{\theta}_0}^{(\lambda)}-\rho_{\mathrm{Baths}}').
\end{align}
We denote the work $\Delta W_A$ with the final state $\rho_{\mathrm{Baths}}'$ of the baths by $\Delta W_A(\rho_{\mathrm{Baths}}')$.
When $f(\lambda)/g(\lambda)\rightarrow 0$, we write $f(\lambda)\ll g(\lambda)$.
Then, the statement of the achievability of FGCB is summarized in the following theorem:
  \begin{theorem}\label{achieve_thm}
 Let $A_{i,\lambda}$ and $B_{i,\lambda}$ $(i=1,2)$ be mutually commutative.
 We assume that Assumption \ref{assump2} is satisfied.
  For any
  $\vb{Q}_{\lambda}=(\Delta Q_{A,2,\lambda},\Delta Q_{B,1,\lambda}, \Delta Q_{B,2,\lambda})$,
  there exists a generalized heat engine implementing $U_{\rm{opt}}^{(\lambda)}(\vb{Q}_{\lambda})\otimes I_C$ in the sense of the explicit-battery formulation B\ref{item_strict_cons}-B\ref{item_noc2}.
  %The final state $\tr_W U_{\mathrm{opt}}(\tau_{\bm{\theta}_0}^{(\lambda)}\otimes\rho_W )U_{\mathrm{opt}}^\dagger$ of the baths of this heat engine is equal to $\rho_{\rm opt}^{(\lambda)}$ regardless of the initial state $\rho_W$ of the battery.
  If $\vb{Q}_{\lambda}$ satisfies
  \begin{align}
   \lambda^{\frac{5}{8}}\ll \|\vb{Q}_{\lambda}\| \ll \lambda,\label{Q_order}
  \end{align}
 then this engine indeed runs with the generalized heat $\vb{Q}_{\lambda}$ up to $o\left(\frac{\|\vb{Q}_{\lambda}\|^2}{\lambda}\right)$, i.e.
  \begin{align}
  \tr A_{2,\lambda}(\tau_{\bm{\theta}_0}^{(\lambda)}-\rho_{\rm opt}^{(\lambda)})&=\Delta Q_{A,2,\lambda}+o\left(\frac{\|\vb{Q}_{\lambda}\|^2}{\lambda}\right)\label{ac_2}\\
  \tr B_{i,\lambda}(\tau_{\bm{\theta}_0}^{(\lambda)}-\rho_{\rm opt}^{(\lambda)})&=\Delta Q_{B,i,\lambda}+o\left(\frac{\|\vb{Q}_{\lambda}\|^2}{\lambda}\right) \; (i=1,2),\label{ac_3}
  \end{align}
  where $\rho_{\rm opt}^{(\lambda)}$ is the final state of the baths.
  The work output of Quantity $A$ of this engine satisfies
  \begin{align}
   \Delta W_A (\rho_{\mathrm{opt}}^{(\lambda)}) &=\Delta W^{\mathrm{opt}}_{A,\lambda}(\vb{Q}_{\lambda})+o\left(\frac{\|\vb{Q}_{\lambda}\|^2}{\lambda}\right),\label{ac_1}
  \end{align}
  where $\Delta W^{\mathrm{opt}}_{A,\lambda}(\vb{Q}_{\lambda})$ is the maximum work up to the second leading order given by FGCB \eqref{FGCB} with the generalized heat $\vb{Q}_{\lambda}$.
  Hence, FGCB is asymptotically achieved up to $o\left(\frac{\|\vb{Q}_{\lambda}\|^2}{\lambda}\right)$ by this engine.
\end{theorem}
%%%explain about the fluctuation outside the proof environment
Firstly, we remark that the final state $\rho_{\mathrm{opt}}^{(\lambda)}$ is not uniquely determined 
because it depends on the choices of the orderings of states $\left\{\ket{\omega_i^{\vb*{\theta}_{0},\lambda}}\right\}_i, \left\{\ket{\omega_i^{\vb*{\theta}_{\lambda},\lambda}}\right\}_i$ among their multiplicity because of the degeneracy.
However, any final state $\rho_{\mathrm{opt}}^{(\lambda)}$ satisfies Theorem \ref{achieve_thm} because any choice makes no difference in the following analysis.
%Although the state $\rho_{\mathrm{opt}}$ itself depends on the choice of the ordering  among its multiplicity because of the degeneracy,
%any $\rho_{\mathrm{opt}}$ satisfies Theorem \ref{achieve_thm}
%since any choice makes no difference in the following analysis.

According to this theorem, we can extract the maximum amount $\Delta W^{\mathrm{opt}}_{A,\lambda}(\vb{Q}_{\lambda})$ of the work given in FGCB \eqref{FGCB} in the asymptotic sense up to $o\left(\frac{\|\vb{Q}_{\lambda}\|^2}{\lambda}\right)$, if we run the protocol $U_{\mathrm{opt}}^{(\lambda)}(\vb{Q}_{\lambda})\otimes I_C$ with appropriate order of $\vb{Q}_{\lambda}$.
%under the additional conditions.
 %The error in
 Though the actual generalized heat of this protocol has the error up to $o\left(\frac{\|\vb{Q}_{\lambda}\|^2}{\lambda}\right)$,
 \begin{align}
  \Delta W^{\mathrm{opt}}_{A,\lambda}\left(\vb{Q}_{\lambda}+o\left(\frac{\|\vb{Q}_{\lambda}\|^2}{\lambda}\right)\right)=\Delta W^{\mathrm{opt}}_{A,\lambda}(\vb{Q}_{\lambda})
 \end{align}
 is obvious from FGCB.
 Thus, the equality in FGCB is achieved by our protocol asymptotically up to $o\left(\frac{\|\vb{Q}_{\lambda}\|^2}{\lambda}\right)$, hence FGCB is tight.
 Furthermore, since the dynamics on $\mathcal{H}_C$ is simply the identity in our protocol, we do not use catalytic effects at all.
 This construction shows that catalytic effects work in small order of $o\left(\frac{\|\vb{Q}_{\lambda}\|^2}{\lambda}\right)$ for the optimal performance.
 Of course, since our protocol achieves GCB in thermodynamic limit,
 our protocol with the generic scaling $\lambda$ is novel even for the regime of thermodynamic limit.
 Although our derivation imposed the condition on
the norm of the observables for the technical simplicity,
there is possibility to remove it.
 %We imposed the condition on the norm of the observables for the technical simplicity, though there is possibility to remove it.
 The conditions \eqref{str_extensive_1} and \eqref{str_extensive_2} with $\alpha<\frac{1}{2}$ are also needed for our analysis to work, which seem to be more essential.
 The reason is that
 larger order than $\lambda^{\alpha}$ $(\alpha\geq\frac{1}{2})$ of the deviation from the extensivity \eqref{extensive_1}, \eqref{extensive_2} possibly degrades
 the performance of the engine.
 %Deviation from the extensivity may reflect the effects of the
 %%%%%%
%from the extensivity degrades the performance.
%Finally, we assume that $\|P\Delta A_{\lambda}^*\|$ is small enough so that $\|P\Delta A_{\lambda}^*\|=o(\lambda^{\frac{2-\alpha}{2}})$.
%Finally, we assume that $\|P\Delta A^*_{\lambda}\|/\lambda^{\frac{3}{4}}\rightarrow\infty$, which is denoted by $\|\Delta A^*_{\lambda}\|=\Omega(\lambda^{\frac{3}{4}})$.

Now, we verify Theorem \ref{achieve_thm}.
The ideal thermal state $\tau^{(\lambda)}_{\bm{\theta}_\lambda}$ attains the equality in FGCB~(\ref{FGCB}) under the given heat amounts (\ref{tl_3}) and (\ref{tl_4}) by its definition, though this state itself is not necessarily achieved from the initial state $\tau^{(\lambda)}_{\bm{\theta}_0}$ through an allowed operation.
Thus, in order to prove Theorem \ref{achieve_thm},
it is sufficient to show that each expectation value of $\rho_{\text{opt}}^{(\lambda)}$ is close to that of $\tau^{(\lambda)}_{\bm{\theta}_\lambda}$ in the order of $o\left(\frac{\|\vb{Q}_{\lambda}\|^2}{\lambda}\right)$.
%analyzing information geometric structure,
Then, we firstly observe the relation between the differences
%$\left|\langle A_{i,\lambda}\rangle_{\rho_{\rm opt}^{(\lambda)}}-\eta_{\lambda,i}(\beta_{\lambda})\right|$
$|\tr A_{i,\lambda}(\tau^{(\lambda)}_{\bm{\theta}_\lambda}-\rho_{\text{opt}}^{(\lambda)})|$, $|\tr B_{i,\lambda}(\tau^{(\lambda)}_{\bm{\theta}_\lambda}-\rho_{\text{opt}}^{(\lambda)})|$
in the expectation values
and the relative entropy $D(\rho_{\rm opt}^{(\lambda)}\|\tau_{\bm{\theta}_\lambda}^{(\lambda)})$.
To do so, it is sufficient to focus on the thermal state at the effective inverse temperature
$\bm{\xi}_{\lambda}:=\tilde{\bm{\theta}}_{\lambda}(\rho_{\mathrm{opt}}^{(\lambda)})$
of $\rho_{\mathrm{opt}}^{(\lambda)}$
since it shares the expectation values with $\rho_{\mathrm{opt}}^{(\lambda)}$ as
\begin{align}
 \eta_{\lambda,i}(\bm{\xi}_\lambda)&=\tr A_{i,\lambda}\rho_{\rm opt}^{(\lambda)},\\
 \eta_{\lambda,i+2}(\bm{\xi}_\lambda)&=\tr B_{i,\lambda}\rho_{\rm opt}^{(\lambda)}\quad(i=1,2).
\end{align}
Then, we show the following lemma.
 \begin{lemma}\label{expect_rel}
  For the effective inverse temperature $\bm{\xi}_{\lambda}$ of $\rho_{\mathrm{opt}}^{(\lambda)}$,
%, i.e. the expectation values are the same as that of $\rho_{\rm opt}^{(\lambda)}$.
  we have
\begin{align}
 &2D(\rho_{\rm opt}^{(\lambda)}\|\tau_{\bm{\theta}_\lambda}^{(\lambda)})
 \max_{t\in[0,1]}\|(J_{\lambda,ij}(\vb{s}_{\lambda}(t)))_{ij}\|\nonumber\\
 \geq&
 \|\bm{\eta}_{\lambda}(\bm{\theta}_{\lambda})-\bm{\eta}_{\lambda}(\bm{\xi}_{\lambda})\|^2,
 \label{optup}
\end{align}
  where
  we denote the matrix whose $(i,j)$-component is $a_{ij}$ by $(a_{ij})_{ij}$,
  and $\vb{s}_{\lambda}(t)$ is the inverse temperature to satisfy $\bm{\eta}_{\lambda}(\vb{s}_{\lambda}(t))=t\bm{\eta}_{\lambda}(\bm{\theta}_{\lambda})+(1-t)\bm{\eta}_{\lambda}(\bm{\xi}_{\lambda})$.
  Here, $\|A\|$ for a matrix $A$ is the matrix norm.
\end{lemma}
%between $\rho_{\rm opt}^{(\lambda)}$ and $\tau_{\bm{\theta}_\lambda}^{(\lambda)}$.
Once this lemma is proved,
our problem is further reduced to the estimation of the left hand side (LHS) of (\ref{optup})
because the difference in the expectation value of each quantity between $\tau^{(\lambda)}_{\bm{\theta}_\lambda}$ and $\rho_{\text{opt}}^{(\lambda)}$ are smaller than $\|\bm{\eta}_{\lambda}(\bm{\theta}_{\lambda})-\bm{\eta}_{\lambda}(\bm{\xi}_{\lambda})\|$.
%Because the differences in the expectation values between $\tau^{(\lambda)}_{\bm{\theta}_\lambda}$ and $\rho_{\text{opt}}^{(\lambda)}$ are smaller than $\|\bm{\eta}_{\lambda}(\bm{\theta}_{\lambda})-\bm{\eta}_{\lambda}(\bm{\xi}_{\lambda})\|$,
%our problem is further reduced to the estimation of the left hand side (LHS) of (\ref{optup})
%from this lemma. we just focus on the catalytic dynamics on the baths' system $\mathcal{H}_{\mathrm{Baths}}$ implemented by unital CP-maps on $\mathcal{H}_{\mathrm{Baths}}\otimes\mathcal{H}_C$ which preserves the state of $\mathcal{H}_C$.
%Under this situation, we derive the tighter constraint on the generalized work extractions than \cite{} by taking account of the scale of the baths' in Sec.\ref{}.
%Further, we explicitly show the achievability of the bound in asymptotic sense by constructing the CP-map.
%Obviously, the
%%%%%%%%%%%%%%%%%%%%%%%%%%%%%%%%%%%%%%%%%%%%%%%%%%%proof of lemma 1
\begin{proof}[Proof of Lemma \ref{expect_rel}]
 To verify Lemma \ref{expect_rel}, we focus on the following information geometric estimations.
 In fact, since $\rho_{\mathrm{opt}}^{(\lambda)}$ has full rank, we can apply the methods of information geometry in Appendix \ref{app_bregman}.
%The states $\rho^{(\lambda)}_{\rm opt}$ and $\tau^{(\lambda)}_{\bm{\xi}_{\lambda}}$ are on the same mixture family, which consists of the states sharing the same expectation values of the observables under the consideration.
 Since $\bm{\xi}_{\lambda}$ is the effective inverse temperature of $\rho_{\mathrm{opt}}^{(\lambda)}$,
 $\tau^{(\lambda)}_{\bm{\xi}_{\lambda}}$ is the thermal state
 sharing the expectation values of $A_{i.\lambda}$ and $B_{i,\lambda}$ with $\rho_{\mathrm{opt}}^{(\lambda)}$.
 %, see Appendix \ref{app_pyth}.
 %projection of $\tau^{(\lambda)}_{\bm{\theta}_{\lambda}}$ onto this mixture family,
 Thus, applying the Pythagorean theorem (Lemma \ref{pyth} in Appendix \ref{app_pyth}), we obtain
% \footnote{The quantum relative entropy is a Bregman divergence when
%the potential function is given as the free entropy $\phi_{\lambda}(\bm{\theta})$ \cite[Theorem 6.5]{Hayashi:2017aa}.
%Then, Pythagorean theorem for Bregman divergence \cite[Theorem 2.3]{Hayashi:2017aa}
% guarantees the first equation of (\ref{pythagorean}).}
\begin{align}
 D(\rho_{\rm opt}^{(\lambda)}\|\tau_{\bm{\theta}_\lambda}^{(\lambda)})
 =&D(\rho_{\rm opt}^{(\lambda)}\|\tau_{\bm{\xi}_{\lambda}}^{(\lambda)})
 +D(\tau_{\bm{\xi}_{\lambda}}^{(\lambda)}\|\tau_{\bm{\theta}_\lambda}^{(\lambda)})\nonumber\\
 \geq&
 D(\tau_{\bm{\xi}_{\lambda}}^{(\lambda)}\|\tau_{\bm{\theta}_\lambda}^{(\lambda)})
 \label{pythagorean}
\end{align}
as illustrated in Fig.~\ref{figure_4d}.
%%%%%%%%%%%%%%%%%%%%%%%%%%%%%%%%%%%%%%%%%%%%%%%%%%%%%figure
\begin{figure}[!t]
\centering
\includegraphics[clip ,width=3.4in]{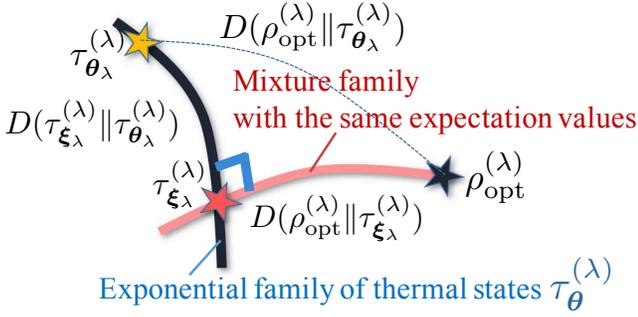}
 \caption{Information geometric positional relationships for the Pythagorean theorem.
 The states sharing the same expectation values form a mixture family.
 The family of the thermal states parametrized by inverse temperatures is an exponential family.
 The definition and details of the exponential and the mixture families of the states are given in
 Appendix \ref{app_bregman}.}
\label{figure_4d}
\end{figure}
%%%%%%%%%%%%%%%%%%%%%%%%%%%%%%%%%%%%%%%%%%%%%%%%%%%%%%%%%f
Furthermore, applying the relation \eqref{apprel_re} in Appendix \ref{app_pyth}, we have the following relation between the relative entropy and expectation values:
\begin{align}
 &D(\tau_{\bm{\xi}_{\lambda}}^{(\lambda)}\|\tau_{\bm{\theta}_\lambda}^{(\lambda)})\nonumber\\
 =&\int_{0}^{1}\sum_{ij}(\eta_{\lambda,i}(\bm{\theta}_{\lambda})-\eta_{\lambda,i}(\bm{\xi}_{\lambda}))(\eta_{\lambda,j}(\bm{\theta}_{\lambda})-\eta_{\lambda,j}(\bm{\xi}_{\lambda}))\nonumber\\
 &\times J^{ij}_{\lambda}(\vb{s}_{\lambda}(t))tdt\nonumber\\
 \geq& \frac{1}{2}\|\bm{\eta}_{\lambda}(\bm{\theta}_{\lambda})-\bm{\eta}_{\lambda}(\bm{\xi}_{\lambda})\|^2 \min_{t\in[0,1]}\frac{1}{\|(J_{\lambda,ij}(\vb{s}_{\lambda}(t)))_{ij}\|},\label{50}
\end{align}
where we used the fact that the maximum eigenvalue of $(J_{\lambda,ij}(\vb{s}_{\lambda}(t)))_{ij}$ is equal to $\|(J_{\lambda,ij}(\vb{s}_{\lambda}(t)))_{ij}\|$ since $(J_{\lambda,ij}(\vb{s}_{\lambda}(t)))_{ij}$ is a positive matrix.
 Combining \eqref{50} with (\ref{pythagorean}), we obtain \eqref{optup}.
 \end{proof}
 %Thus, our problem is reduced to the estimation of the relative entropy $D(\rho_{\rm opt}^{(\lambda)}\|\tau_{\bm{\theta}_\lambda}^{(\lambda)})$, and $\max_{t\in[0,1]}\|(J_{\lambda,ij}(\vb{s}_{\lambda}(t)))_{ij}\|$.
%
%%%%%%%%%%%%%%%%%%%%%%%%%%%%%%%%%%%%%%%%%%%%%%%%%%%%%%%%%%%%%%%%%%%%%%%%%%%%%%
 \begin{proof}[Proof of Theorem \ref{achieve_thm}]
  
The order $\lambda^{\frac{5}{8}}\ll\|\vb{Q}_{\lambda}\|\ll\lambda$ of the generalized heat is sufficient for
%where $f(\lambda)\ll g(\lambda)$ denote that $\frac{f(\lambda)}{g(\lambda)}\rightarrow 0$.
%check52
the relative entropy $D(\rho_{\rm opt}^{(\lambda)}\|\tau_{\bm{\theta}_\lambda}^{(\lambda)})$ to satisfy
\begin{align}
 D(\rho_{\rm opt}^{(\lambda)}\|\tau_{\bm{\theta}_\lambda}^{(\lambda)})=\order{\frac{\|\vb{Q}_{\lambda}\|^2}{\lambda^{2}}}+\order{\lambda^{-\frac{1}{2}}}\label{relative_ent}
\end{align}
as shown in Appendix \ref{apps_rel}.
  We just show an outline of the proof of \eqref{relative_ent} here.
  
  From the construction of $\rho_{\rm opt}^{(\lambda)}$, the following holds:
\begin{align}
 &D(\rho_{\rm opt}^{(\lambda)}\|\tau_{\bm{\theta}_\lambda}^{(\lambda)})\nonumber\\
 =&
 \tr \rho_{\rm opt}^{(\lambda)}(\log \rho_{\rm opt}^{(\lambda)}-\log \tau^{(\lambda)}_{\bm{\theta}_{\lambda}})\nonumber\\
 =&
 \sum_{j}
 p_{\bm{\theta}_0}^{(\lambda)}(j)
 (\log p_{\bm{\theta}_0}^{(\lambda)}(j)-\log p_{\bm{\theta}_{\lambda}}^{(\lambda)}(j)).
\end{align}
  Defining a random variable
\begin{align}
 Y_{l}^{(\lambda)}(j):=
 \left\{
  \begin{array}{cc}
   \frac{\log p_{\bm{\theta}_0}^{(\lambda)}(j)-\lambda\nu}{\sqrt{\lambda}} &(l=0)
  \\ \frac{\log p_{\bm{\theta}_\lambda}^{(\lambda)}(j)-\lambda\nu}{\sqrt{\lambda}} &(l=1),
  \end{array}
 \right.
 \label{def_y_main}
\end{align}
we have another expression the relative entropy
\begin{align}
 D(\rho_{\rm opt}^{(\lambda)}\|\tau_{\bm{\theta}_\lambda}^{(\lambda)})
 =\sqrt{\lambda}\left(\mathbb E_{p_{\bm{\theta}_0}^{(\lambda)}}[Y_{0}^{(\lambda)}]-\mathbb E_{p_{\bm{\theta}_0}^{(\lambda)}}[Y_{1}^{(\lambda)}]\right),\label{rel_ent_an_main}
\end{align}
where
$\nu$ denotes the asymptotic density of the negative entropy $\nu:=-(\sum_{i=0}^m\eta_{i}(\bm{\theta}_0)\beta^i+\phi(\bm{\theta}_0))$, and
  $\mathbb E_{p}[X]$ denotes the expectation value of a random variable $X$ with probability distribution $p$.
  To estimate the relative entropy, it is difficult to calculate $\mathbb E_{p_{\bm{\theta}_0}^{(\lambda)}}[Y_{1}^{(\lambda)}]$.
Instead,
we approximate $\Delta_{\lambda}(j):=Y_0^{(\lambda)}(j)-Y_1^{(\lambda)}(j)$ by a quadratic polynomial of $Y_0^{(\lambda)}(j)$.
In this way, we can calculate $\mathbb E_{p_{\bm{\theta}_0}^{(\lambda)}}[\Delta_{\lambda}]$ by calculating the moments of $Y_0^{(\lambda)}$.
  To do so, we compare the number of states.
  The idea is that the number of states
\begin{align}
 N_{l}^{(\lambda)}(a):=\left|\left\{j | Y_l^{(\lambda)}(j)\geq a\right\}\right|
\end{align}
is asymptotically close to $(Y_l^{(\lambda)})^{-1}(a)$ since $a$ is $(Y_l^{(\lambda)})^{-1}(a)$-th largest value of $Y_l^{(\lambda)}$.
Thus, asymptotically solving the equation
\begin{align}
 N_1^{(\lambda)}(a-\Delta)= N_0^{(\lambda)}(a)\label{eqq_70}
\end{align}
with respect to $\Delta$, and approximating $\Delta$ by a quadratic polynomial $Q(a)$ of $a$,
we obtain the desired approximation of $\Delta_\lambda(j)$ as
 \begin{align}
  &\Delta_{\lambda}(j)\nonumber\\
  =&Y_0^{(\lambda)}(j)-Y_1^{(\lambda)}((Y_0^{(\lambda)})^{-1}[Y_0^{(\lambda)}(j)])\nonumber\\
  \approx&
Y_0^{(\lambda)}(j)-N_1^{(\lambda)-1}(N_0^{(\lambda)}(Y_0^{(\lambda)}(j)))\nonumber\\
\approx& Q(Y_0^{(\lambda)}(j))\label{approx_mm}
 \end{align}
  by substituting $Y_0^{(\lambda)}(j)$ to $a$.
  To solve the equation \eqref{eqq_70}, we apply a similar method to \cite{PhysRevA.92.052308,PhysRevE.96.012128} to apply the strong large deviation \cite{bahadur60:_deviat_sampl_mean,joutard13:_stron_large_deviat_theor} to the estimation of $N_{l}^{(\lambda)}(a)$.
  In its derivation in Appendix \ref{apps_rel}, we generalize the central limit theorem to apply for our situation in Appendix \ref{app_gCLT}.
  Then, calculating \eqref{rel_ent_an_main} by using \eqref{approx_mm}, we obtain \eqref{relative_ent}.
  
  %  In fact, $\lambda^\alpha\ll\lambda^{\frac{3}{4}}\ll\|\vb{Q}_{\lambda}\|\ll \lambda^{\frac{2-\alpha}{2}}\ll\lambda$ is satisfied in the assumption of the theorem.
Next, combining (\ref{relative_ent}) and (\ref{optup}), we obtain
\begin{align}
  \max_{t\in[0,1]}\|(J_{\lambda,ij}(\vb{s}_{\lambda}(t)))_{ij}\|
  =\order{\lambda}\label{jj}
\end{align}
as proved in Appendix \ref{apps_fish}.
Finally, it turns out that
 \begin{align}
  \|\bm{\eta}_{\lambda}(\bm{\theta}_{\lambda})-\bm{\eta}_{\lambda}(\bm{\xi}_{\lambda})\|
  =\sqrt{\order{\frac{\|\vb{Q}_{\lambda}\|^2}{\lambda}}
  +\order{\lambda^{\frac{1}{2}}}}\label{raw_est}
 \end{align}
by (\ref{optup}), (\ref{relative_ent}) and (\ref{jj}).
Then, we check that the order \eqref{Q_order} is sufficient for the right hand side (RHS) of (\ref{raw_est}) to be $o\left(\frac{\|\vb{Q}_{\lambda}\|^2}{\lambda}\right)$, i.e.~
 \begin{align}
  \sqrt{\order{\frac{\|\vb{Q}_{\lambda}\|^2}{\lambda}}
  +\order{\lambda^{\frac{1}{2}}}}
 =o\left(\frac{\|\vb{Q}_{\lambda}\|^2}{\lambda}\right)\label{order_reduce}.
 \end{align}
In fact,
  $\lambda^{\frac{5}{8}}\ll\|\vb{Q}_{\lambda}\|\ll\lambda $ is sufficient for (\ref{order_reduce}) to be satisfied.
% \begin{align}
%  \langle A_{i,\lambda}\rangle_{\rho_{\rm opt}^{(\lambda)}}-\eta_{\lambda,i}(\bm{\theta}_0)
%  =\Delta A_{i,\lambda}^*+o\left(\frac{\|\vb{Q}_{\lambda}\|^2}{\lambda}\right),\label{sec_ach}
% \end{align}
% which means that the optimal work extraction is achieved up to $o(\frac{\|\vb{Q}_{\lambda}\|^2}{\lambda})$ terms.
  %Then, we obtain the appropriate order of $\vb{Q}_{\lambda}$ to make (\ref{ac_1})-(\ref{ac_3}) true.
 %Furthermore, to verify that the extracted work in our protocol can
%really behave as work-like,
%we focus on the order of the fluctuation of the work.
%For any initial pure state of the battery,
%the fluctuation of $A_W$ of
%the final state is of order $\order{\lambda^{\frac{1}{2}}}$ from the Assumption 1.
%Hence, when $\lambda^{\frac{3}{4}} \ll \|\vb{Q}_{\lambda}\|$,
%this fluctuation is $o\left(\frac{\|\vb{Q}_{\lambda}\|^2}{\lambda}\right)$.
%Indeed, this regime of $\vb{Q}_{\lambda}$ is compatible with
%the above regime for \eqref{order_reduce} since $\frac{5}{8},\frac{1+\alpha}{2}<\frac{3}{4}<\frac{2-\alpha}{2}$ when $\alpha<\frac{1}{2}$.
%  Thus, we obtain $\lambda^{\frac{3}{4}}\ll \|\vb{Q}_{\lambda}\|\ll \lambda^{\frac{2-\alpha}{2}}$ as appropriate order of generalized endothermic heat when the fluctuation of the extracted work is taken into account.  
  Hence, Theorem \ref{achieve_thm} is proved.
 %%%%%%%%%%%%%
  %\mathrm{V}_{p^{(\lambda)}_{\bm{\theta}_0}}\left[a_{1,\lambda}(\omega_i^{\vb*{\theta}_{0},\lambda})+a_{2,\lambda}(\omega_i^{\vb*{\theta}_{0},\lambda})\right]\nonumber\\
%   &+\mathrm{V}_{p^{(\lambda)}_{\bm{\theta}_0}}\left[a_{1,\lambda}(\omega_i^{\vb*{\theta}_{\lambda},\lambda})+a_{2,\lambda}(\omega_i^{\vb*{\theta}_{\lambda},\lambda})\right]\nonumber\\
%   &+2\sqrt{\mathrm{V}_{p^{(\lambda)}_{\bm{\theta}_0}}\left[a_{1,\lambda}(\omega_i^{\vb*{\theta}_{0},\lambda})+a_{2,\lambda}(\omega_i^{\vb*{\theta}_{0},\lambda})\right]
   %   \mathrm{V}_{p^{(\lambda)}_{\bm{\theta}_0}}\left[a_{1,\lambda}(\omega_i^{\vb*{\theta}_{\lambda},\lambda})+a_{2,\lambda}(\omega_i^{\vb*{\theta}_{\lambda},\lambda})\right]}
   %--------------------
   %\mathrm{V}_{p^{(\lambda)}_{\bm{\theta}_0}}[a_{1,\lambda}(\omega_i^{\vb*{\theta}_{0},\lambda})]+\mathrm{V}_{p^{(\lambda)}_{\bm{\theta}_0}}[a_{2,\lambda}(\omega_i^{\vb*{\theta}_{0},\lambda})]\nonumber\\
%   &+\mathrm{V}_{p^{(\lambda)}_{\bm{\theta}_0}}[a_{1,\lambda}(\omega_i^{\vb*{\theta}_{\lambda},\lambda})]+\mathrm{V}_{p^{(\lambda)}_{\bm{\theta}_0}}[a_{2,\lambda}(\omega_i^{\vb*{\theta}_{\lambda},\lambda})]\nonumber\\
%   &+2\sqrt{\mathrm{V}_{p^{(\lambda)}_{\bm{\theta}_0}}[a_{1,\lambda}(\omega_i^{\vb*{\theta}_{0},\lambda})]\mathrm{V}_{p^{(\lambda)}_{\bm{\theta}_0}}[a_{2,\lambda}(\omega_i^{\vb*{\theta}_{0},\lambda})]}\nonumber\\
%   &+2\sqrt{\mathrm{V}_{p^{(\lambda)}_{\bm{\theta}_0}}[a_{1,\lambda}(\omega_i^{\vb*{\theta}_{\lambda},\lambda})]\mathrm{V}_{p^{(\lambda)}_{\bm{\theta}_0}}[a_{2,\lambda}(\omega_i^{\vb*{\theta}_{\lambda},\lambda})]}\nonumber\\
 \end{proof}
%%%%%%%%%%%%%%%%%%%%%%%%
 
%%%%the following should be also included in Theorem 2
%%%this sentence is unclear
%It is natural that sufficiently large heat compared to the fluctuation of the baths is needed for heat engine to work well.
%The upper bound may represent small enough order of the endotherm so that the disturbance of the baths is sufficiently small in reflection of the deviation from the extensivity.

\subsection{Non-commutative quantities}\label{sub_noncomm}
Now, we extend our protocol to the case where $A_{i,\lambda}$ and $B_{j,\lambda}$ are not commutative.
We use the same battery system $\mathcal H_W=\mathcal{H}_{W_a}\otimes\mathcal{H}_{W_b}$.
Especially, we still assume that the battery observables $A_W$ and $B_W$ commute.
This is natural since it is sufficient to use an individual system for each quantity.
In this case, instead of the strict conservation law \eqref{strict_conserve}, we just impose the average conservation law:
 \begin{enumerate}
  \renewcommand{\labelenumi}{B\arabic{enumi}*.}
 \item Average conservation law:\begin{align}
 &\tr U(\tau^{(\lambda)}_{\bm{\theta}_0}\otimes\rho_C\otimes\rho_W)U^{\dagger}\left(\sum_{j=1}^2A_{j,\lambda}+A_{W}\right)\nonumber\\
 =&\tr (\tau^{(\lambda)}_{\bm{\theta}_0}\otimes\rho_C\otimes\rho_W)\left(\sum_{j=1}^2A_{j,\lambda}+A_{W}\right),\\
 &\tr U(\tau^{(\lambda)}_{\bm{\theta}_0}\otimes\rho_C\otimes\rho_W)U^{\dagger}\left(\sum_{j=1}^2B_{j,\lambda}+B_{W}\right)\nonumber\\
 =&\tr (\tau^{(\lambda)}_{\bm{\theta}_0}\otimes\rho_C\otimes\rho_W)\left(\sum_{j=1}^2B_{j,\lambda}+B_{W}\right).
 \label{average_conserve}
\end{align}\label{item_ave_cons}
 \end{enumerate}
 The other constraints B\ref{item_cycle}-B\ref{item_noc2} remain the same as the commutative case.
 That is, we consider the conditions B\ref{item_ave_cons}*, B\ref{item_cycle}-B\ref{item_noc2} on a unitary operation $U$ on $\mathcal{H}_{\mathrm{Baths}}\otimes\mathcal{H}_C\otimes\mathcal{H}_W$
to be allowed as a dynamics of the generalized heat engine.
 %That is, if a unitary $U$ satisfies B\ref{item_ave_cons}*, B\ref{item_cycle}-B\ref{item_noc2}, we consider that there exists a generalized heat engine to implement $U$ whose output $A$-type work is $\tr A_W U \rho_0 U^{\dagger}- \tr A_W \rho_0 $
%for the initial state $\rho_0=\tau_{\bm{\theta}_0}^{(\lambda)}\otimes \rho_C \otimes \rho_W$ of the total system.

As for the first law of the thermodynamics,
Lostaglio {\it et al.}~\cite{PhysRevX.5.021001} pointed out that
some external resource of the coherence may be missed
in the formulation
without the strict conservation of the energy.
That is,
the coherence with respect to the energy eigenstates
can be increased by an operation with just the average conservation,
though it is impossible for strictly energy conservative operations.
This implies that some resource of the coherence is implicitly used to implement such an operation.
This may be the case also for generic multiple conservative quantities under the consideration.
Thus, it may be appropriate to call it `semi-explicit' battery formulation, reflecting the possibility of the lack of some resource in the formulation, while the battery system is explicitly taken into account \footnote{However, as for the energy conservation, \cite{PhysRevX.5.021001} also pointed out that if some resource of coherence is appropriately considered through the method by \cite{Aberg:2014aa}, the results under the average conservation would be revived under the strict conservation law.
This is also possibly the case for the multiple conservative quantities.}.
%as pointed out by.
%Also pointed out by \cite{PhysRevX.5.021001}
We show the achievability for non-commutative quantities in this semi-explicit battery formulation in the following.

We construct a global unitary operation satisfying B\ref{item_ave_cons}*, B\ref{item_cycle}-B\ref{item_noc2}.
Because of non-commutativity, the simultaneous eigenbasis no longer exists.
However, $\tau_{\bm \theta}^{(\lambda)}$ is diagonalized by a basis depending on $\bm \theta$ and $\lambda$.
Thus, for a given vector $\vb{Q}_{\lambda}$,
we denote the diagonalization of $\tau_{{\bm \theta}_0}^{(\lambda)}$ and $\tau_{{\bm \theta}_{\lambda}}^{(\lambda)}$ respectively by:
 \begin{align}
  \tau_{{\bm \theta}_0}^{(\lambda)}=&\sum_{i\in\mathbb N_{d_{\lambda}}}p_{{\bm \theta}_0}^{(\lambda)}(i)\ketbra{\psi_i}{\psi_i}\label{y1}\\
 \tau_{{\bm \theta}_{\lambda}}^{(\lambda)}=&\sum_{i\in\mathbb N_{d_{\lambda}}}p_{{\bm \theta}_{\lambda}}^{(\lambda)}(i)\ketbra{\varphi_i}{\varphi_i}\label{y2},
 \end{align}
 where $p^{(\lambda)}_{\bm{\theta}_0}(1)\geq p^{(\lambda)}_{\bm{\theta}_0}(2)\geq \dots$ and
 $p^{(\lambda)}_{\bm{\theta}_{\lambda}}(1)\geq p^{(\lambda)}_{\bm{\theta}_{\lambda}}(2)\geq \dots$ hold.
Note that $\ket{\psi_i},\ket{\varphi_i}$ depend also on $\lambda$, though we omit the notation for simplicity.
 Then, we define $\rho_{\rm opt,nc}^{(\lambda)}:=\sum_i\ket{\varphi_i}\ev{\tau_{\bm{\theta}_0}^{(\lambda)}}{\psi_i}\bra{\varphi_i}$
 as with the commutative case.
 %we have the protocol $\Gamma_{\rm opt}^{\vb{Q}_{\lambda}}\otimes {\rm id}_C$ in the implicit-battery formulation.

 Then, we construct the protocol in the explicit-battery formulation.
 With the same battery system as the commutative case, we define the unitary operator $U_{\rm opt, nc}^{(\lambda)}(\vb{Q}_{\lambda})$ on $\mathcal{H}_{\mathrm{Baths}}\otimes\mathcal{H}_{W}$ as
  \begin{align}
 &U_{\rm opt,nc}^{(\lambda)}(\vb{Q}_{\lambda})\nonumber\\
 :=&\sum_{i}
 \ketbra{\varphi_i}{\psi_i}\nonumber\\
 &\otimes
 \Delta_{A}^{c_a^{-1}(\ev{\sum_{l=1}^2 A_{l,\lambda}}{\psi_i}
   -\ev{\sum_{l=1}^2 A_{l,\lambda}}{\varphi_i})}\nonumber\\
 &\otimes
 \Delta_{B}^{c_b^{-1}(\ev{\sum_{l=1}^2 B_{l,\lambda}}{\psi_i}
   -\ev{\sum_{l=1}^2 B_{l,\lambda}}{\varphi_i})}.
  \end{align}
  The full protocol on $\mathcal{H}_{\mathrm{Baths}}\otimes\mathcal{H}_C\otimes\mathcal{H}_W$ is $U_{\rm opt,nc}^{(\lambda)}(\vb{Q}_{\lambda})\otimes I_C$.
  As with the commutative case, the cyclicity and no-cheating condition hold.
  Further, the average conservation B\ref{item_ave_cons}* is satisfied,
  though the strict conservation is not necessarily.
  Especially, the final state $\rho_{\rm opt,nc}^{(\lambda)}=\tr_{\mathcal{H}_W} U_{\rm opt,nc}^{(\lambda)}(\vb{Q}_{\lambda})(\tau_{\bm{\theta}_0}^{(\lambda)}\otimes\rho_W)U_{\rm opt,nc}^{(\lambda)}(\vb{Q}_{\lambda})^{\dagger}$ does not depend on the state of the battery.

  As the achievement of FGCB in this case, we show the following non-commutative version of Theorem \ref{achieve_thm}:
  \begin{theorem}\label{achieve_thm_non}
  Let $A_{i,\lambda}$ and $B_{i,\lambda}$ be not necessarily commutative.
 Let Assumption \ref{assump2} be satisfied.
For any
  $\vb{Q}_{\lambda}=(\Delta Q_{A,2,\lambda},\Delta Q_{B,1,\lambda}, \Delta Q_{B,2,\lambda})$,
  there exists a generalized heat engine implementing $U_{\rm{opt,nc}}^{(\lambda)}(\vb{Q}_{\lambda})\otimes I_C$ in the sense of the semi-explicit battery formulation B\ref{item_ave_cons}*, B\ref{item_cycle}-B\ref{item_noc2}.
  %The final state $\tr_W U_{\mathrm{opt}}(\tau_{\bm{\theta}_0}^{(\lambda)}\otimes\rho_W )U_{\mathrm{opt}}^\dagger$ of the baths of this heat engine is equal to $\rho_{\rm opt}^{(\lambda)}$ regardless of the initial state $\rho_W$ of the battery.
  If $\vb{Q}_{\lambda}$ satisfies \eqref{Q_order},
 then this engine indeed runs with the generalized heat $\vb{Q}_{\lambda}$ up to $o\left(\frac{\|\vb{Q}_{\lambda}\|^2}{\lambda}\right)$, i.e.
  \begin{align}
  \tr A_{2,\lambda}(\tau_{\bm{\theta}_0}^{(\lambda)}-\rho_{\rm opt,nc}^{(\lambda)})&=\Delta Q_{A,2,\lambda}+o\left(\frac{\|\vb{Q}_{\lambda}\|^2}{\lambda}\right)\label{ac_2nc}\\
  \tr B_{i,\lambda}(\tau_{\bm{\theta}_0}^{(\lambda)}-\rho_{\rm opt,nc}^{(\lambda)})&=\Delta Q_{B,i,\lambda}+o\left(\frac{\|\vb{Q}_{\lambda}\|^2}{\lambda}\right) \; (i=1,2),\label{ac_3nc}
  \end{align}
  where $\rho_{\rm opt, nc}^{(\lambda)}$ is the final state of the baths.
  The work output of Quantity $A$ of this engine satisfies
  \begin{align}
   \Delta W_A (\rho_{\mathrm{opt,nc}}^{(\lambda)}) &=\Delta W^{\mathrm{opt}}_{A,\lambda}(\vb{Q}_{\lambda})+o\left(\frac{\|\vb{Q}_{\lambda}\|^2}{\lambda}\right).\label{ac_1nc}
  \end{align}
  Hence, FGCB is asymptotically achieved up to $o\left(\frac{\|\vb{Q}_{\lambda}\|^2}{\lambda}\right)$ by this engine.
   %%%%%%%%%%%%%%%%%%%%%%%%%%
 %There exists a generalized heat engine implements $U_{\rm{opt,nc}}^{(\lambda)}$ in the semi-explicit battery formulation .
%  If this engine runs with the generalized heat $\vb{Q}_{\lambda}=(\Delta Q_{A,2,\lambda},\Delta Q_{B,1,\lambda}, \Delta Q_{B,2,\lambda})$ in the order
%  \begin{align}
%   \lambda^{\frac{3}{4}}\ll \|\vb{Q}_{\lambda}\|\ll \lambda^{\frac{2-\alpha}{2}},\label{Q_order}
%  \end{align}
% then its final state $\rho_{\rm opt,nc}^{(\lambda)}$ satisfies
%  \begin{align}
%   \tr A_{1,\lambda}(\tau_{\bm{\theta}_0}^{(\lambda)}-\rho_{\rm opt,nc}^{(\lambda)})&=\Delta W^{\mathrm{opt}}_{A,\lambda}(\vb{Q}_{\lambda})+o\left(\frac{\|\vb{Q}_{\lambda}\|^2}{\lambda}\right)\label{ac_1}\\
%  \tr A_{2,\lambda}(\tau_{\bm{\theta}_0}^{(\lambda)}-\rho_{\rm opt,nc}^{(\lambda)})&=\Delta Q_{A,2,\lambda}+o\left(\frac{\|\vb{Q}_{\lambda}\|^2}{\lambda}\right)\label{ac_2}\\
%  \tr B_{i,\lambda}(\tau_{\bm{\theta}_0}^{(\lambda)}-\rho_{\rm opt,nc}^{(\lambda)})&=\Delta Q_{B,i,\lambda}+o\left(\frac{\|\vb{Q}_{\lambda}\|^2}{\lambda}\right) \; (i=1,2),\label{ac_3}
%  \end{align}
%  which means that FGCB is asymptotically achieved up to $o\left(\frac{\|\vb{Q}_{\lambda}\|^2}{\lambda}\right)$.
  \end{theorem}
  The statement is the same as Theorem \ref{achieve_thm} except the construction of the protocol $U_{\mathrm{opt,nc}}^{(\lambda)}$ and its final state $\rho_{\mathrm{opt,nc}}^{(\lambda)}$, and the average conservation law B\ref{item_ave_cons}*.
  To show Theorem \ref{achieve_thm_non}, it is enough to do the same process as the proof of Theorem \ref{achieve_thm}.
  At first, Lemma \ref{expect_rel} holds by replacing $\rho_{\mathrm{opt}}^{(\lambda)}$ with $\rho_{\mathrm{opt,nc}}^{(\lambda)}$.
  The proof is the same since the Pythagorean theorem Lemma \ref{pyth} is proved with non-commutative observables in Appendix \ref{app_pyth}, so that
  \eqref{pythagorean} holds.
  Furthermore, the estimation \eqref{relative_ent} of the relative entropy is also proved in a similar way
  % we need
%  the additional assumption:
%  \begin{Ass}
%   \begin{align}
%   \tr (A_{i,\lambda}-\tr A_{i,\lambda}\tau_{\bm{\theta}}^{(\lambda)})^2\tau_{\bm{\theta}}^{(\lambda)}&=o(\lambda^2)\\
%   \tr (B_{i,\lambda}-\tr B_{i,\lambda}\tau_{\bm{\theta}}^{(\lambda)})^2\tau_{\bm{\theta}}^{(\lambda)}&=o(\lambda^2)
%  \end{align}
%  holds uniformly on a neighborhood of $\bm{\theta}_0$.
%  \end{Ass}
%  This assumption is independent of Assumption \ref{assumption} for non-commutative case since this way of the product gives a different variation of the non-commutative version of the variance from the canonical correlation.
 %It guarantees that this variation of the non-commutative variance is not too large.
 because the final state $\rho_{\rm opt,nc}^{(\lambda)}$ commutes with $\tau^{(\lambda)}_{\bm{\theta}_{\lambda}}$ as we point out in Appendix \ref{apps_rel}.
 The estimation \eqref{jj} of the canonical correlation matrix is also established in non-commutative case in Appendix \ref{apps_fish}.
 The remaining part obviously has nothing to do with non-commutativity.
 %In fact, just the expectation values and the eigenvalues $p_{{\bm \theta}_0}^{(\lambda)}(i)$, $p_{{\bm \theta}_{\lambda}}^{(\lambda)}(i)$ are relevant in the proof of Theorem \ref{achieve_thm}.
 Thus, FGCB is also achieved in the non-commutative case in the semi-explicit battery formulation.
 \section{Examples}\label{sec_example}
 In this section, we give some examples of generalized heat engines to apply our general theory established in the above.
 In particular, we treat baths with non-i.i.d.~scaling here.
 In each case, we firstly have to verify the asymptotic extensivity (Assumption \ref{assumption}) to ensure the applicability of our setup.
 Then, we calculate the second order coefficients of the optimal performance to investigate the behavior of the finite-size effect for each model.
 
 Other examples with i.i.d.~scaling are in Appendix \ref{sec_iid_example}.
 In Appendix \ref{sub_iid}, we confirm that the previous result \cite{PhysRevE.96.012128} is reproduced for the baths with i.i.d.~scaling and only one conserved quantity.
 We also treat a toy model with non-commutative two conserved quantities in Appendix \ref{sub_spin}, though its scaling is i.i.d.
 It shows some non-trivial behavior of finite-size effects with multiple conserved quantities.
 
%\subsection{Ideal gas inside of container}
%%%%%%%%%%%%%%%%%%%%%%%%%%%%%%%%%%%%average conservation
%\tr(\sum_{j}A_{B,j}^{(i)}+A_{W}^{(i)})(\rho_{BW}'\otimes\rho_C-\tau_B\otimes\rho_C\otimes\rho_W)&=0
%%%%%%%%%%%%%%%%%%%%%%%%%%%%%%%%

\subsection{1D Ising model}\label{ising}
At first, we investigate an engine with heat baths composed of 1-dimensional (1D) Ising spin chain.
Both the hot and the cold baths consist of $n$ particles of 1D Ising spin chain, whose respective Hamiltonians $H_n^{(h)}$ and $H_n^{(c)}$ are given as
\begin{align}
 H_{n}^{(b)} = -J_b(\sum_{i=1}^{n-1}\hat{s}^{(b)}_i \hat{s}^{(b)}_{i+1} + \hat{s}^{(b)}_n \hat{s}^{(b)}_1), \quad (b=h,c)
\end{align}
where $\hat{s}^{(b)}_i$ is the spin $z$-component operator at site $i$ of the hot $(b=h)$ or the cold $(b=c)$ bath, and $J_b$ is its coupling constant.
The initial state of the baths is the Gibbs state
 \begin{align}
  \tau^{(n)}_{(\beta_h,\beta_c)}=
 \frac{e^{-\beta_{h}H_{n}^{(h)}-\beta_{c}H_n^{(c)}}}{\tr e^{-\beta_{h}H^{(h)}_{n}-\beta_{c}H^{(c)}_n}},
 \end{align}
 where $\beta_h$ and $\beta_c$ are the inverse temperatures of the hot and the cold baths respectively.
The scaling is given by the number $n$ of the spins,
which is an example of non-i.i.d.~scaling.
The partition function $\mathcal Z_n^{(b)}(\beta_b)$ of the 1D Ising model is easily calculated by the transfer matrix
\begin{align}
 T=
 \begin{pmatrix}
  e^{\beta_b J_b} & e^{-\beta_b J_b}\\
  e^{-\beta_b J_b} & e^{\beta_b J_b}
 \end{pmatrix}
\end{align}
as follows \cite{Baxter:1982aa}:
%cite? for transfer matrix method
\begin{align}
 &\mathcal Z_n^{(b)}(\beta_b)\nonumber\\
 =&\sum_{s_1,\dots,s_n=-1,1}e^{\beta_b J_b(\sum_{i=1}^{n-1}s_i s_{i+1} + s_n s_1)}\nonumber\\
 =&\tr T^n
 =(2\sinh \beta_b J_b)^n + (2\cosh \beta_b J_b)^n\nonumber\\
 =&(2\cosh \beta_b J_b)^n[1+\tanh^n \beta_b J_b]\nonumber\\
 =&(2\cosh \beta_b J_b)^n[1+o(1)].
\end{align}
Thus,
the free entropy of the baths is obtained as
\begin{align}
 &\phi_{n}(\beta_h,\beta_c)\nonumber\\
 =&\log \mathcal Z_n^{(h)}(\beta_h)\mathcal Z_n^{(c)}(\beta_c)\nonumber\\
 =&n(\log [2\cosh\beta_h J_h]+\log[2\cosh \beta_c J_c]) + o(1),
\end{align}
and the asymptotic extensivity is verified.

Then, for the work extraction $\Delta W$ and the heat $\Delta Q_{h,n}$from the hot bath,
the following FGCB holds:
\begin{align}
 &\Delta W\nonumber\\
 \leq & \left(1-\frac{\beta_{h}}{\beta_{c}}\right)\Delta Q_{h,n}
 -C\frac{\Delta Q_{h,n}^2}{n}
 +o\left(\frac{\Delta Q_{h,n}^2}{n}\right).
\end{align}
To obtain the coefficient $C$,
it is enough to calculate the asymptotic density
 $\sigma_{h}^2$ and $\sigma_{c}^2$ of the variance of the energy of the baths
 since the asymptotic density $g^{ij}(\beta_h,\beta_c)$ of the inverse matrix of the Fisher information is similar as that of the i.i.d.~case \eqref{gij}.
 These are obtained as
 \begin{align}
  \sigma_{b}^2=\frac{J_b^2}{\cosh^2\beta_b J_b}\quad (b=h,c).
 \end{align}
 Thus, we have
 \begin{align}
  C
  =&\frac{1}{2}\left(
   \frac{g^{11}(\beta_{h},\beta_{c})\beta_{h}^2}{(\beta_{c})^3}
   +\frac{g^{22}(\beta_{h},\beta_{c})}{\beta_{c}}
 \right)\nonumber\\
 =&\frac{\beta_{h}^2}{2\sigma_c^2\beta_{c}^3}
  +\frac{1}{2\sigma_h^2\beta_{c}}\nonumber\\
  =&\frac{\beta_{h}^2\cosh^2 \beta_c J_c}{2 \beta_{c}^3 J_c^2}
  +\frac{\cosh^2 \beta_h J_h}{2 \beta_{c}J_h^2}.
 \end{align}
%Similarly as the example of the i.i.d.~case,
%%
%we calculate the variance of the energy of each bath
%to calculate the coefficient $C_{AA}$ as follows.

This formula implies that the absolute value of the coupling constant $J_b$ directly affects the optimal performance in the finite-size regime.
Especially, for fixed temperatures, the coefficient $C$ takes its minimum when
$J_b$ satisfies
\begin{align}
 2\beta_b J_b \sinh 2\beta_b J_b - \cosh 2\beta_b J_b - 1=0,
\end{align}
which gives the best choice of $J_b$ for the work extraction.
On the other hand, since the sign of $J_b$ makes no difference,
the optimal performance does not depend on whether
the system is ferromagnetic or anti-ferromagnetic.
%%%%%%%%%%%%
%%%%%%%%%%%%
%%%%%%%%%%%%

\subsection{Heat engine with two baths exchanging particles}\label{sub_particles}
%We consider the similar setting as in the end of the Sec.~\ref{Sub_setup}.
%Each bath consists of fermion particles.
The next simple example is the heat engine exchanging not only energy but also particles between two baths, which may be used to model some electric cell, particle transportation, etc.
This is a first example of continuous scaling not based on i.i.d.~particles.
\subsubsection{General observation of the model}
Let the bath system $\mathcal{H}_{\mathrm{Baths}}$ be split into the cold bath $\mathcal{H}_{\mathrm{Bath},c}$ and the hot bath $\mathcal{H}_{\mathrm{Bath},h}$.
Each $\mathcal{H}_{\mathrm{Bath},b}$ $(b=c,h)$ has the Hamiltonian $H_{b,\lambda}$ and the number operator $N_{b,\lambda}$
with a scale parameter $\lambda$ as follows:
\begin{align}
 H_{b,\lambda}&=\sum_{\bm{n}=(n_1,n_2,\dots,n_{L_\lambda})} \left(\sum_{i=1}^{L_{b,\lambda}}E_{b,\lambda}(i)n_i \right) \ketbra{\bm{n}}{\bm{n}}\\
 N_{b,\lambda}&=\sum_{\bm{n}=(n_1,n_2,\dots,n_{L_\lambda})} \left(\sum_{i=1}^{L_{b,\lambda}}n_i \right) \ketbra{\bm{n}}{\bm{n}},
\end{align}
where $E_{b,\lambda}(i)$ is the $i$-th energy level, $L_{b,\lambda}$ is the number of levels of the Hamiltonian.
%Here, $\zah_{+,\nu_{0,\lambda}}$ denotes $\{0,1,\dots,\nu_{0,\lambda}\}$.
%$\nu_{0,\lambda}$ is the total number of the particle in the situation for boson, while $\nu_{0,\lambda}=1$ for fermion.
%Let the other system $(\mathcal{H}_1,H_{c,\lambda},N_{1,\lambda})$ be similarly defined.
The initial state is the grand canonical Gibbs state
with the initial generalized inverse temperature $\bm{\theta_0}=(\beta_c,\beta_h,-\beta_c\mu_c, -\beta_h\mu_h)$ with $\beta_c >\beta_h$:
\begin{align}
 &\tau_{\bm{\theta}_0}^{(\lambda)}\nonumber\\
 =&
 \frac{e^{-\beta_c H_{c,\lambda}+\beta_c\mu_c N_{c,\lambda}}}{\tr e^{-\beta_c H_{c,\lambda}+\beta_c\mu_c N_{c,\lambda}}}\otimes
 \frac{e^{-\beta_h H_{h,\lambda}+\beta_h\mu_h N_{h,\lambda}}}{\tr e^{-\beta_h H_{h,\lambda}+\beta_h\mu_h N_{h,\lambda}}},
\end{align}
where $\beta_b>0$ and $\mu_b$ are the inverse temperature and the chemical potential of $\mathcal{H}_{\mathrm{Bath},b}$ $(b=c,h)$, respectively.
Thus, each bath works as a heat and particle bath simultaneously.
%Since two bath systems are independent, it is enough to check the condition for respective system in order to apply the general theory.
Once the Assumption \ref{assumption} is verified, we have the following FGCB for the work (energy) extraction $\Delta W$
under
the endothermic heat $\Delta Q_{h,\lambda}=o(\lambda)$ from the hot bath
and the particle number $\Delta N_{b,\lambda}=o(\lambda)$ absorbed from the bath $\mathcal{H}_{\mathrm{Bath},b}$ $(b=c,h)$:
%and smallness of the boundary terms, for given appropriate size of
%heat $\Delta Q_{1,\lambda}$ absorbed from the bath $\mathcal{H}_1$ to the engine, and the particle number gain $\Delta N_{b,\lambda}$ of each bath $\mathcal{H}_b$ $(b=0,1)$,
%we obtain the achievable maximum work extraction from the baths to the battery
\begin{align}
 &\Delta W\nonumber\\
 \leq & \left(1-\frac{\beta_{h}}{\beta_{c}}\right)\Delta Q_{h,\lambda}
 +\mu_c\Delta N_{c,\lambda}+\frac{\beta_h}{\beta_{c}}\mu_h\Delta N_{h,\lambda}\nonumber\\
  &
 -C_{HH}\frac{\Delta Q_{h,\lambda}^2}{\lambda}
 -\sum_{b=c,h}C_{NN}^{b,b}\frac{\Delta N_{b,\lambda}^2}{\lambda}\nonumber\\
 &-C_{NN}^{c,h}\frac{\Delta N_{c,\lambda}\Delta N_{h,\lambda}}{\lambda}\nonumber\\
 &-\sum_{b=c,h}C_{HN}^{b}\frac{\Delta Q_{h,\lambda}\Delta N_{b,\lambda}}{\lambda}\nonumber\\
  &+o\left(\frac{\beta_h ^2\Delta Q_{h,\lambda}^2+\Delta N_{c,\lambda}^2+\Delta N_{h,\lambda}^2}{\lambda}\right),\label{FGCB_ex3}
\end{align}
where the signs of $\Delta Q_{h,\lambda}$ and $\Delta N_{b,\lambda}$ $(b=c,h)$ are taken positive if they are absorbed from the bath to the engine.
%We made the coefficients having the dimension of $\lambda^{-1}$, in \eqref{FGCB_ex3}.
The coefficients are given as
\begin{align}
 C_{HH}=& \frac{1}{2\beta_c}\left[
 \frac{\sigma_{N_h}^2}{\sigma_{H_h}^2\sigma_{N_h}^2-V_{HN}^{(h)2}}
   +\frac{\beta_h^2}{\beta_c^2}\frac{\sigma_{N_c}^2}{\sigma_{H_c}^2\sigma_{N_c}^2-V_{HN}^{(c)2}}
 \right]
 \label{CCa}\\
 C_{NN}^{h,h}=&
 \frac{1}{2\beta_c}\left[
 \frac{\sigma_{H_h}^2}{\sigma_{N_h}^2\sigma_{H_h}^2-V_{HN}^{(h)2}}
 +\frac{\beta_h^2\mu_h^2}{\beta_c^2}\frac{\sigma_{N_c}^2}{\sigma_{H_c}^2\sigma_{N_c}^2-V_{HN}^{(c)2}}
 \right]
 \\
 C_{NN}^{c,c}=&
 \frac{1}{2\beta_c}\left[
 \frac{\sigma_{H_c}^2}{\sigma_{N_c}^2\sigma_{H_c}^2-V_{HN}^{(c)2}}
 +\mu_c^2\frac{\sigma_{N_c}^2}{\sigma_{H_c}^2\sigma_{N_c}^2-V_{HN}^{(c)2}}
 \right.\nonumber\\
 &\hspace{10mm}\left.
 -2\mu_c\frac{V_{HN}^{(c)}}{\sigma_{H_c}^2\sigma_{N_c}^2-V_{HN}^{(c)2}}
 \right]
 \\
 C_{HN}^{h}=&
 -\frac{1}{\beta_c}\left[
 \frac{\beta_h^2\mu_h}{\beta_c^2}\frac{\sigma_{N_c}^2}{\sigma_{H_c}^2\sigma_{N_c}^2-V_{HN}^{(c)2}}
 +\frac{V_{HN}^{(h)}}{\sigma_{H_h}^2\sigma_{N_h}^2-V_{HN}^{(h)2}}
 \right]
 \\
 C_{HN}^{c}=&
 \frac{1}{\beta_c}\left[
 \frac{\beta_h}{\beta_c}\frac{V_{HN}^{(c)}}{\sigma_{H_c}^2\sigma_{N_c}^2-V_{HN}^{(c)2}}
 -\frac{\beta_h\mu_c}{\beta_c}\frac{\sigma_{N_c}^2}{\sigma_{H_c}^2\sigma_{N_c}^2-V_{HN}^{(c)2}}
 \right]
 \\
 C_{NN}^{c,h}=&
 \frac{1}{\beta_c}\left[
 \frac{\beta_h\mu_h\mu_c}{\beta_c}\frac{\sigma_{N_c}^2}{\sigma_{H_c}^2\sigma_{N_c}^2-V_{HN}^{(c)2}}\right.\nonumber\\
 &\hspace{10mm}\left.-\frac{\beta_h\mu_h}{\beta_c}\frac{V_{HN}^{(c)}}{\sigma_{H_c}^2\sigma_{N_c}^2-V_{HN}^{(c)2}}
 \right]
 ,\label{CCe}
\end{align}
where $\sigma_{H_b}^2$, $\sigma_{N_b}^2$ and $V_{HN}^{(b)}$ are
the respective asymptotic densities of the variance ${\rm Var}[A]:=\tr A^2\tau_{\theta_0}^{(\lambda)} -(\tr A\tau_{\theta_0}^{(\lambda)} )^2$
and covariance ${\rm Cov}[A,B]:=\tr AB\tau_{\theta_0}^{(\lambda)} -(\tr A\tau_{\theta_0}^{(\lambda)})(\tr B\tau_{\theta_0}^{(\lambda)})$ of each quantity defined as
$\sigma_{H_b(N_b)}^2:=\lim_{\lambda\rightarrow\infty}{\rm Var}[H_{b,\lambda}(N_b)]/\lambda$, $V_{HN}^{(b)}:=\lim_{\lambda\rightarrow\infty}{\rm Cov}[H_{b,\lambda},N_{b,\lambda}]/\lambda$ $(b=c,h)$.
Thus, we obtain the explicit form of dependence of the optimal performance on the fluctuation of the energy and the particle number as well as their correlation
in the coefficients of the finite-size effect.
%where $(J_{b,\lambda}^{ij})_{ij}$ is the inverse of the covariance matrix
% \begin{align}
%  &
% \begin{pmatrix}
%  \sigma_{H_{b,\lambda}}^2 & \text{Cov}(H_{b,\lambda},N_{b,\lambda})\\
%  \text{Cov}(H_{b,\lambda},N_{b,\lambda}) & \sigma_{N_{b,\lambda}}^2
% \end{pmatrix}
%  \nonumber\\
% :=&
% \begin{pmatrix}
%  \langle H_{b,\lambda}^2 \rangle_{\tau_{\theta_0}^{(\lambda)}}
%  -\langle H_{b,\lambda} \rangle_{\tau_{\theta_0}^{(\lambda)}}^2
%  &
%  \langle H_{b,\lambda}N_{b,\lambda} \rangle_{\tau_{\theta_0}^{(\lambda)}}
%  -\langle H_{b,\lambda} \rangle_{\tau_{\theta_0}^{(\lambda)}}\langle N_{b,\lambda} \rangle_{\tau_{\theta_0}^{(\lambda)}}\\
%  \langle H_{b,\lambda}N_{b,\lambda} \rangle_{\tau_{\theta_0}^{(\lambda)}}
%  -\langle H_{b,\lambda} \rangle_{\tau_{\theta_0}^{(\lambda)}}\langle N_{b,\lambda} \rangle_{\tau_{\theta_0}^{(\lambda)}}
%  &
%  \langle N_{b,\lambda}^2 \rangle_{\tau_{\theta_0}^{(\lambda)}}
%  -\langle N_{b,\lambda} \rangle_{\tau_{\theta_0}^{(\lambda)}}^2
% \end{pmatrix}
%  ,
% \end{align}
% which is given by
%  \begin{align}
%   &\left(J_{i,\lambda}^{kl}\right)_{kl}\nonumber\\
%   =&
%   \frac{1}{\sigma_{H_{b,\lambda}}^2\sigma_{N_{b,\lambda}}^2-\text{Cov}(H_{b,\lambda},N_{b,\lambda})^2}
%  \begin{pmatrix}
%   \sigma_{N_{b,\lambda}}^2 & -\text{Cov}(H_{b,\lambda},N_{b,\lambda})\\
%   -\text{Cov}(H_{b,\lambda},N_{b,\lambda}) & \sigma_{H_{b,\lambda}}^2
%  \end{pmatrix}
%   .
   %  \end{align}

   On the other hand, we have the following FGCB for the particle number extraction $\Delta N_{\rm tot}$ under the endothermic heat $\Delta Q_{b,\lambda}$ from $\mathcal{H}_{\mathrm{Bath},b}$ $(b=c,h)$ and the particle number $\Delta N_{h,\lambda}$ absorbed from one bath, say hot bath:
   \begin{align}
 &\Delta N_{\rm tot}\nonumber\\
 \leq & \left(1-\frac{\beta_{h}\mu_h}{\beta_{c}\mu_c}\right)\Delta N_{h,\lambda}
 +\mu_c^{-1}\Delta Q_{c,\lambda}+\frac{\beta_h}{\beta_{c}\mu_c}\Delta Q_{h,\lambda}\nonumber\\
  &
 -\tilde C_{HH}\frac{\Delta N_{h,\lambda}^2}{\lambda}
 -\sum_{b=c,h}\tilde C_{NN}^{b,b}\frac{\Delta Q_{b,\lambda}^2}{\lambda}\nonumber\\
 &-\tilde C_{NN}^{c,h}\frac{\Delta Q_{c,\lambda}\Delta Q_{h,\lambda}}{\lambda}
 -\sum_{b=c,h}\tilde C_{HN}^{b}\frac{\Delta N_{h,\lambda}\Delta Q_{b,\lambda}}{\lambda}\nonumber\\
  &+o\left(\frac{\Delta N_{h,\lambda}^2+\beta_c^2\Delta Q_{c,\lambda}^2+\beta_h^2\Delta Q_{h,\lambda}^2}{\lambda}\right),
\end{align}
where the coefficients are similarly calculated.
\subsubsection{A concrete model: an ideal Fermi gas inside the one dimensional well potential}
  As a concrete model, we consider an ideal Fermi gas.
  Let each bath $\mathcal{H}_{\mathrm{Bath},b}$ $(b=c,h)$ be composed of an ideal Fermi gas inside the infinite well potential
  \begin{align}
   V_{b,\lambda}(x)
   =
   \left\{
   \begin{array}{cc}
    0 & (x\in [0,\lambda l_b])\\
    \infty & (x\notin [0,\lambda l_b])
   \end{array}
   \right.
   ,
  \end{align}
  where $l_b$ is the length parameter to determine the rate of the size between two baths.
  $\lambda$ is a dimensionless scaling parameter.
  For simplicity, we set $l_b$ as the unit length for both baths.
  In this case, the energy eigenvalues of one particle is given by
  \begin{align}
   E_{b,\lambda}(i):=E_{\lambda}(i)
   =\frac{\hbar^2\pi^2 i^2}{2m \lambda^2}=:\frac{E_0}{\lambda^2}i^2
   \quad (i=1,2,\dots)
   ,
  \end{align}
  where $m$ is the mass of the particle.
  Moreover, we introduce a cut-off energy $E$ to this Hamiltonian such that
  $E_\lambda(i)\leq E$.
  That is because the dimension should be finite
  to apply our general theory, strictly speaking.
  Nevertheless, with large enough $E$, this toy model
  can be an approximation of the true square well potential.
  %which grows with the length $\lambda$.
  %This model may also relevant
  %for
  %a situation where higher energy level should be outside the bath for some reason.
  In this case, the number $L_\lambda$ of levels becomes finite, which is written as
  \begin{align}
   L_{\lambda}=\max_{\frac{E_0}{\lambda^2}i^2\leq E}i
   =\left\lfloor \sqrt{\frac{E}{E_0}}\lambda\right\rfloor.
  \end{align}
  Then, the free entropy of the bath $\mathcal{H}_{\mathrm{Bath},b}$ $(b=c,h)$ satisfies the asymptotic form
  \begin{align}
   &\phi_{b,\lambda}(\beta_b,\mu_b)\nonumber\\
   =&
   \log \sum_{(n_1,n_2,\dots,n_{L_\lambda})\in \{0,1\}^{L_{0,\lambda}}}
   \prod_{i=1}^{L_{\lambda}}e^{\beta_b(-E_{\lambda}(i)+\mu_b)n_i}\nonumber\\   
   =&
   \frac{\lambda}{2\sqrt{E_0}}
   \int_0^{E} \epsilon^{-\frac{1}{2}}\log (1+e^{\beta_b\mu_b-\beta_b \epsilon})d\epsilon +\order{1}\nonumber\\
   =&
   \int_0^{E} \frac{\lambda\sqrt{2m}}{2\pi\hbar}\epsilon^{-\frac{1}{2}}\log (1+e^{\beta_b\mu_b-\beta_b \epsilon})d\epsilon +\order{1}\nonumber\\
   =:&
   \lambda \phi_b(\beta_b,\mu_b)+\order{1}.
  \end{align}
  Thus, Assumption \ref{assumption} is satisfied with smaller deviation from the extensivity than $\order{\lambda^{\frac{1}{2}}}$.
  Moreover, since the relations $\|H_{b,\lambda}\|\leq E L_\lambda=\order{\lambda}$ and $\|N_{b,\lambda}\|=L_\lambda=\order{\lambda}$ also hold,
  all the conditions for the achievability for Theorem \ref{achieve_thm} are verified. Hence this is indeed an example where the maximum work extraction in FGCB (\ref{FGCB_ex3}) is achievable.

  Now, we further calculate the second order coefficients \eqref{CCa}-\eqref{CCe} in FGCB \eqref{FGCB_ex3} in low temperature approximation.
  Supposing that $E$ is sufficiently large, we regard $E$ as $\infty$.
  The asymptotic density of the energy $\epsilon_b$ and the particle number $n_b$ are given as
  \begin{align}
   \epsilon_b&=\frac{\sqrt{2m}}{2\pi\hbar}\int_{0}^{\infty}\frac{\epsilon^{\frac{1}{2}}}{e^{\beta_b \epsilon -\beta_b\mu_b}+1}d\epsilon\\
   n_b&=\frac{\sqrt{2m}}{2\pi\hbar}\int_{0}^{\infty}\frac{\epsilon^{-\frac{1}{2}}}{e^{\beta_b \epsilon -\beta_b\mu_b}}d\epsilon.
  \end{align}
  For sufficiently low temperature where $1\ll\mu_b\beta_b$ holds,
  $I:=\int_{0}^{\infty}\frac{F(\epsilon)}{e^{\beta_b \epsilon -\beta_b\mu_b}} d\epsilon$ is approximated as
  \begin{align}
   I\approx \int_{0}^{\mu_b}F(\epsilon)d\epsilon+\frac{\pi^2}{6}\beta_b^{-2}F'(\mu_b).
  \end{align}
  Then, $\epsilon_b$ and $n_b$ are approximated as
  \begin{align}
   \epsilon_b&=\frac{\sqrt{2m}}{2\pi\hbar}\mu_b^{-\frac{1}{2}}\left[\frac{2}{3}\mu_b^2+\frac{\pi^2}{12}\beta_b^{-2}\right]\\
   n_b&=\frac{\sqrt{2m}}{2\pi\hbar}\mu_b^{\frac{1}{2}}\left[2-\frac{\pi^2}{12}\beta_b^{-2}\mu_b^{-2}\right].
  \end{align}
  Calculating their derivatives, we obtain the variances and correlation as
  \begin{align}
   \sigma_{H_b}^2&=\frac{\sqrt{2m}}{2\pi\hbar}\frac{8\beta_b^2\mu_b^2+\pi^2}{8\beta_b^3\mu_b^{\frac{1}{2}}}\label{shb}\\
   \sigma_{N_b}^2&=\frac{\sqrt{2m}}{2\pi\hbar}\frac{8\beta_b^2\mu_b^2+\pi^2}{8\beta_b^3\mu_b^{\frac{5}{2}}}\\
   V_{HN}^{(b)}&=\frac{\sqrt{2m}}{2\pi\hbar}\frac{24\beta_b^2\mu_b^2-\pi^2}{24\beta_b^3\mu_b^{\frac{3}{2}}}.\label{vhn}
  \end{align}
  We should consider the finite-size effect under
  the fixed first order coefficients, namely, we fix $r:=\beta_h/\beta_c$, $\mu_h$ and $\mu_c$.
  Then, we obtain the second order coefficients
  as follows:
  \begin{align}
   C_{HH}=&\frac{9\hbar\beta_c^2}{\sqrt{2m}\pi}\left[r^2\mu_c^{\frac{1}{2}}\frac{8\beta_c^2\mu_c^2+\pi^2}{24\beta_c^2\mu_c^2+\pi^2}\right.\nonumber\\
   &\hspace{15mm}\left.
   +r^3\mu_h^{\frac{1}{2}}\frac{8\beta_c^2r^2\mu_h^2+\pi^2}{24\beta_c^2r^2\mu_h^2+\pi^2}\right]\label{CHHp}\\
   C_{NN}^{h,h}=&\frac{9\hbar\beta_c^2}{\sqrt{2m}\pi}\left[r^2\mu_h^2\mu_c^{\frac{1}{2}}\frac{8\beta_c^2\mu_c^2+\pi^2}{24\beta_c^2\mu_c^2+\pi^2}\right.\nonumber\\
   &\hspace{15mm}\left.
   +r^3\mu_h^{\frac{5}{2}}\frac{8\beta_c^2r^2\mu_h^2+\pi^2}{24\beta_c^2r^2\mu_h^2+\pi^2}\right]\\
   C_{NN}^{c,c}=&\frac{24\pi\hbar\mu_c^{\frac{5}{2}}\beta_c^2}{\sqrt{2m}(24\beta_c^2\mu_c^2+\pi^2)}\\
   C_{NN}^{c,h}=&\frac{24\pi\hbar r\mu_h\mu_c^{\frac{3}{2}}\beta_c^2}{\sqrt{2m}(24\beta_c^2\mu_c^2+\pi^2)}\\
   C_{HN}^{h}=&-\frac{6\hbar\beta_c^2}{\sqrt{2m}\pi}\left[3r^2\mu_h\mu_c^{\frac{1}{2}}\frac{8\beta_c^2\mu_c^2+\pi^2}{24\beta_c^2\mu_c^2+\pi^2}\right.\nonumber\\
   &\hspace{18mm}\left.
   +r^3\mu_h^{\frac{3}{2}}\frac{24\beta_c^2r^2\mu_h^2-\pi^2}{24\beta_c^2r^2\mu_h^2+\pi^2}\right]\\
   C_{HN}^{c}=&-\frac{24\pi\hbar r\mu_c^{\frac{3}{2}}\beta_c^2}{\sqrt{2m}(24\beta_c^2\mu_c^2+\pi^2)}.\label{CHNp}
  \end{align}
  According to these coefficients, it is remarkable that
  the optimal performance with finite-size effects explicitly depends on
  the mass $m$ of the particle.
  It implies that heavier particles have better performance for heat engines.
  According to the interpretation of the finite-size effect as mentioned shortly after Theorem \ref{Thm_FGCB},
  this feature implies that the performance is gained because the large mass leads to small response of the baths due to the large inertia.

  In addition, even though we fix the first order coefficients, the second order coefficients \eqref{CHHp}-\eqref{CHNp} depend on the inverse temperature $\beta_c$.
  Their expressions imply that the small $\beta_c$ (high temperature) gives the better performance.
  This behavior is also consistent with the response of the inverse temperature as follows.
  The heat capacities get larger for the higher temperature as seen from the expressions \eqref{shb}-\eqref{vhn}.
  Hence, the higher the temperature gets,
  the smaller the response of the inverse temperature to the variation of the conserved quantities becomes.
  %is explicitly included in , reflecting that 
 %\section{Comparison with the prior works}\label{sec_discussion}
  
  \section{Conclusion}\label{sec_conclusion}
  We have revealed the effects of the finiteness of the baths with arbitrary multiple conserved quantities on the optimal performance of the generalized heat engine.
  We have extended the scaling to the generic form, imposing the extensivity.
  Under this generic scaling,
  we have derived FGCB as a fine-grained upper bound on the performance of generalized heat engines.
  FGCB includes the second order terms of order $\order{\frac{\|\vb{Q}_{\lambda}\|^2}{\lambda}}$ as the finite-size effects.
  Contrary to the thermodynamic limit regime,
  the coefficients of this finite-size effects terms reflect the canonical correlations between the multiple conserved quantities of the baths as well as the generalized inverse temperatures.
  In particular, for the case without correlation between different baths,
  large fluctuation and small correlation of the quantities enlarge the optimal performance.
  %The second example in Sec.~\ref{sec_example} of the spin-$\frac{1}{2}$ bath explicitly illustrates this feature.
  %On the other hand, the third example of the continuous scaling, the volume of the container including identical particles.

  FGCB has been given for the implicit-battery formulation for wide applicability of the theory.
  However, to show the achievability of FGCB,
  %check ' in physical sense ' or so? 
  we should construct a protocol under
  the explicit-battery formulation.
  We have imposed independence of the state of the battery
  on the explicit-battery operations
  as the
  no-cheating condition
  to guarantee that the battery really works only as a storage of extracted work,
  but not as an entropy sink.
  In this sense, the energy transfer to the battery is indeed work-like.
  Under the conservation laws, the cyclicity of the working body, and the no-cheating condition,
  we have explicitly constructed a protocol with an explicit battery.
  Our protocol has been given by a permutation of the basis of the baths, which works independently of the detail of the system.
  %This protocol is completed in finite steps, and does not need any infinite sequence of operations or the thermodynamic limit.
  %%%%%%%%%%%
  %Furthermore, our protocol does not depend on also the state of the working body.
  Though the equality in FGCB is attained by the thermal state at the ideal final inverse temperature $\bm\theta_{\lambda}$ which is determined by the conditions (\ref{tl_1})-(\ref{tl_4}), this state cannot necessarily be obtained from the initial thermal state through the operations
  in finite-size bath.
  Instead, the resultant state of our protocol is very close to this ideal thermal state.
  The closeness in terms of the relative entropy shows that our protocol indeed achieves the equality in FGCB up to $o\left(\frac{\|\vb{Q}_{\lambda}\|^2}{\lambda}\right)$, which is negligible in our regime.
  We have shown this estimation
  by making use of the information geometric structure.
  One of the technical key points is the extension of the central limit theorem, which is needed for the strong large deviation estimation for our generic scaling, whose detail is given in Appendices \ref{app_gCLT} and \ref{app_relative_ent}.
  In our protocol, the dynamics on the working body $\mathcal{H}_C$ is trivial, and completely split from the baths and the battery.
  Thus, no catalytic effects work in this asymptotically optimal protocol, which means
  that the improvement of the optimal performance by catalytic effects is of order $o\left(\frac{\|\vb{Q}_{\lambda}\|^2}{\lambda}\right)$.
  However, note that the working body $\mathcal{H}_C$ should be needed to physically realize the dynamics
  even if the resultant map per one cycle is trivial like our protocol.
  
  Strictly speaking, we have imposed additional conditions in Theorem \ref{achieve_thm}.
  One is on the order of the norm of each quantity as \eqref{mat_norm_0}.
  Since this condition is needed just for a technical reason,
  it is possibly removed in future works.
  The others are the conditions \eqref{str_extensive_1} and \eqref{str_extensive_2} that the order of the deviation from the extensivity \eqref{extensive_1}, \eqref{extensive_2}
  is sufficiently small as $\order{\lambda^{\alpha}}$ with $\alpha<\frac{1}{2}$.
  This is possibly more essential in a physical sense, since
  great deal of the deviation from the extensivity of each quantity possibly degrades the performance of the engine.
  Further investigation is needed to reveal such an effects caused by the deviation from the extensivity on the performance of protocols.
  %%%%%%%%%%%%
  In addition, to verify that
  our protocol achieves the optimal performance in our analysis,
  it is also needed that
  we run the engine with the heat of the order
  $\lambda^{\frac{5}{8}}\ll\|\vb{Q}_{\lambda}\|\ll \lambda$.
  %%%%%%
  This condition is required to verify that the relative entropy between
  the ideal final thermal state and the final state of our protocol is small enough.
  %Though this condition may imply that
%  small enough order of endotherm so that the disturbance of the baths is sufficiently small in reflection of boundary effects, its meaning is still unclear.
   It is a future work to further investigate the relation between the amount of generalized heat and the scale.
  It is an interesting feature that the quality of the protocol may alter according to its amount of the generalized heat.
  %no-cheating condition0
  %catalytic effect, fluctuations 2
  %optimality info geo, st LD (CLT extension) 1
  Furthermore, it remains to verify the relation between the work fluctuation and the performance, though this is also important in order to investigate the realistic usefulness of the heat engine \cite{PhysRevLett.115.260601,PhysRevLett.118.100602}.

  Our protocols have similar forms for both commutative and non-commutative cases.
  Nevertheless, only the average conservation is satisfied for non-commutative case, though the strict conservation law is satisfied for the commutative case.
  %state independence
  While the validity of the average conservation law for the protocol may depend on the initial state in general,
  our protocol satisfies the conservation law regardless of the initial state of the battery and working body,
  just depending on the baths.
  %Although protocols under only the average conservation law indefinitely uses coherence,
  Although our protocol for the non-commutative case indefinitely uses coherence,
  it may be revived if some resource of coherence is appropriately included in our operation as pointed out by
  \cite{PhysRevX.5.021001}.
  Giving protocols for multiple non-commutative quantities under strictly conservation law is an important open but challenging problem.

  Finally, we have applied our general results to some examples. 1D Ising spin chain was a first example for the non-i.i.d. scaling with asymptotic extensivity. We have shown that the coupling constant of the spin chain affects the optimal performance for the finite baths. Especially, the best value of the coupling constant gives the largest optimal performance. On the other hand, even the finite-size effect is independent of whether the spin chain is ferromagnetic or anti-ferromagnetic.
  As for an example of multiple conserved quantities with non-i.i.d.~scaling, we have considered a heat engine with an ideal gas exchanging particles.
  Though it is so famous canonical example, it was for the first time to explicitly calculate the coefficients of the finite-size-effect terms in the optimal performance of that heat engine.
  For an ideal Fermi gas inside a well potential,
  we have found that these coefficients explicitly depend on the mass of the particle, which is again quite different from the nature in thermodynamic limit.
  This fact implies that heavier particles have better performance for heat engines.
  From these examples, we have already seen that
  the finite-size effect depends on
  the peculiar parameters for each model such as the coupling constant and mass in various ways.
  It is an important future work to
  investigate the finite-size effect for
  %our general theory to
  more practical heat engines in detail, and to compare it with our general result.

  Our protocol may be hard to experimentally realize since it involves in microscopic control of the baths' basis.
  Thus, a realistic protocol should be considered as a future work.
  Recently, a realization of thermal operations (with infinite baths) by realistic operations was studied \cite{1511.06553}.
  Though that result cannot be directly applied to the finite-size regime, our protocol may be realized by some combination of realistic operations.
  Then, our model may be applicable to an electric battery, or biological systems in a realistic mesoscopic scale.

  Since our analysis is based on the asymptotic analysis of finite-size systems,
the obtained results clarify the optimal performance of 
mesoscopic systems.
We consider that our analysis is a first step to universal understanding of quantum thermodynamics in various scale.
 % This analysis can be one of the first steps to universal understanding of quantum thermodynamics in various scale.
   \begin{acknowledgments}
    The authors would like to thank Dr. Hiroyasu Tajima for his helpful comments.
   KI was supported by JSPS KAKENHI Grant Number JP16J03549.
   MH was supported in part by a MEXT Grant-in-Aid for Scientific Research (B) No. 16KT0017,
Fund for the Promotion of Joint International Research (Fostering Joint International Research) No. 15KK0007,
the Okawa Research Grant, and Kayamori Foundation of Informational Science Advancement.
   \end{acknowledgments}

  \widetext
  \appendix

  \section{Information geometry for density matrices}\label{app_bregman}
  In this section, we review the detail of the information geometric analysis in the proof of Lemma \ref{expect_rel} in Sec. \ref{sub_ach_FGCB} including non-commutative case.
  At first, we give a brief review on the information geometry based on the theory of the Bregman divergence.
  This theory gives an abstract framework for the information geometry.
  Then, we can use the results from this theory just by applying it to individual cases.
  Next, we do so for our case.
  \subsection{The Bregman divergence}\label{app_breg_pyth}
  We review an abstract framework of information geometry in terms of the Bregman divergence \cite{Hayashi:2017aa}.
  The meaning of the following abstraction will become clear when we apply this theory to our state family in the next subsection.
  
  We consider a twice-differentiable strictly convex function $\mu$ defined on an open subset $\Theta$ of $\mathbb R^D$.
  The set $\Theta$ usually corresponds to the parameter space of the states in consideration.
  Then, we define the Bregman divergence of $\mu$ as
  \begin{align}
   D^{\mu}(\bar{\theta}\|\theta):=\sum_{k=1}^{D}\frac{\partial\mu}{\partial \theta^k}(\bar{\theta})(\bar{\theta}^k-\theta^k)-\mu(\bar{\theta})+\mu(\theta).
  \end{align}
  The Bregman divergence is a `distance measure' of the abstract parameter space $\Theta$ induced by $\mu$, which is called a potential function.
  It is an advantage of the abstract theory that once we find such a potential function $\mu$, we can apply all the results based on $\mu$.

  Since $\mu$ is strictly convex, $\theta \mapsto \eta(\theta):=\nabla\mu(\theta)$ is one-to-one.
  Thus, $\eta$ gives another parametrization.
  Because $\mu$ plays the role of the free entropy,
  its derivatives $\eta_k=\frac{\partial\mu}{\partial \theta^k}$ correspond to the expectation values.
  The Bregman divergence can be expressed by this dual parameter.
  To do so, we observe the Legendre transformation $\nu$ of $\mu$
  \begin{align}
   \nu(\eta):=\max_{\tilde{\theta}}\left[\sum_{k}\eta_k \tilde{\theta}^k-\mu(\tilde{\theta})\right].
  \end{align}
  Then, the Bregman divergence $D^{\nu}(\eta\|\bar{\eta})$ for $\nu$ is also defined since $\nu$ is also a strictly convex function of $\eta$.
  When $\eta = \nabla \mu(\theta)$, we have
  \begin{align}
   \nu(\eta)=\sum_{k}\eta_k \theta^k - \mu(\theta)
  \end{align}
  by the definition.
  Using this relation for $\eta=\nabla \mu(\theta)$ and $\bar{\eta}=\nabla \mu(\bar{\theta})$,
  we obtain
  \begin{align}
   D^{\mu}(\bar{\theta}\|\theta)
   =&\sum_{k}\bar{\eta}_k (\bar{\theta})(\bar{\theta}^k-\theta^k)-\mu(\bar{\theta})+\mu(\theta)
   =\sum_{k}\theta^k (\eta_k - \bar{\eta}_k)+ \left(\sum_{k}\bar{\eta}_k\bar{\theta}^k-\mu(\bar{\theta})\right)
   -\left(\sum_{k}\eta_k \theta^k - \mu(\theta)\right)\nonumber\\
   =&\sum_{k}\theta^k (\eta_k - \bar{\eta}_k) + \nu(\bar{\eta})-\nu(\eta)=D^{\nu}(\eta\|\bar{\eta})\nonumber\\
   =&\int_{0}^{1}\sum_{k,j}(\eta_k - \bar{\eta}_k)(\eta_j - \bar{\eta}_j)
   \frac{\partial^2 \nu}{\partial\eta_k\partial\eta_j}(\bar{\eta}+(\bar{\eta}-\eta)t)t dt,\label{apprel}
  \end{align}
  where the last line follows from the Taylor's formula.
  It should be remarked that the matrix $(\frac{\partial^2 \nu}{\partial\eta_k\partial\eta_j}(\eta))_{kj}$
  is verified to be the inverse of $(\frac{\partial^2 \mu}{\partial\mu^k\partial\theta^j}(\theta))_{kj}$ from the chain rule and the inverse relation $\theta^k=\frac{\partial\nu}{\partial\eta_k}(\eta)$.

  With a point $\theta'\in\Theta$ and $l$ linearly independent vectors $v_1,\dots, v_l \in \mathbb R^D$,
  an $l$-dimensional flat $\mathcal E=\{\theta\in\Theta | \theta=\theta'+\sum_{j=1}^l a^j v_j, (a^1,a^2,\dots,a^l)\in \mathbb R^l\}$ is defined.
  Such a flat $\mathcal E$ is called an exponential subfamily of $\Theta$ whose generator is $\{v_1,\dots, v_l\}$.
  As the name indicates, this is an abstraction of exponential family, i.e.~a family of generalized thermal states.
  As a `dual flat' of the exponential subfamily $\mathcal E$,
  $\mathcal M=\{\theta\in\Theta| b_j = \sum_{i=1}^D v_j^i \eta_i (\theta) \  (j=1,\dots,l)\}$ with some fixed real numbers $b_1,\dots,b_l$
  is called a mixture subfamily of $\Theta$ whose generator is $\{v_1,\dots, v_l\}$.
  The definition of a mixture subfamily means that $\mathcal M$ is a flat with respect to the dual parameter $\eta$.
  Hence, $\mathcal M$ corresponds to the state family with fixed expectation values.
  Then, the following Pythagorean theorem \cite{Amari:2000aa} for the Bregman divergence holds:
  \begin{proposition}[Amari \cite{5290302}]\label{prop_pythagorean}
   Let $\mathcal M$ be an mixture subfamily of $\Theta$ whose generator is $\{v_1,\dots, v_l\}$.
   For an arbitrary point $\theta\in\Theta$,
   there exists a unique intersection $\theta^*$ between $\mathcal M$ and the exponential subfamily $\mathcal E$ containing $\theta$ with the same generator $\{v_1,\dots, v_l\}$.
   This $\theta^*$ satisfies the following:
   \begin{enumerate}
    \item For any point $\theta'\in \mathcal M$, $D^{\mu}(\theta'\|\theta)=D^{\mu}(\theta'\|\theta^*)+D^{\mu}(\theta^*\|\theta)$ holds.
    \item $\theta^*=\arg\min_{\theta'\in \mathcal M} D^{\mu}(\theta'\|\theta)$.
   \end{enumerate}
  \end{proposition}

  \subsection{Application of the Pythagorean theorem to the state family}\label{app_pyth}
  Now, we apply the above abstract theory of the Bregman divergence to our situation.
  First of all, we parametrize all of the full-rank states of $\mathcal {H}_{\mathrm{Baths}}$ as follows.
  Since the set of all Hermitian matrices on $\mathcal {H}_{\mathrm{Baths}}$ can be seen as a real vector space whose dimension is $D+1:=d_{\lambda}(d_\lambda+1)/2$, there exists a basis $\{E_1,E_2,\dots,E_{D+1}\}$
  , where we omit the label $\lambda$ on $D$ for simplicity of the notation.
  Because the observables $A_{i,\lambda}$, $B_{i,\lambda}$ $(i=1,2)$ of the baths and the identity matrix $I$ are linearly independent Hermitian matrices, we can take the basis $\{E_1,E_2,\dots,E_D\}$ such that $E_1=A_{1,\lambda}$, $E_2=A_{2,\lambda}$, $E_3=B_{1,\lambda}$, $E_4=B_{2,\lambda}$, and $E_{D+1}=I$.
  %An example of a basis is $\{E_{ii}, E_{i,j}+E_{j,i}, \mathrm{i}(E_{i,j}+E_{j,i})\  (i,j=1,2,\dots,D)\}$, where $E_{i,j}$ denotes the $(i,j)$-matrix unit.
  Then, the parametrization $\exp(\sum_{i=1}^{D+1}\xi^i E_i)/\tr \exp(\sum_{i=1}^{D+1}\xi^i E_i)$ of the states by $(\xi^1,\xi^2,\dots,\xi^{D+1})\in \mathbb R^{D+1}$ runs all the full-rank states $\rho$ since $\log \rho$ is Hermitian, and $\sum_{i=1}^{D+1}\xi^i E_i$ runs all the Hermitian matrices.
  However, this parametrization is still redundant in the sense that for any $a\in \mathbb R$, $(\xi_1,\dots,\xi_D, a)$ corresponds to the same state
  $\rho(\xi_1,\dots,\xi_D):=\exp(\sum_{i=1}^{D}\xi^i E_i)/\tr \exp(\sum_{i=1}^{D}\xi^i E_i)$ since
  \begin{align}
  \frac{\exp(\sum_{i=1}^{D}\xi^i E_i+aI)}{\tr \exp(\sum_{i=1}^{D}\xi^i E_i+aI)}=\frac{e^a\exp(\sum_{i=1}^{D}\xi^i E_i)}{e^a\tr \exp(\sum_{i=1}^{D}\xi^i E_i)}=\rho(\xi_1,\dots,\xi_D).
  \end{align}
  Hence, we employ the parametrization $\rho(\xi)=\exp(\sum_{i=1}^{D}\xi^i E_i)/\tr \exp(\sum_{i=1}^{D}\xi^i E_i)$ by $\xi=(\xi_1,\dots,\xi_D)\in\mathbb R^D$, so that the parameter space is $\Theta=\mathbb R^D$.
  The potential function is $\mu(\xi):=\log \tr \exp(\sum_{i=1}^{D}\xi^i E_i)$.
  Indeed, it is a twice-differentiable strictly convex function, which can be verified by observing that its Hessian matrix $(\frac{\partial^2\mu}{\partial\xi^i\partial\xi^j}(\xi))_{ij}$ is positive definite as follows.
  The Hessian matrix is equal to the matrix $(K_{i,j}(\xi))_{ij}$ composed of the canonical correlations
  \begin{align}
   K_{i,j}(\xi):=\int_{0}^{1}ds\;\tr \rho(\xi)^{1-s} (E_i-\eta_i(\xi)) \rho(\xi)^s (E_j -\eta_j(\xi))
  \end{align}
  between $E_i$ and $E_j$, where $\eta_i(\xi):=\frac{\partial \mu}{\partial \xi^i}(\xi)$ is equal to the expectation value $\tr \rho(\xi) E_i$ of $E_i$.
  Thus, it is sufficient to show the positivity of $(K_{i,j}(\xi))_{ij}$.
  To do so, we firstly observe that the canonical correlation is a positive definite inner product:
  \begin{lemma}\label{lem_can_inner}
   Let $\rho$ be a state with full-rank.
   Then, for any matrix $X$, we have $\int_{0}^{1}ds \ \tr \rho^{1-s} X \rho^s X\geq 0$.
   In addition, $\int_{0}^{1}ds \ \tr \rho^{1-s} X \rho^s X= 0$ if and only if $X=0$.
  \end{lemma}
  \begin{proof}
    Using the commutativity inside of the trace, we obtain
   \begin{align}
    \tr \rho^{1-s} X \rho^s X
    =\tr \rho^{\frac{1-s}{2}}\rho^{\frac{1-s}{2}} X \rho^{\frac{s}{2}} \rho^{\frac{s}{2}} X
    =\tr (\rho^{\frac{1-s}{2}} X \rho^{\frac{s}{2}})(\rho^{\frac{s}{2}} X \rho^{\frac{1-s}{2}})
    =\tr (\rho^{\frac{s}{2}} X \rho^{\frac{1-s}{2}})^\dagger(\rho^{\frac{s}{2}} X \rho^{\frac{1-s}{2}})
    \geq 0
   \end{align}
   for any $0<s<1$.
   If $\tr (\rho^{\frac{s}{2}} X \rho^{\frac{1-s}{2}})^\dagger(\rho^{\frac{s}{2}} X \rho^{\frac{1-s}{2}})=0$, then $\rho^{\frac{s}{2}} X \rho^{\frac{1-s}{2}}=0$ holds.
   Since $\rho$ is invertible, $\rho^t$ $(0<t<1)$ is also.
   Then, $X=\rho^{-\frac{s}{2}}\rho^{\frac{s}{2}} X \rho^{\frac{1-s}{2}}\rho^{-\frac{1-s}{2}}=0$.
  \end{proof}
  Then, we show the positivity:
  \begin{lemma}
   $(K_{i,j}(\xi))_{ij}$ is positive definite for any $\xi\in\mathbb R^D$.
  \end{lemma}
  \begin{proof}
   For any vector $(a^1,\dots,a^D)\in\mathbb R^D$, we have
   \begin{align}
    \sum_{i,j=1}^D a^i K_{i,j}(\xi) a^j
    =\int_{0}^{1}ds \ \tr \rho(\xi)^{1-s} X \rho(\xi)^s X,
   \end{align}
   where $X=\sum_{i=1}^{D} a^i (E_i-\eta_i(\xi))$.
   Hence, $\sum_{i,j=1}^D a^i K_{i,j}(\xi) a^j\geq 0$ follows
   from Lemma \ref{lem_can_inner}.
   If $\sum_{i,j=1}^D a^i K_{i,j}(\xi) a^j= 0$, $\sum_{i=1}^{D} a^i (E_i-\eta_i(\xi))=0$ holds again by Lemma \ref{lem_can_inner}.
   Then, since $E_i$ $(i=1,\dots,D)$ and $I$ are linearly independent,
   $(a^1,\dots,a^D)=0$ follows from the expression
   \begin{align}
    \sum_{i=1}^{D}a^i E_i - \left[\sum_{i=1}^{D}a^i \eta_i(\xi)\right]I=0.
   \end{align}
   Thus, $(K_{i,j}(\xi))_{ij}$ is a positive definite matrix.
  \end{proof}
  Thus, $\mu(\xi)$ is verified to be strictly convex.
  The Bregman divergence associated with $\mu(\xi)$ is nothing but the relative entropy as follows:
  \begin{align}
   D^{\mu}(\bar{\xi}\|\xi)&=\sum_{k=1}^{D}\frac{\partial\mu}{\partial \xi^k}(\bar{\xi})(\bar{\xi}^k-\xi^k)-\mu(\bar{\xi})+\mu(\xi)
   =\sum_{k=1}^{D}\tr \rho(\bar{\xi}) E_k (\bar{\xi}^k-\xi^k)-\mu(\bar{\xi})+\mu(\xi)\nonumber\\
   =&\tr \rho (\bar{\xi})(\log \rho(\bar{\xi})-\log \rho (\xi))
   =D(\rho (\bar{\xi})\|\rho (\xi)).
  \end{align}

  The exponential subfamily $\mathcal E:=\{\xi\in\mathbb R^D|\xi= \sum_{i=1}^4 \theta^i v_i,\  \bm{\theta}=(\theta^1,\theta^2,\theta^3,\theta^4)\in\mathbb R^4\}$ with its generator $v_1=(1,0,\dots,0)$, $v_2=(0,1,0,\dots,0)$, $v_3=(0,0,1,0,\dots,0)$, $v_4=(0,0,0,1,0,\dots,0)$
  corresponds to the exponential family $\mathcal E_S:= \{\tau_{\bm{\theta}}^{(\lambda)}|\bm{\theta}\in\mathbb R^4\}$ of the thermal states by observing
  \begin{align}
   \tau_{\bm{\theta}}^{(\lambda)}
   =\frac{\exp[\sum_{i=1}^{2}(\theta^{i}A_{i,\lambda}+\theta^{i+2}B_{i,\lambda})]}{\tr \exp[\sum_{i=1}^{2}(\theta^{i}A_{i,\lambda}+\theta^{i+2}B_{i,\lambda})]}
   =\frac{\exp[\sum_{i=1}^{4}\theta^{i}E_i]}{\tr \exp[\sum_{i=1}^{4}\theta^{i}E_i]}
   =\frac{\exp[\sum_{k=1}^{D}\sum_{i=1}^{4}\theta^{i}v_i^k E_k]}{\tr \exp[\sum_{k=1}^{D}\sum_{i=1}^{4}\theta^{i}v_i^k E_k]}
   =\rho\left(\sum_{i=1}^{4}\theta^{i}v_i\right).
  \end{align}
  On the other hand, the mixture subfamily $\mathcal M:=\{\xi\in\mathbb R^D| b_j = \sum_{k=1}^{D} v_j^k \eta_k (\xi)\  (j=1,2,3,4)\}$ with the same generator $v_1,v_2,v_3,v_4$ corresponds to the state family
  $\mathcal M_S:= \{\rho > 0| \tr \rho E_j = b_j,\  (j=1,2,3,4)\}$ whose expectation values of $A_{i,\lambda}$ and $B_{i,\lambda}$ are fixed
  because
  \begin{align}
   b_j=\sum_{k=1}^{D} v_j^k \eta_k (\xi)
   =\sum_{k=1}^{D} v_j^k \tr \rho(\xi) E_k
   =\tr \rho(\xi) \sum_{k=1}^{D} v_j^k E_k
   =\tr \rho(\xi) E_j\quad (j=1,2,3,4).
  \end{align}
  Especially, for an arbitrary full-rank state $\rho$,
  the mixture subfamily $\mathcal M$ with $b_j=\tr \rho E_j$ corresponds to the state family whose expectation values of $E_j$ are shared with $\rho$.
  We denote the corresponding state family of $\mathcal M$ by $\mathcal M_S(\rho)$.
  %, since there exists the corresponding parameter $\xi\in\mathbb R^D$ such that $\rho=\rho(\xi)$, 
  Then, applying Proposition \ref{prop_pythagorean} to $\mathcal M$ and $\mathcal E$ in terms of our Bregman divergence, relative entropy, we obtain the desired Pythagorean theorem for our situation:
  \begin{lemma}\label{pyth}
   For an arbitrary full-rank state $\rho$,
   there exists a unique thermal state $\tau_{\bm{\theta}^*}^{\lambda}\in \mathcal E_S$ such that $\tau_{\bm{\theta}^*}^{\lambda}\in \mathcal M_S(\rho)$.
   Moreover, for an arbitrary thermal state $\tau_{\bm{\theta}}^{\lambda}\in \mathcal E_S$, we have
   \begin{align}
    D(\rho\|\tau_{\bm{\theta}}^{\lambda})= D(\rho\|\tau_{\bm{\theta}^*}^{\lambda})+D(\tau_{\bm{\theta}^*}^{\lambda}\|\tau_{\bm{\theta}}^{\lambda}).
   \end{align}
  \end{lemma}
  Notice that Lemma \ref{pyth} is valid for both the non-commutative and commutative $A_{i,\lambda}$ and $B_{i,\lambda}$.

  Furthermore, the thermal states $\tau_{\bm{\theta}}^{(\lambda)}$ can be also seen to be a state family parametrized by the generalized inverse temperature $\bm{\theta}$.
  The relative entropy $D(\tau_{\bar{\bm{\theta}}}^{(\lambda)}\|\tau_{\bm{\theta}}^{(\lambda)})$
  is again equal to the Bregman divergence associated with the free entropy $\phi_{\lambda}(\bm{\theta})$ as the strictly convex function on the parameter.
  Then, applying \eqref{apprel} to this Bregman divergence,
  we obtain
  \begin{align}
   D(\tau_{\bm{\xi}}^{(\lambda)}\|\tau_{\bm{\theta}}^{(\lambda)})
 =\int_{0}^{1}\sum_{ij}(\eta_{\lambda,i}(\bm{\theta})-\eta_{\lambda,i}(\bm{\xi}))(\eta_{\lambda,j}(\bm{\theta})-\eta_{\lambda,j}(\bm{\xi}))
 J^{ij}_{\lambda}(\vb{s}_{\lambda}(t))tdt\label{apprel_re}
  \end{align}
  for any generalized inverse temperatures $\bm{\xi}$ and $\bm{\theta}$,
  where $\vb{s}_{\lambda}(t)$ is the generalized inverse temperature satisfying
  $\bm{\eta}_{\lambda}(\vb{s}_{\lambda}(t))=t\bm{\eta}_{\lambda}(\bm{\theta})+(1-t)\bm{\eta}_{\lambda}(\bm{\xi})$.

  \section{A generalization of the central limit theorem}\label{app_gCLT}
  In this section, we show the following generalization of the central limit theorem to apply it to the thermal state satisfying Assumption \ref{assumption}.
  This is needed to verify the strong large deviation theorem (Lemma \ref{stld_thm}) in the next section.
  You can skip this section until Theorem \ref{Thm_gCLT} is used.
  
  Let $(X_{\lambda})_{\lambda\in\Lambda}$ be a family of random variables with each finite sample space $\Omega_\lambda$, where $\Lambda$ is the set of all positive real numbers or all positive integers.
   Let $M_{\lambda}(t):=\mathbb E[e^{t X_{\lambda}}]$ be the moment generating function, and $\psi_{\lambda}(t):=\log \mathbb E[e^{t X_{\lambda}}]$ be the cumulant generating function (cgf) of $X_{\lambda}$, where $\mathbb E$ denotes the expectation value.
   We denote the cumulative distribution function (cdf) of $(X_{\lambda}-\mathbb E[X_{\lambda}])/\sigma_{\lambda}$ by $F_{\lambda}(x):=\mathbb P\left(\frac{X_{\lambda}-\mathbb E[X_{\lambda}]}{\sigma_{\lambda}}\leq x\right)$, where $\sigma_{\lambda}$ is the standard deviation of $X_{\lambda}$.
    We use the following lemma \cite[Lemma 2, pp. 538]{feller71:_introd_probab_theor_its_applic}:
   \begin{lemma}[Feller \cite{feller71:_introd_probab_theor_its_applic}]\label{feller2x}\label{feller_lem}
    Let $F$ be a probability distribution whose expectation value is $0$.
    Let $\varphi$ be the characteristic function
    \begin{align}
    \varphi(\zeta):=\int_{-\infty}^{\infty}e^{i\zeta x}F(d x)
    \end{align}
    of $F$.
    %We denote the derivative of $\mathcal N$ by $g$, and denote the Fourier transform of $g$ by $\gamma$.
Let $\mathcal N$ be the cumulative distribution function of the standard Gaussian distribution.
    Then,
    \begin{align}
     |F(x)-\mathcal N (x)|\leq \int_{-T}^{T}\left|\frac{\varphi (\zeta)-e^{-\frac{1}{2}\zeta^2}}{\zeta}\right|d \zeta + \frac{24 m}{T}
    \end{align}
    holds for any $x\in \mathbb R$, $T>0$ and $m\geq 1/\sqrt{2\pi}$.
   \end{lemma}
  Then, we give the following generalization of the central limit theorem.
   
  \begin{theorem}\label{Thm_gCLT}
   If the cgf asymptotically satisfies
   \begin{align}
    \psi_{\lambda}(t)=\lambda\psi(t)+o(\lambda)\label{pt_conv}
   \end{align}
   pointwise with a function $\psi(t)$ on some interval $I:=[a_1,a_2]\ni 0$,
   the following asymptotic expansion uniformly holds for large enough $\lambda$:
   \begin{align}
    F_{\lambda}(x)=\mathcal N (x) +\order{\lambda^{-\frac{1}{2}}}.\label{CLT_eq}
   \end{align}
  \end{theorem}
   \begin{proof}
    {\bf Step1:} In this step, applying the method by Curtiss \cite{curtiss1942},
    we prove that the cgf $\psi_{\lambda}$ is extended to a holomorphic function on a small region around the real axis independently of $\lambda$.
    In addition, we show that this holomorphic function satisfies \eqref{pt_conv} uniformly on this region.
    
   We set $\mathbb E[X_{\lambda}]=0$ without loss of generality.
   %We fix an interval $I:=[-1,1]$.
   Then, since $M_\lambda$ is convex, it takes the maximum on $I$ at $a_1$ or $a_2$.
   By (\ref{pt_conv}), because $M_{\lambda}(t_0)^{\lambda^{-1}}$ converges to $e^{\psi(t_0)}$ for fixed $t_0=a_1\text{ or }a_2$,
   $M_{\lambda}(t)^{\lambda^{-1}}$ is uniformly bounded on $t\in I$.
   Because of
   \begin{align}
    |M_{\lambda}(t+i\zeta)|
    :=|\mathbb E[e^{(t+i\zeta)X_\lambda}]|
    \leq \mathbb E[|e^{(t+i\zeta)X_\lambda}|]=\mathbb E[e^{t X_\lambda}]=M_{\lambda}(t)\quad (\forall\zeta\in\mathbb{R}),
   \end{align}
   $M_{\lambda}(z)^{\lambda^{-1}}$ is uniformly bounded on the strip $\mathcal S:=\{z\in \mathbb{C}|\Re z\in I\}$.
   Thus, by Vitali's theorem, there exists a holomorphic function $m(z)$ such that $\lim_{\lambda\rightarrow\infty} M_{\lambda}(z)^{\lambda^{-1}}=m(z)$
   uniformly in any bounded closed subregion of $\mathcal S$.
   Since $m(t)=e^{\psi(t)}>0$ $(t\in I)$,
   $\Re m(z)>0$ $(z\in B_{\tilde\delta}:=\{z\in\mathbb{C}| |z|\leq \tilde\delta\})$ holds
   for sufficiently small $\tilde\delta>0$.
   Thus, $\psi(z):=\log m(z)$ is well defined as a holomorphic function on $B_{\tilde\delta}$.
   Because $M_{\lambda}(z)^{\lambda^{-1}}$ converges uniformly to $m(z)$ on $B_{\tilde\delta}$, the relation
   $\Re M_{\lambda}(z)^{\lambda^{-1}}>0$ $(z\in B_{\tilde\delta})$ holds for sufficiently large $\lambda$, hence the relation $\Re M_{\lambda}(z)>0$ does.
   Hence, $\psi_{\lambda}(z):=\log M_{\lambda}(z)$ is similarly well defined as a holomorphic function on $B_{\tilde\delta}$.
   Hence, $\psi_{\lambda}^{(n)}(z)=\lambda \psi^{(n)}(z)+o(\lambda)$ holds for any $n$, where $f^{(n)}$ denotes the $n$-th derivative of $f$.
   Especially, we have
   \begin{align}
   \sigma_{\lambda}^2=\psi_{\lambda}''(0)=\lambda \psi''(0)+o(\lambda)=\order{\lambda}.\label{sigmal}
   \end{align}

    {\bf Step 2:} In this step, combining the estimations in \cite{feller71:_introd_probab_theor_its_applic} and the asymptotic behavior of the cgf, we establish the desired estimation \eqref{CLT_eq}.
    
   The quantity $|\psi^{(3)}(z)|$ has the maximum value on $B_{\tilde\delta}$
   since $\psi^{(3)}(z)$ is holomorphic.
   Thus, because of $|\psi^{(3)}_{\lambda}(z)|=|\lambda \psi^{(3)}(z)+ o(\lambda)|\leq \lambda (|\psi^{(3)}(z)|+ o(1))$, there exists $C_0>0$ such that
   \begin{align}
    |\psi^{(3)}_{\lambda}(z)|\leq 6C_0\lambda \quad(z\in B_{\tilde\delta})\label{C_0}
   \end{align}
   holds for large enough $\lambda$.
   Then, we take a $\delta >0$ as
   \begin{align}
   \delta<\min\left\{\tilde\delta, \frac{\psi''(0)}{8C_0}\right\}<\frac{\psi_{\lambda}''(0)}{4C_0\lambda}=\frac{\sigma_{\lambda}^2}{4C_0\lambda},\label{delta}
   \end{align}
    where the last inequality holds for sufficiently large $\lambda$.
    Since $F_{\lambda}$ is the distribution function of $X_{\lambda}/\sigma_{\lambda}$,
    the characteristic function $\varphi_{\lambda}(\zeta)$ of $F_{\lambda}$ is equal to $M_{\lambda}(i\zeta/\sigma_{\lambda})$
    since $M_{\lambda}$ is analytically continued on $\mathcal S$.
    In addition, because $\psi_{\lambda}=\log M_{\lambda}$ is analytically continued on $B_{\tilde{\delta}}$,
    we have $\varphi_{\lambda}(\zeta)= e^{\psi_{\lambda}(i\zeta/\sigma_{\lambda})}$ for any $\zeta$ such that
    $|\zeta|/\sigma_{\lambda}\leq \delta < \tilde{\delta}$.    
    Applying Lemma \ref{feller2x} with $T=\delta \sigma_{\lambda}$ and $m=1$, we have
   \begin{align}
    |F_{\lambda}(x)-\mathcal N(x)|\leq
    \int_{-\delta\sigma_\lambda}^{\delta\sigma_\lambda}
    \left|\frac{e^{\psi_{\lambda}\left(\frac{i\zeta}{\sigma_\lambda}\right)}-e^{-\frac{1}{2}\zeta^2}}{\zeta}\right|d\zeta
    +\frac{C}{\sigma_\lambda}\label{clt_est}
   \end{align}
   with a constant $C:=24/\delta$.
   The second term is $\order{\lambda^{-\frac{1}{2}}}$ from
   \eqref{sigmal}.
   Then, we apply a similar method to \cite[pp. 534]{feller71:_introd_probab_theor_its_applic}.
   Observing that $|e^{\alpha}-1|\leq |\alpha|e^{\gamma}$ for any $\gamma \geq |\alpha|$,
   we have
   \begin{align}
    \left|e^{\psi_{\lambda}\left(\frac{i\zeta}{\sigma_\lambda}\right)+\frac{1}{2}\zeta^2}-1\right|\leq \left|\psi_{\lambda}\left(\frac{i\zeta}{\sigma_\lambda}\right)+\frac{1}{2}\zeta^2\right|
    e^{\gamma}\label{estimateA}
   \end{align}
    for any $\gamma \geq \left|\psi_{\lambda}\left(\frac{i\zeta}{\sigma_\lambda}\right)+\frac{1}{2}\zeta^2\right|$.
   By the Taylor expansion of $\psi_{\lambda}$ around $0$ for $i\zeta/\sigma_{\lambda}$ where $|\zeta|\leq \delta\sigma_{\lambda}$, there exists $\theta_{\lambda}\in B_{|\zeta|/\sigma_{\lambda}}\subset B_{\delta}\subset B_{\tilde\delta}$ such that
   \begin{align}
    \left|\psi_{\lambda}\left(\frac{i\zeta}{\sigma_\lambda}\right)+\frac{1}{2}\zeta^2\right|
    =\frac{1}{6}\left|\psi_{\lambda}^{(3)}\left(\theta_{\lambda}\right)\right|
    \frac{|\zeta|^3}{\sigma_\lambda^3}.
   \end{align}
%   where $\theta_{\lambda}\in B_{\delta}\subset B_{\tilde\delta}$.
   Since we focus on the domain $|\zeta|\leq \delta \sigma_{\lambda}$ of the integral in \eqref{clt_est}, we have
   \begin{align}
    \frac{1}{6}\left|\psi_{\lambda}^{(3)}\left(\theta_{\lambda}\right)\right|
    \frac{|\zeta|^3}{\sigma_\lambda^3}
    \stackrel{(a)}{\leq}
    \lambda C_0 \frac{|\zeta|^3}{\sigma_\lambda^3}
    \leq
    \lambda C_0 \delta \frac{|\zeta|^2}{\sigma_\lambda^2}
    \stackrel{(b)}{\leq}
    \frac{1}{4}|\zeta|^2\label{key_estimation}
   \end{align}
    for large enough $\lambda$, where
    $(a)$ and $(b)$ follow from \eqref{C_0} and \eqref{delta} respectively.
   Thus, we can take $\gamma =\frac{1}{4}|\zeta|^2$ in \eqref{estimateA}.
   Applying the first inequality in \eqref{key_estimation} combined with \eqref{estimateA},
   we have the following estimation of the integral in \eqref{clt_est} as
   \begin{align}
    \int_{-\delta\sigma_\lambda}^{\delta\sigma_\lambda}
    \left|\frac{e^{\psi_{\lambda}\left(\frac{i\zeta}{\sigma_\lambda}\right)}-e^{-\frac{1}{2}\zeta^2}}{\zeta}\right|d\zeta
    \leq
    \lambda C_0 {\sigma_\lambda^{-3}}
    \int_{-\delta\sigma_\lambda}^{\delta\sigma_\lambda}
    \zeta^2 e^{-\frac{1}{2}\zeta^2}d\zeta
    \leq
    \lambda C_0 {\sigma_\lambda^{-3}}
    \int_{-\infty}^{\infty}
    \zeta^2 e^{-\frac{1}{2}\zeta^2}d\zeta
    =\order{\lambda^{-\frac{1}{2}}}
   \end{align}
   since \eqref{sigmal} holds, and the Gaussian integral is finite.
   Thus, $|F_{\lambda}(x)-\mathcal N(x)|=\order{\lambda^{-\frac{1}{2}}}$ is proved.
   \end{proof}
  \section{Strong large deviation for the number of states}\label{sec_stld}
   In this section,
   we prepare a key lemma (Lemma \ref{stld_thm}) to
   deal with the estimation \eqref{relative_ent} of the relative entropy for the proof of Theorems \ref{achieve_thm} and \ref{achieve_thm_non}.
   Here, as in Sec.~\ref{Sub_achievability}, we assume Assumption \ref{assump2}, i.e.~the asymptotic extensivity of the free entropy $\phi_{\lambda}$ of the thermal states $\tau_{\bm{\theta}}^{(\lambda)}$ (Definition \ref{def_thermal}) and its derivatives.
   Let $\nu$ be the asymptotic density $\nu:=-(\sum_{i=1}^4\eta_{i}(\bm{\theta}_0)\theta_0^i+\phi(\bm{\theta}_0))$ of the negative entropy of the initial thermal state $\tau_{\bm{\theta}_{0}}^{(\lambda)}$.
   Recall that the probability distributions $p_{\bm{\theta}_0}^{(\lambda)}$ and $p_{\bm{\theta}_\lambda}^{(\lambda)}$
   are defined in \eqref{diag1} (commutative case), \eqref{y1} and \eqref{y2} (non-commutative case) as the eigenvalues of the density matrices of the thermal states,
  where $\bm{\theta}_{\lambda}$ is defined by (\ref{tl_1})-(\ref{tl_4}) with a vector $\vb{Q}_{\lambda}=(\Delta Q_{A,2,\lambda}, \Delta Q_{B,1,\lambda}, \Delta Q_{B,2,\lambda})$ in the main text.
 We assume $\lambda^{-\frac{5}{8}} \ll \|\vb{Q}_{\lambda}\| \ll \lambda$ as in Theorems \ref{achieve_thm} and \ref{achieve_thm_non}.
   For our purpose, we need a detailed estimation of the number of states $N_{l}^{(\lambda)}(a)$ $(l=0,1)$ $(a\in\mathbb R)$ defined as
   \begin{align}
    N_{l}^{(\lambda)}(a):=
    \left\{
     \begin{array}{cc}
      \# \left\{j | \frac{1}{\lambda}\log p_{\bm{\theta}_0}^{(\lambda)}(j)\geq
	  \nu + \lambda^{-\frac{1}{2}} a \right\} &(l=0)\\
      \# \left\{j | \frac{1}{\lambda}\log p_{\bm{\theta}_{\lambda}}^{(\lambda)}(j)\geq
	  \nu + \lambda^{-\frac{1}{2}} a \right\} &(l=1).
     \end{array}
 \right.\label{key_c77}
   \end{align}
   We carry out the estimation of $N_{l}^{(\lambda)}(a)$ by
   slightly modifying the strong large deviation theorem by Joutard \cite{joutard13:_stron_large_deviat_theor}.

  We firstly prepare some notations and results needed for the estimation along the line of \cite{joutard13:_stron_large_deviat_theor}.
   We regard $\lambda^{-1}\log p_{\bm{\theta}_0}^{(\lambda)}$ and $\lambda^{-1}\log p_{\bm{\theta}_\lambda}^{(\lambda)}$ as the random variables $Z_{0,\lambda}(j):=\lambda^{-1}\log p_{\bm{\theta}_0}^{(\lambda)}(j)$ and $Z_{1,\lambda}(j):=\lambda^{-1}\log p_{\bm{\theta}_\lambda}^{(\lambda)}(j)$ $(j\in \mathbb N_{d_{\lambda}})$ which are uniformly distributed on $\mathbb N_{d_{\lambda}}$.
   We denote the distribution function of $\lambda Z_{l,\lambda}$ by $K_{l,\lambda}$.
   Let $\varphi_{l,\lambda}$ $(l=0,1)$ be the normalized cgf of $\lambda Z_{l,\lambda}$,
\begin{align}
 \varphi_{l,\lambda}(t)
 :=\lambda^{-1}\log \mathbb E[e^{t \lambda Z_{l,\lambda}}]
 =\lambda^{-1}\log\sum_{j\in\mathbb N_{d_{\lambda}}}\frac{1}{d_{\lambda}}e^{t\lambda Z_{l,\lambda}(j)}.
 %=\lambda^{-1}\log\sum_{j\in\mathbb N_{d_{\lambda}}}\frac{1}{d_{\lambda}}(p_{\bm{\theta}_\lambda}^{(\lambda)}(j))^t.
\end{align}
They have other expressions
 \begin{align}
  \varphi_{0,\lambda}(t)&= \lambda^{-1}\left(\phi_{\lambda}(t\bm{\theta}_{0})-t\phi_{\lambda}(\bm{\theta}_{0})\right)-\lambda^{-1}\log d_{\lambda}\label{cgf0}\\
 \varphi_{1,\lambda}(t)&= \lambda^{-1}\left(\phi_{\lambda}(t\bm{\theta}_{\lambda})-t\phi_{\lambda}(\bm{\theta}_{\lambda})\right)-\lambda^{-1}\log d_{\lambda}\label{cgf1}.
 \end{align}
 Then, by Assumption \ref{assump2}, there exists an interval $I_1$ including $1$ such that both of $\varphi_{l,\lambda}(t)+\lambda^{-1}\log d_{\lambda}$ $(l=0,1)$ converge to
 $\varphi(t):=\phi(t\bm{\theta}_0)-t\phi(\bm{\theta}_0)$ uniformly with respect to $t$ on $I_1$.
Then, we define $\Lambda_{\lambda}(t):=\lambda^{-1}\sum_{i,j}t[\eta_{i} (t\bm\theta_0)-\eta_{i} (\bm\theta_0)]g^{ij}(\bm\theta_0)y_{\lambda,j}$, $(y_{\lambda,1},y_{\lambda,2},y_{\lambda,3},y_{\lambda,4}):=\bm{\eta}_{\lambda}(\bm{\theta}_{\lambda})-\bm{\eta}_{\lambda}(\bm{\theta}_{0})$, where
$\bm{\eta}_{\lambda}$ is defined in \eqref{dual_coord}.
 Since $\varphi$ is strictly convex, $f(x):=(\varphi')^{-1}(x)$ is well defined.

 The first and second derivatives of $\varphi_{1,\lambda}$ are related with those of $\varphi_{0,\lambda}$ by using the Taylor expansion and the expressions \eqref{cgf0} and \eqref{cgf1} as
 \begin{align}
 \varphi_{1,\lambda}'(t)&=\varphi'_{0,\lambda}(t)+\Lambda_{\lambda}'(t)+\order{\lambda^{-\frac{1}{2}}}
 +\order{\frac{\|\vb{Q}_{\lambda}\|^2}{\lambda^2}}\label{phi1}\\
\varphi_{1,\lambda}'(1)&= \lambda^{-1}[-\sum_{i}\eta_{\lambda,i}(\bm{\theta}_{\lambda})\theta_{\lambda}^i-\phi_{\lambda}(\bm{\theta}_{\lambda})]
  = - S(\tau_{\bm{\theta}_{\lambda}}^{(\lambda)})
  = - S(\tau_{\bm{\theta}_{0}}^{(\lambda)})
  = \varphi_{0,\lambda}' (1)\label{crucial}\\
 \varphi_{1,\lambda}''(t)&=
 \varphi''_{0,\lambda}(t)+\Lambda_{\lambda}''(t)+\order{\lambda^{-\frac{1}{2}}}
 +\order{\frac{\|\vb{Q}_{\lambda}\|^2}{\lambda^2}}.\label{phi2}
 \end{align}

\if0
Taylor expansion and the expression \eqref{cgf0} and \eqref{cgf1} as
 \begin{align}
 \varphi_{1,\lambda}'(t)&=\varphi'_{0,\lambda}(t)+\Lambda_{\lambda}'(t)+\order{\lambda^{-\frac{1}{2}}}
 +\order{\frac{\|\vb{Q}_{\lambda}\|^2}{\lambda^2}}\label{phi1}\\
&= \lambda^{-1}[-\sum_{i}\eta_{\lambda,i}(\bm{\theta}_{\lambda})\theta_{\lambda}^i-\phi_{\lambda}(\bm{\theta}_{\lambda})]
  = - S(\tau_{\bm{\theta}_{\lambda}}^{(\lambda)})
  = - S(\tau_{\bm{\theta}_{0}}^{(\lambda)})
  = \varphi_{0,\lambda}' (1)\label{crucial}
  \end{align}
since we have defined $\bm{\theta}_{\lambda}$ to satisfy \eqref{tl_1}.
\fi

 Now, we give an asymptotic expansion of $N_{l}^{(\lambda)}(a)$ through a strong large deviation estimation of the upper tail probability $\mathbb P(Z_{l,\lambda}\geq \nu + \lambda^{-\frac{1}{2}}a)=d_{\lambda}^{-1} N_{l}^{(\lambda)}(a)$ in the same way as \cite{joutard13:_stron_large_deviat_theor}.
 
  \begin{lemma}\label{stld_thm}
   Let $\lambda^{-\frac{5}{8}} \ll \|\vb{Q}_{\lambda}\| \ll \lambda$ and Assumption \ref{assump2} be satisfied.
   Then, for any $a\in \mathbb R$ and sufficiently large $\lambda$,
   defining $r_{k,\lambda}^{l}$ $(l=0,1), (k=0,1,2)$ by
 \begin{align}
  r_{2,\lambda}^{1}:=&
  \frac{1}{2}\varphi_{0,\lambda}'(1)f''(\nu)+\frac{1}{2}\varphi_{0,\lambda}''(1)f'(\nu)^2
 -f'(\nu)-\frac{1}{2}\nu f''(\nu)
 + \frac{1}{2}\Lambda_{\lambda}''(1)f'(\nu)^2
  \label{r12}\\
  r_{1,\lambda}^{1}:=&
  [\varphi_{0,\lambda}'(1)f'(\nu)-\nu f'(\nu) -1 + \varphi''(1)^{-1} (\nu-\varphi_{0,\lambda}'(1))\Lambda_{\lambda}''(1)f'(\nu) ]\lambda^{\frac{1}{2}}\label{r11}\\
  r_{0,\lambda}^{1}:=&
  -\lambda \nu
 -\frac{1}{2}\varphi''(1)^{-1}(\nu-\varphi_{0,\lambda}'(1))^2\lambda
 -\log\sqrt{2\pi}
 -\frac{1}{2}\log \varphi''(1)
 -\frac{1}{2}\log \lambda \nonumber\\
 &+\frac{1}{2}\varphi''(1)^{-2}\Lambda_{\lambda}''(1)(\nu-\varphi_{0,\lambda}'(1))^2\lambda
 -\frac{1}{2\varphi''(1)}\Lambda_{\lambda}''(1)
  \label{r10}\\
  r_{2,\lambda}^{0}:=&
  \frac{1}{2}\varphi_{0,\lambda}'(1)f''(\nu)+\frac{1}{2}\varphi_{0,\lambda}''(1)f'(\nu)^2
  -f'(\nu)-\frac{1}{2}\nu f''(\nu)\label{r02}\\
  r_{1,\lambda}^{0}:=&
  [\varphi_{0,\lambda}'(1)f'(\nu)-\nu f'(\nu) -1 ]\lambda^{\frac{1}{2}}\label{r01}\\
  r_{0,\lambda}^{0}:=&
  -\lambda \nu
  -\frac{1}{2}\varphi''(1)^{-1}(\nu-\varphi_{0,\lambda}'(1))^2\lambda
 -\log\sqrt{2\pi}
 -\frac{1}{2}\log \varphi''(1)
 -\frac{1}{2}\log \lambda
  \label{r00},
 \end{align}
  we have
  \begin{align}
  N_{1}^{(\lambda)}(a)
  =&\exp [r_{2,\lambda}^{1} a^2 + r_{1,\lambda}^{1} a + r_{0,\lambda}^{1} +\order{\lambda^{-\frac{1}{2}}} + \order{\lambda^{-2}\|\vb{Q}_{\lambda}\|^{2}}].\label{an_stld}\\
  N_{0}^{(\lambda)}(a)
  =&\exp [r_{2,\lambda}^{0} a^2 + r_{1,\lambda}^{0} a + r_{0,\lambda}^{0} +\order{\lambda^{-\frac{1}{2}}} ].\label{an_stld0}
  \end{align}
 \end{lemma}
 
 \begin{remark}
  Joutard gave the strong large deviation theorem (Theorem 1 of \cite{joutard13:_stron_large_deviat_theor}) under his assumptions (A.1) and (A.2) in \cite{joutard13:_stron_large_deviat_theor}.
 The latter (A.2) is the Edgeworth expansion, which is also satisfied in our case up to the first order.
 However, the former (A.1) requires that there exist functions $\varphi_{l}$, $J_{l}$ independently of $\lambda$ such that
 \begin{align}
  \varphi_{l,\lambda}'(t)=&\varphi_{l}'(t)+ \lambda^{-1}J_l(t) + o(\lambda^{-1})\label{a1_2}.
 \end{align}
 Because of $\Lambda_{\lambda}'(t)=\order{\lambda^{-1}\|\vb{Q}_{\lambda}\|}$ and $\lambda^{\frac{5}{8}}\ll \|\vb{Q}_{\lambda}\|$,
 % and $\varphi_{0,\lambda}'(t)=\order{1}$,
 $\Lambda_{\lambda}'(t)$ has strictly larger order than $\lambda^{-1}$.
 Hence, \eqref{phi1} contradicts (\ref{a1_2}).
 As for $\varphi'_{0,\lambda}$, because just
 \begin{align}
  \varphi'_{0,\lambda}(t)=\varphi'(t) + \order{\lambda^{\alpha-1}}
 \end{align}
 is guaranteed, (\ref{a1_2}) is not necessarily satisfied.
  In addition, he only treated the tail probability of the form $\mathbb P(Z_{\lambda}\geq a)$, where $a$ does not depend on the scale $\lambda$.
  In our case, $a$ is replaced by $\nu + \lambda^{-\frac{1}{2}}a$.
 Hence, we cannot directly apply Theorem 1 of \cite{joutard13:_stron_large_deviat_theor} for our situation.
 We will slightly modify his proof to obtain Lemma \ref{stld_thm}.
 \end{remark}
 
For the proof of Lemma \ref{stld_thm}, we prepare several lemmas.

\begin{lemma}\label{L1}
  \begin{align}
   \nu-\varphi_{0,\lambda}'(1) =& \order{\lambda^{\alpha - 1}} = o(\lambda^{-\frac{1}{2}})\label{ord1}\\
   1-\varphi_{0,\lambda}''(1)f'(\nu) =& \order{\lambda^{\alpha - 1}}= o(\lambda^{-\frac{1}{2}})\label{ord2}\\
   \Lambda_{\lambda}''(1) =& \order{\lambda^{-1}\|\vb{Q}_{\lambda}\|}=o(1).\label{ord3}
  \end{align}
\end{lemma}

\begin{proof}
These relations follow from Assumption \ref{assump2} and $\|\vb{Q}_{\lambda}\|=o(\lambda)$,
\end{proof}

We focus on $t_{\lambda,a}:= f(\nu + \lambda^{-\frac{1}{2}}a)$ as the variable $t$.
 Using the exponential tilting of the measure, 
 we define
 %the distribution function $K^*_{l,\lambda}$, and
 the random variable $\lambda Z_{l,\lambda}^*$ whose distribution function is $K_{l,\lambda}^*$ $(l=0,1)$ which is defined as
 \begin{align}
  K^*_{l,\lambda}(u):= \int_{-\infty< x \leq u} \exp[x t_{\lambda,a} - \lambda \varphi_{l,\lambda}(t_{\lambda,a})] d K_{l,\lambda}(x).
 \end{align}
 In fact, it is a distribution function since
 \begin{align}
  \lim_{u\rightarrow \infty}K^*_{l,\lambda}(u)
  =\int_{-\infty< x < \infty} \exp[x t_{\lambda,a} - \lambda \varphi_{l,\lambda}(t_{\lambda,a})] d K_{l,\lambda}(x)
  =\frac{\mathbb E[e^{t \lambda Z_{l,\lambda}}]}{\mathbb E[e^{t \lambda Z_{l,\lambda}}]}
  =1
 \end{align}
 and the other conditions are trivially satisfied.
Since the mean and the variance of $\lambda Z_{l,\lambda}^*$ are respectively equal to $\lambda \varphi_{l,\lambda}'(t_{\lambda,a})$ and $\varphi_{l,\lambda}''(t_{\lambda,a})$,
 we define the standardized random variable $V_{l,\lambda}$ as
 \begin{align}
  V_{l,\lambda}:=\frac{\lambda Z^*_{l,\lambda}-\lambda \varphi_{l,\lambda}'(t_{\lambda,a})}{\sqrt{\lambda\varphi_{l,\lambda}''(t_{\lambda,a})}}.
 \end{align}
Then, we have the following lemma.
\begin{lemma}\label{L2}
The distribution function $F_{l,\lambda}$ of the random variable $V_{l,\lambda}$ 
$(l=0,1)$
satisfies the central limit theorem as
 \begin{align}
  \sup_{y\in \mathbb R}|F_{l,\lambda}(y)- \mathcal N(y)| = \order{\lambda^{-\frac{1}{2}}}\label{gCLT_app}.
 \end{align}
\end{lemma}

\begin{proof}
%We show that the distribution functions $F_{l,\lambda}(y)$ $(l=0,1)$ of $V_{l,\lambda}$ satisfy the central limit theorem by applying Theorem \ref{Thm_gCLT}.
The cgf $\psi_{l,\lambda}(s)$ of $\lambda Z^*_{l,\lambda}$ is calculated as
 \begin{align}
  \psi_{l,\lambda}(s)
  =\log \int_{-\infty< u < \infty} e^{s u} d K_{l,\lambda}^* (u)
  =\log \int_{-\infty< u < \infty} e^{(s+t_{\lambda,a}) u - \lambda \varphi_{l,\lambda}(t_{\lambda,a})} d K_{l,\lambda} (u)
  %=\log \frac{\mathbb E[e^{(s+t_{\lambda,a})\lambda Z_{l,\lambda}}]}{\mathbb E[e^{t_{\lambda,a}\lambda Z_{l,\lambda}}]}.
  =\lambda \varphi_{l,\lambda}(s+t_{\lambda,a}) - \lambda \varphi_{l,\lambda}(t_{\lambda,a}).
  \label{c122}
 \end{align}
 As for $l=1$, \eqref{cgf1} and \eqref{c122} yield
 \begin{align}
  \psi_{1,\lambda}(s)
  =\phi_{\lambda}((s+t_{\lambda,a})\bm{\theta}_{\lambda})-(s+t_{\lambda,a})\phi_{\lambda}(\bm{\theta}_{\lambda})
  -\left(\phi_{\lambda}(t_{\lambda,a}\bm{\theta}_{\lambda})-t_{\lambda,a}\phi_{\lambda}(\bm{\theta}_{\lambda})\right)
  =\phi_{\lambda}((s+t_{\lambda,a})\bm{\theta}_{\lambda})
  -(s+t_{\lambda,a})\phi_{\lambda}(\bm{\theta}_{\lambda}).
 \end{align}
 Thus,
 because of Assumption \ref{assump2}, \eqref{tla_exp} and the definition of $\bm{\theta}_{\lambda}$,
 there exists a small interval $I_0\ni 0$ such that
 \begin{align}
  \psi_{1,\lambda}(s) = \lambda [\phi ((s+1)\bm{\theta}_0) - (s+1)\phi (\bm{\theta}_0)] + o(\lambda)
 \end{align}
 holds for any $s\in I_0$.
 In the same way, we also have
 \begin{align}
  \psi_{0,\lambda}(s) = \lambda [\phi ((s+1)\bm{\theta}_0) - (s+1)\phi (\bm{\theta}_0)] + o(\lambda).
 \end{align}
 Therefore, both $\lambda Z_{l,\lambda}$ $(l=0,1)$ satisfy the condition of Theorem \ref{Thm_gCLT}. Hence the distribution functions $F_{l,\lambda}$ of their standardized random variable $V_{l,\lambda}$ satisfy the central limit theorem
 \eqref{gCLT_app}
 by Theorem \ref{Thm_gCLT}.
 
\end{proof}

% The third derivative of $\varphi_{l,\lambda}$ $(l=0,1)$ satisfies
% \begin{align}
%  \varphi_{l,\lambda}'''(t)&= \varphi'''(t)+o(1)\label{phi3}.
% \end{align}

\begin{proof}[Proof of Lemma \ref{stld_thm}]
\noindent{\bf Step 1:}\quad
First, we prepare several formulas for $t_{\lambda,a}$ and $\varphi_{l,\lambda}$ together with its derivatives.
 The Taylor expansion with \eqref{cgf0}, \eqref{cgf1} yields
  \begin{align}
   t_{\lambda,a}=& 1 + f'(\nu) a \lambda^{-\frac{1}{2}} + \frac{1}{2} f''(\nu) a^{2} \lambda^{-1} + \order{\lambda^{-\frac{3}{2}}},\label{tla_exp}\\
  \varphi_{l,\lambda}(t_{\lambda,a})
 =&\varphi_{l,\lambda}(1)+\varphi_{l,\lambda}'(1)[\lambda^{-\frac{1}{2}}f'(\nu)a +\frac{1}{2}\lambda^{-1}f''(\nu)a^2]
 +\frac{1}{2}\varphi_{l,\lambda}''(1)\lambda^{-1}f'(\nu)^2 a^2
  +\order{\lambda^{-\frac{3}{2}}}\nonumber\\
 =&-\lambda^{-1}\log d_{\lambda} + \lambda^{-\frac{1}{2}}\varphi_{l,\lambda}'(1)f'(\nu)a +\frac{1}{2}\lambda^{-1}[\varphi_{l,\lambda}'(1)f''(\nu)a^2
 +\varphi_{l,\lambda}''(1)f'(\nu)^2 a^2]
   +\order{\lambda^{-\frac{3}{2}}}.\label{ptla}
  \end{align}
 Furthermore, the Taylor expansion gives
  \begin{align}
   \varphi_{1,\lambda}'(t_{\lambda,a})
   =&
   \varphi_{1,\lambda}'(1)
  +\lambda^{-\frac{1}{2}}\varphi_{1,\lambda}''(1)f'(\nu)a
   +\order{\lambda^{-1}}\nonumber\\
   \stackrel{(\rm a)}{=}&
   \varphi_{0,\lambda}'(1)
  +\lambda^{-\frac{1}{2}}[\varphi_{0,\lambda}''(1)+ \Lambda_{\lambda}''(1)]f'(\nu)a
   +\order{\lambda^{-1}}\label{tphi1}\\
   \varphi_{1,\lambda}''(t_{\lambda,a})
   =&
   \varphi_{1,\lambda}''(1)
   +\order{\lambda^{-\frac{1}{2}}}\nonumber\\
   \stackrel{(\rm b)}{=}&\varphi_{0,\lambda}''(1) + \Lambda_{\lambda}''(1) +\order{\lambda^{-\frac{1}{2}}},\label{tphi2}
  \end{align}
  where (a) follows from (\ref{phi1}) and (\ref{crucial}), and (b) follows from (\ref{phi2}).
 %Here, we decided the residual orders $\order{\lambda^{-\frac{3}{2}}}$, $\order{\lambda^{-1}}$, and $\order{\lambda^{-\frac{1}{2}}}$ in \eqref{ptla}, \eqref{tphi1}, and \eqref{tphi2} in response to the orders needed for the later analysis.
 Here, we calculated \eqref{ptla}, \eqref{tphi1}, and \eqref{tphi2} up to the necessary orders for the later analysis.
  In the same way, we have
  \begin{align}
   \varphi_{0,\lambda}'(t_{\lambda,a})
   =&
   \varphi_{0,\lambda}'(1)
  +\lambda^{-\frac{1}{2}}\varphi_{0,\lambda}''(1)f'(\nu)a
   +\order{\lambda^{-1}}\label{tphi01}\\
   \varphi_{0,\lambda}''(t_{\lambda,a})
   =&
   \varphi_{0,\lambda}''(1)
   +\order{\lambda^{-\frac{1}{2}}}.\label{tphi02}
  \end{align}
  %The third derivative satisfies
%  \begin{align}
%   \varphi_{l,\lambda}'''(t_{\lambda,a})= \varphi_{l,\lambda}'''(1) + o(1) = \varphi'''(1) + o(1)
%   \quad (l=0,1)\label{tphi3}
%  \end{align}
%  because of \eqref{phi3}.
Also, applying \eqref{tla_exp}, \eqref{tphi2} and \eqref{tphi02},
   we have the asymptotic expansions for $u_{l,\lambda}:=t_{\lambda,a}\sqrt{\lambda \varphi''_{l,\lambda}(t_{\lambda,a})}$ as
   \begin{align}
   \log u_{1,\lambda}
 =&
 \frac{1}{2}\log \varphi''(1)
 +\frac{1}{2\varphi''(1)}\Lambda_{\lambda}''(1)
 +\frac{1}{2}\log \lambda
 +\order{\lambda^{-\frac{1}{2}}}
   + \order{\lambda^{-2}\|\vb{Q}_{\lambda}\|^{2}},\label{uul2}\\
   \log u_{0,\lambda}
 =&
 \frac{1}{2}\log \varphi''(1)
 +\frac{1}{2}\log \lambda
 +\order{\lambda^{-\frac{1}{2}}}
 + \order{\lambda^{-2}\|\vb{Q}_{\lambda}\|^{2}}.\label{uul}
  \end{align}

\noindent{\bf Step 2:}\quad
In this step, we divide the probability $\mathbb P(Z_{l,\lambda}\geq \nu + \lambda^{-\frac{1}{2}}a)$ into two parts to estimate it.
 Defining
   \begin{align}
    b_{l,\lambda}:=&\lambda t_{\lambda,a} (\varphi'_{l,\lambda}(t_{\lambda,a}) - (\nu + \lambda^{-\frac{1}{2}} a))\\
    c_{l,\lambda}:=&\frac{\sqrt{\lambda}(\nu + \lambda^{-\frac{1}{2}} a - \varphi'_{l,\lambda}(t_{\lambda,a}))}{\sqrt{\varphi''_{l,\lambda}(t_{\lambda,a})}}
  \end{align}
  we calculate the probability $\mathbb P(Z_{l,\lambda}\geq \nu + \lambda^{-\frac{1}{2}}a)$ as
  \begin{align}
   &\mathbb P(Z_{l,\lambda}\geq \nu + \lambda^{-\frac{1}{2}}a)\nonumber\\
   =& \int_{u\geq \lambda (\nu + \lambda^{-\frac{1}{2}}a)} d K_{l,\lambda}(u)\nonumber\\
   =& \int_{u\geq \lambda (\nu + \lambda^{-\frac{1}{2}}a)}
   e^{-u t_{\lambda,a} + \lambda \varphi_{l,\lambda}(t_{\lambda,a})}
   e^{u t_{\lambda,a} - \lambda \varphi_{l,\lambda}(t_{\lambda,a})}
   d K_{l,\lambda}(u)\nonumber\\
   =& \int_{u\geq \lambda (\nu + \lambda^{-\frac{1}{2}}a)} e^{-u t_{\lambda,a} + \lambda \varphi_{l,\lambda}(t_{\lambda,a})} d K_{l,\lambda}^*(u)\nonumber\\
   =& e^{\lambda [\varphi_{l,\lambda}(t_{\lambda,a}) - t_{\lambda,a} \varphi'_{l,\lambda}(t_{\lambda,a})]}
   \int_{u\geq \lambda (\nu + \lambda^{-\frac{1}{2}}a)} e^{-u t_{\lambda,a} + \lambda \varphi_{l,\lambda}(t_{\lambda,a})} d K_{l,\lambda}^*(u)\nonumber\\
   \stackrel{(\rm a)}{=}& e^{\lambda [\varphi_{l,\lambda}(t_{\lambda,a}) - t_{\lambda,a} (\nu + \lambda^{-\frac{1}{2}} a)]} e^{-\lambda t_{\lambda,a} (\varphi'_{l,\lambda}(t_{\lambda,a}) - (\nu + \lambda^{-\frac{1}{2}} a))}
   \int_{y \geq c_{l,\lambda}} e^{-u_{l,\lambda} y} d F_{l,\lambda}(y) \nonumber\\
   =&e^{\lambda [\varphi_{l,\lambda}(t_{\lambda,a}) - t_{\lambda,a} (\nu + \lambda^{-\frac{1}{2}} a)]} e^{-b_{l,\lambda}}
   \int_{y \geq c_{l,\lambda}} e^{-u_{l,\lambda} y} d F_{l,\lambda}(y),\label{target}
  \end{align}
  where the equality $(\rm a)$ follows from integration by substitution with $y=(u-\lambda \varphi_{l,\lambda}'(t_{\lambda,a}))/\sqrt{\lambda\varphi_{l,\lambda}''(t_{\lambda,a})}$.
  Then, we divide the integral into two parts
  \begin{align}
   \int_{y \geq c_{l,\lambda}} e^{-u_{l,\lambda} y} d F_{l,\lambda}(y)
   =\int_{y \geq c_{l,\lambda}} e^{-u_{l,\lambda} y} d \mathcal N (y) + \int_{y \geq c_{l,\lambda}} e^{-u_{l,\lambda} y} d (F_{l,\lambda}(y) - \mathcal N (y))
   =: J_{l,1} + J_{l,2}.
  \end{align}
  The latter is estimated by integration by parts as
  \begin{align}
   |J_{l,2}|=&\left|e^{-u_{l,\lambda}c_{l,\lambda}} (F_{l,\lambda}(c_{l,\lambda})-\mathcal N (c_{l,\lambda})) + \int_{y\geq c_{l,\lambda}} u_{l,\lambda} e^{-u_{l,\lambda} y} (F_{l,\lambda}(y)-\mathcal N (y)) d y \right|\nonumber\\
   \leq &\left|e^{b_{l,\lambda}} + \int_{y\geq c_{l,\lambda}} u_{l,\lambda} e^{-u_{l,\lambda} y} d y \right| \sup_{y}\left|F_{l,\lambda}(y)-\mathcal N (y)\right|
   =2 e^{b_{l,\lambda}} \order{\lambda^{-\frac{1}{2}}}\label{J2_est}
  \end{align}
  by $-u_{l,\lambda}c_{l,\lambda}=b_{l,\lambda}$ and (\ref{gCLT_app}).
  
\noindent{\bf Step 3:}\quad
In this step, we calculate the former part $J_{l,1}$.
  The former part $J_{l,1}$ is also calculated by using integration by parts as follows:
  \begin{align}
   &J_{l,1}\nonumber\\
   =&\frac{1}{\sqrt{2\pi}}\int_{y \geq c_{l,\lambda}} e^{-u_{\lambda} y} e^{-\frac{y^2}{2}} d y \nonumber\\
   =&\frac{1}{\sqrt{2\pi}}\left[
   \frac{e^{-\frac{c_{l,\lambda}^2}{2}} e^{b_{l,\lambda}}}{u_{l,\lambda}}
   -\frac{1}{u_{l,\lambda}}\int_{y \geq c_{l,\lambda}} y e^{-u_{\lambda} y} e^{-\frac{y^2}{2}} d y
   \right] \nonumber\\
   =&\frac{1}{\sqrt{2\pi}}\left[
   \frac{e^{-\frac{c_{l,\lambda}^2}{2}} e^{b_{l,\lambda}}}{u_{l,\lambda}}
   -\frac{1}{u_{l,\lambda}}
    \left[
    \frac{c_{l,\lambda} e^{-\frac{c_{l,\lambda}^2}{2}} e^{b_{l,\lambda}}}{u_{l,\lambda}}
    -\frac{1}{u_{l,\lambda}}\int_{y \geq c_{l,\lambda}} (1-y^2) e^{-u_{\lambda} y} e^{-\frac{y^2}{2}} d y
    \right]
   \right].\label{o1_est}
  \end{align}
  Since $t_{\lambda,a}= f(\nu + \lambda^{-\frac{1}{2}} a)\rightarrow 1$, we have $t_{\lambda,a}\geq 0$ for large enough $\lambda$, which yields $u_{l,\lambda} \geq 0$.
  Hence, we have
  \begin{align}
   \left|\int_{y \geq c_{l,\lambda}} (1-y^2) e^{-u_{\lambda} y} e^{-\frac{y^2}{2}} d y \right|
   \leq
   e^{-u_{\lambda} c_{l,\lambda}} \int_{y \geq c_{l,\lambda}} |1-y^2| e^{-\frac{y^2}{2}} d y 
   =\order{1} e^{b_{l,\lambda}}\label{o1_int}
  \end{align}
  because Gaussian integrals are finite.
  By substituting \eqref{tphi1} and \eqref{tphi2}, the constant $c_{1,\lambda}$ is calculated as
  \begin{align}
   c_{1,\lambda}
   =&
   \lambda^{\frac{1}{2}}(\nu + \lambda^{-\frac{1}{2}} a - \varphi'_{1,\lambda}(t_{\lambda,a}))\varphi''_{1,\lambda}(t_{\lambda,a})^{-\frac{1}{2}}\nonumber\\
 =&[(\nu-\varphi_{0,\lambda}'(1))\lambda^{\frac{1}{2}} + (1-\varphi_{0,\lambda}''(1)f'(\nu))a - \Lambda_{\lambda}''(1)f'(\nu)a + \order{\lambda^{-\frac{1}{2}}}+\order{\lambda^{-2}\|\vb{Q}_{\lambda}\|^{2}}]\nonumber\\
 &\times\varphi''(1)^{-\frac{1}{2}}[1+\varphi''(1)^{-1}\Lambda_{\lambda}''(1)+\order{\lambda^{-\frac{1}{2}}}+ \order{\lambda^{-2}\|\vb{Q}_{\lambda}\|^{2}}]^{-\frac{1}{2}}\label{o1_c1}.
  \end{align}
Due to Lemma \ref{L1}, 
the equation \eqref{o1_c1} implies 
\begin{align}
c_{1,\lambda}=o(1).\label{ee1}
\end{align}
Similarly, we have
   \begin{align}
   c_{0,\lambda}
   =&
   [(\nu-\varphi_{0,\lambda}'(1))\lambda^{\frac{1}{2}} + (1-\varphi_{0,\lambda}''(1)f'(\nu))a + \order{\lambda^{-\frac{1}{2}}}]
  \varphi''(1)^{-\frac{1}{2}}[1+\order{\lambda^{-\frac{1}{2}}}]^{-\frac{1}{2}}\nonumber \\
   =& o(1) \label{o1_c0}
   \end{align}
   from \eqref{ord1}, \eqref{ord2}, \eqref{tphi01} and \eqref{tphi02}.
  Thus, evaluating \eqref{o1_est} with \eqref{o1_int}, \eqref{ee1} and \eqref{o1_c0}, we have
  \begin{align}
   J_{l,1}
   =
   \frac{e^{-\frac{c_{l,\lambda}^2}{2}} e^{b_{l,\lambda}}}{\sqrt{2\pi} u_{l,\lambda}}\left[1+\frac{1}{u_{l,\lambda}}\order{1}\right]
   =\frac{e^{-\frac{c_{l,\lambda}^2}{2}} e^{b_{l,\lambda}}}{\sqrt{2\pi} u_{l,\lambda}}\left[1+\order{\lambda^{-\frac{1}{2}}}\right]
   ,\label{J1_est}
  \end{align}
  where we apply $u_{l,\lambda}= \order{\sqrt{\lambda}}$ for the last equality.

\noindent{\bf Step 4:}\quad
We make further calculation of $c_{l,\lambda}^2$.
  Squaring \eqref{o1_c1}, we have
  \begin{align}
   c_{1,\lambda}^2
   =&
   [(\nu-\varphi_{0,\lambda}'(1))^2\lambda
   -2(\nu-\varphi_{0,\lambda}'(1))\Lambda_{\lambda}''(1)f'(\nu) \lambda^{\frac{1}{2}} a
   +(1-\varphi_{0,\lambda}''(1) f'(\nu))^2 a^2
   + \Lambda_{\lambda}''(1)^2 f'(\nu)^2 a^2 \nonumber\\
   &-2(1-\varphi_{0,\lambda}''(1)f'(\nu))f'(\nu)a^2 \Lambda_{\lambda}''(1)
   +2(1-\varphi_{0,\lambda}''(1)f'(\nu)) (\nu-\varphi_{0,\lambda}'(1)) \lambda^{\frac{1}{2}} a
 + \order{\lambda^{-\frac{1}{2}}}+\order{\lambda^{-2}\|\vb{Q}_{\lambda}\|^{2}}]\nonumber\\
 &\times\varphi''(1)^{-1}[1-\varphi''(1)^{-1}\Lambda_{\lambda}''(1)+\order{\lambda^{-\frac{1}{2}}}+ \order{\lambda^{-2}\|\vb{Q}_{\lambda}\|^{2}}]\nonumber\\
 \stackrel{(\rm a)}{=}&
 [(\nu-\varphi_{0,\lambda}'(1))^2\lambda
 -2(\nu-\varphi_{0,\lambda}'(1))\Lambda_{\lambda}''(1)f'(\nu) \lambda^{\frac{1}{2}} a
 + \order{\lambda^{-\frac{1}{2}}}+\order{\lambda^{-2}\|\vb{Q}_{\lambda}\|^{2}}]\nonumber\\
 &\times\varphi''(1)^{-1}[1-\varphi''(1)^{-1}\Lambda_{\lambda}''(1)+\order{\lambda^{-\frac{1}{2}}}+ \order{\lambda^{-2}\|\vb{Q}_{\lambda}\|^{2}}]\nonumber\\
 =&
 \varphi''(1)^{-1}(\nu-\varphi_{0,\lambda}'(1))^2\lambda
 -2 \varphi''(1)^{-1} (\nu-\varphi_{0,\lambda}'(1))\Lambda_{\lambda}''(1)f'(\nu) \lambda^{\frac{1}{2}} a \nonumber\\
 &-\varphi''(1)^{-2}\Lambda_{\lambda}''(1)(\nu-\varphi_{0,\lambda}'(1))^2\lambda
 + 2\varphi''(1)^{-2}(\nu-\varphi_{0,\lambda}'(1)) \Lambda_{\lambda}''(1)^2 f'(\nu) \lambda^{\frac{1}{2}} a
   +\order{\lambda^{-\frac{1}{2}}} + \order{\lambda^{-2}\|\vb{Q}_{\lambda}\|^2}\nonumber\\
   \stackrel{(\rm b)}{=}&
 \varphi''(1)^{-1}(\nu-\varphi_{0,\lambda}'(1))^2\lambda
 -2 \varphi''(1)^{-1} (\nu-\varphi_{0,\lambda}'(1))\Lambda_{\lambda}''(1)f'(\nu) \lambda^{\frac{1}{2}} a \nonumber\\
 &-\varphi''(1)^{-2}\Lambda_{\lambda}''(1)(\nu-\varphi_{0,\lambda}'(1))^2\lambda
   +\order{\lambda^{-\frac{1}{2}}} + \order{\lambda^{-2}\|\vb{Q}_{\lambda}\|^2}
   ,\label{ccl}    
  \end{align}
   where the terms $\Lambda_{\lambda}''(1)^2 f'(\nu)^2 a^2$, $(1-\varphi_{0,\lambda}''(1)f'(\nu))^2 a^2$, $2(1-\varphi_{0,\lambda}''(1)f'(\nu))f'(\nu)a^2 \Lambda_{\lambda}''(1)$, and $2(1-\varphi_{0,\lambda}''(1)f'(\nu)) (\nu-\varphi_{0,\lambda}'(1)) \lambda^{\frac{1}{2}} a$
   are included in $\order{\lambda^{-\frac{1}{2}}}+ \order{\lambda^{-2}\|\vb{Q}_{\lambda}\|^{2}}$ at (a) because of \eqref{ord1}-\eqref{ord3},
   and (b) follows from
   \begin{align}
   (\nu-\varphi_{0,\lambda}'(1)) \Lambda_{\lambda}''(1)^2 \lambda^{\frac{1}{2}}
   =\order{\lambda^{\alpha-\frac{1}{2}} \lambda^{-2}\|\vb{Q}_{\lambda}\|^{2}}
   =\order{\lambda^{-2}\|\vb{Q}_{\lambda}\|^{2}},
   \end{align}
   which is obtained from \eqref{ord1}, \eqref{ord3} and $\alpha < \frac{1}{2}$.
   Similarly, we have
  \begin{align}
   c_{0,\lambda}^2
 =&
 \varphi''(1)^{-1}(\nu-\varphi_{0,\lambda}'(1))^2\lambda
   +\order{\lambda^{-\frac{1}{2}}} + \order{\lambda^{-2}\|\vb{Q}_{\lambda}\|^2}.\label{c0_sub}
  \end{align}
   
\noindent{\bf Step 5:}\quad
Finally, we calculate $N_{1}^{(\lambda)}(a)$.
   %combining , , and (\ref{ccl})-(\ref{uul}) to evaluate ,
   That is, we obtain
    \begin{align}
     &\frac{1}{d_{\lambda}} N_{1}^{(\lambda)}(a)
     =\mathbb P(Z_{1,\lambda}\geq \nu + \lambda^{-\frac{1}{2}}a)
     \nonumber\\
    \stackrel{(\rm a)}{=}&e^{\lambda [\varphi_{1,\lambda}(t_{\lambda,a}) - t_{\lambda,a} (\nu + \lambda^{-\frac{1}{2}} a)]} e^{-b_{1,\lambda}}
     \frac{e^{-\frac{c_{1,\lambda}^2}{2}} e^{b_{1,\lambda}}}{\sqrt{2\pi} u_{1,\lambda}}\left[1+\order{\lambda^{-\frac{1}{2}}}\right]\nonumber\\
     \stackrel{(\rm b)}{=}&\frac{1}{d_{\lambda}}
     \exp \left[
 -\lambda \nu
 +[\varphi_{0,\lambda}'(1)f'(\nu)-\nu f'(\nu) -1 ]\lambda^{\frac{1}{2}} a \right.\nonumber\\
 &\left. +[\frac{1}{2}\varphi_{0,\lambda}'(1)f''(\nu)+\frac{1}{2}\varphi_{0,\lambda}''(1)f'(\nu)^2
 -f'(\nu)-\frac{1}{2}\nu f''(\nu)
 + \frac{1}{2}\Lambda_{\lambda}''(1)f'(\nu)^2
 +\order{\lambda^{\alpha-2}\|\vb{Q}_{\lambda}\|}
 ]a^2
 -\frac{1}{2}c_{1,\lambda}^2
 -\log \sqrt{2\pi}u_{1,\lambda}\right.\nonumber\\
 &\left. +\order{\lambda^{-\frac{1}{2}}}
 \right]\nonumber\\
 \stackrel{(\rm c)}{=}&
 \frac{1}{d_{\lambda}}
 \exp \left[
 \left[\frac{1}{2}\varphi_{0,\lambda}'(1)f''(\nu)+\frac{1}{2}\varphi_{0,\lambda}''(1)f'(\nu)^2
 -f'(\nu)-\frac{1}{2}\nu f''(\nu)
 + \frac{1}{2}\Lambda_{\lambda}''(1)f'(\nu)^2
 \right]a^2
 \right.\nonumber\\
 &\left. +[\varphi_{0,\lambda}'(1)f'(\nu)-\nu f'(\nu) -1 + \varphi''(1)^{-1} (\nu-\varphi_{0,\lambda}'(1))\Lambda_{\lambda}''(1)f'(\nu) ]\lambda^{\frac{1}{2}} a\right.\nonumber\\
 &
 -\lambda \nu
 -\frac{1}{2}\varphi''(1)^{-1}(\nu-\varphi_{0,\lambda}'(1))^2\lambda
 -\log\sqrt{2\pi}
 -\frac{1}{2}\log \varphi''(1)
 -\frac{1}{2}\log \lambda
 +\frac{1}{2}\varphi''(1)^{-2}\Lambda_{\lambda}''(1)(\nu-\varphi_{0,\lambda}'(1))^2\lambda
 -\frac{1}{2\varphi''(1)}\Lambda_{\lambda}''(1)
 \nonumber\\
 &\left.
 +\order{\lambda^{-\frac{1}{2}}} + \order{\lambda^{-2}\|\vb{Q}_{\lambda}\|^{2}} \right],\label{c44}
    \end{align}
   where (a) follows from combining (\ref{J2_est}), (\ref{J1_est}) with (\ref{target}), and
   we apply (\ref{tla_exp}) and (\ref{ptla}) to $\varphi_{1,\lambda}(t_{\lambda,a}) - t_{\lambda,a} (\nu + \lambda^{-\frac{1}{2}} a)$ at (b), and we substitute \eqref{uul2}, (\ref{uul}), (\ref{ccl}) and (\ref{c0_sub}) at (c).
  Similarly, we have
  \begin{align}
   &\frac{1}{d_{\lambda}}N_{0}^{(\lambda)}=\mathbb P(Z_{l,\lambda}\geq \nu + \lambda^{-\frac{1}{2}}a)\nonumber\\
   =&
   \frac{1}{d_{\lambda}}
   \exp \left[
 \left[\frac{1}{2}\varphi_{0,\lambda}'(1)f''(\nu)+\frac{1}{2}\varphi_{0,\lambda}''(1)f'(\nu)^2
 -f'(\nu)-\frac{1}{2}\nu f''(\nu)
 \right]a^2\right.\nonumber\\
 &\left. +[\varphi_{0,\lambda}'(1)f'(\nu)-\nu f'(\nu) -1 ]\lambda^{\frac{1}{2}} a \right.\nonumber\\
  &\left.
  -\lambda \nu
  -\frac{1}{2}\varphi''(1)^{-1}(\nu-\varphi_{0,\lambda}'(1))^2\lambda
 -\log\sqrt{2\pi}
 -\frac{1}{2}\log \varphi''(1)
 -\frac{1}{2}\log \lambda
  +\order{\lambda^{-\frac{1}{2}}} \right].\label{c45}
  \end{align}
  The equations \eqref{c44} and \eqref{c45} are equivalent with \eqref{an_stld} and \eqref{an_stld0}, hence the proof is completed.
 \end{proof}
   
  \section{Key Estimations for the proof of Theorems \ref{achieve_thm} and \ref{achieve_thm_non}}\label{app_relative_ent}
In this section, we prove the respective estimations \eqref{relative_ent} and \eqref{jj} of the relative entropy and the canonical correlations needed for the proof of Theorems \ref{achieve_thm} and \ref{achieve_thm_non}.
  We carry out all the proofs so that they are valid in the case where
  the observables $A_{i,\lambda}$, $B_{i,\lambda}$ are not necessarily commutative.
\subsection{Estimation of the relative entropy by applying the strong large deviation (proof of \eqref{relative_ent})}\label{apps_rel}
  We implement the estimation of the relative entropy for the proof of Theorems \ref{achieve_thm} and \ref{achieve_thm_non} in the main text by using Lemma \ref{stld_thm} prepared in the previous section.
 % We again assume that in the following.
 %In fact, it is assumed in Theorems \ref{achieve_thm} and \ref{achieve_thm_non}.
  Our goal is the following theorem:
  \begin{theorem}\label{rel_thm}
   Under Assumption \ref{assump2} and $\lambda^{\frac{5}{8}}\ll\|\vb{Q}_{\lambda}\|\ll\lambda$ , we have
   \begin{align}
   D(\rho_{\rm opt}^{(\lambda)}\|\tau_{\bm{\theta}_\lambda}^{(\lambda)})&=\order{\frac{\|\vb{Q}_{\lambda}\|^2}{\lambda^{2}}}+\order{\lambda^{-\frac{1}{2}}}\nonumber\\
   D(\rho_{\rm opt, nc}^{(\lambda)}\|\tau_{\bm{\theta}_\lambda}^{(\lambda)})&=\order{\frac{\|\vb{Q}_{\lambda}\|^2}{\lambda^{2}}}+\order{\lambda^{-\frac{1}{2}}}.\label{key_est}
   \end{align}
  \end{theorem}
%where the distributions $p_{\bm{\theta}_0}^{(\lambda)}$ and $p_{\bm{\theta}_\lambda}^{(\lambda)}$
%are defined in \eqref{diag1} (commutative case), \eqref{y1} and \eqref{y2} (non-commutative case).

\begin{proof}
 We proceed the estimation as follows in a similar method to \cite{PhysRevA.92.052308,PhysRevE.96.012128}.
 At first, we deal with the case when $A_{i,\lambda}$, $B_{i,\lambda}$ $(i=1,2)$ are mutually commutative.
We stepwise reduce our problem as follows.

\noindent{\bf Step 1:}\quad
 From the construction of $\rho_{\rm opt}^{(\lambda)}$, the following holds:
\begin{align}
 &D(\rho_{\rm opt}^{(\lambda)}\|\tau_{\bm{\theta}_\lambda}^{(\lambda)})\nonumber\\
 =&
 \tr \rho_{\rm opt}^{(\lambda)}(\log \rho_{\rm opt}^{(\lambda)}-\log \tau^{(\lambda)}_{\bm{\theta}_{\lambda}})\nonumber\\
 =&
 \sum_{j}
 p_{\bm{\theta}_0}^{(\lambda)}(j)
 (\log p_{\bm{\theta}_0}^{(\lambda)}(j)-\log p_{\bm{\theta}_{\lambda}}^{(\lambda)}(j)).
\end{align}
 Defining the random variable
\begin{align}
 Y_{l}^{(\lambda)}(j):=
 \left\{
  \begin{array}{cc}
   \frac{\log p_{\bm{\theta}_0}^{(\lambda)}(j)-\lambda\nu}{\sqrt{\lambda}} &(l=0)
  \\ \frac{\log p_{\bm{\theta}_\lambda}^{(\lambda)}(j)-\lambda\nu}{\sqrt{\lambda}} &(l=1),
  \end{array}
 \right.
 \label{def_y}
\end{align}
we have another expression the relative entropy
\begin{align}
 D(\rho_{\rm opt}^{(\lambda)}\|\tau_{\bm{\theta}_\lambda}^{(\lambda)})
 =\sqrt{\lambda}\left(\mathbb E_{p_{\bm{\theta}_0}^{(\lambda)}}[Y_{0}^{(\lambda)}]-\mathbb E_{p_{\bm{\theta}_0}^{(\lambda)}}[Y_{1}^{(\lambda)}]\right),\label{rel_ent_an}
\end{align}
where
$\mathbb E_{p}[X]$ denotes the expectation value of a random variable $X$ with probability distribution $p$.
To estimate the relative entropy, it is difficult to calculate $\mathbb E_{p_{\bm{\theta}_0}^{(\lambda)}}[Y_{1}^{(\lambda)}]$.
Instead,
we approximate $\Delta_{\lambda}(j):=Y_0^{(\lambda)}(j)-Y_1^{(\lambda)}(j)$ by a quadratic polynomial of $Y_0^{(\lambda)}(j)$.
In this way, we can calculate $\mathbb E_{p_{\bm{\theta}_0}^{(\lambda)}}[\Delta_{\lambda}]$ by calculating the moments of $Y_0^{(\lambda)}$.
%%%%%%%%%%%%%%%%%
%we can estimate this relative entropy because of $\mathbb E_{p_{\bm{\theta}_0}^{(\lambda)}}[Y_{0}^{(\lambda)}]=\lambda^{-\frac{1}{2}}\tr(\log \tau_{\bm{\theta}}^{(\lambda)}- S(\tau_{\bm{\theta}_0}^{(\lambda)}))\tau_{\bm{\theta}}^{(\lambda)}=\order{\lambda^{\alpha-\frac{1}{2}}}$
%and
%$\mathbb E_{p_{\bm{\theta}_0}^{(\lambda)}}[(Y_{0}^{(\lambda)})^2]
%=\lambda^{-1}\tr(\log \tau_{\bm{\theta}}^{(\lambda)}- S(\tau_{\bm{\theta}_\lambda}^{(\lambda)}))^2\tau_{\bm{\theta}}^{(\lambda)}=\order{1}+O(\lambda^{-2}\|\vb{Q}_{\lambda}\|)+\order{\lambda^{2\alpha-1}}$.
%%%%%%%%%%%%%%%%%
%%**
%The following analysis is based on the similar method as \cite{PhysRevA.92.052308,PhysRevE.96.012128}.

\noindent{\bf Step 2:}\quad
To compare $Y_1^{(\lambda)}(j)$ with $Y_0^{(\lambda)}(j)$, we compare the number of states $N_{1}^{(\lambda)}$ with $N_{0}^{(\lambda)}$ defined by (\ref{key_c77}).
The number of states $N_{l}^{(\lambda)}$ is expressed by $Y_{l}^{(\lambda)}$ as
\begin{align}
 N_{l}^{(\lambda)}(a)=\left|\left\{j | Y_l^{(\lambda)}(j)\geq a\right\}\right|\label{key_c7}.
\end{align}
As will be shown in {\bf Step 3}, the equation
\begin{align}
 \log N_1^{(\lambda)}(a-x)= \log N_0^{(\lambda)}(a)\label{number_eq}
\end{align}
for $x$ with a constant $a$ is asymptotically solved as
\begin{align}
 \sqrt{\lambda} x
 =&
 -\frac{1}{2}\Lambda_{\lambda}''(1)f'(\nu)^2 a^2
 -\varphi''(1)^{-1} (\nu-\varphi_{0,\lambda}'(1))\Lambda_{\lambda}''(1)f'(\nu)\lambda^{\frac{1}{2}}a 
 \nonumber\\
 &-\frac{1}{2}\varphi''(1)^{-2}\Lambda_{\lambda}''(1)(\nu-\varphi_{0,\lambda}'(1))^2\lambda
 +\frac{1}{2\varphi''(1)}\Lambda_{\lambda}''(1)
 +\order{\lambda^{-\frac{1}{2}}} + \order{\lambda^{-2}\|\vb{Q}_{\lambda}\|^{2}}\nonumber\\
 =:& \sqrt{\lambda}q_{\lambda}(a)+\order{\lambda^{-\frac{1}{2}}} + \order{\lambda^{-2}\|\vb{Q}_{\lambda}\|^{2}}.\label{solsol}
\end{align}
Since $Y_{l}^{(\lambda)}$ satisfies
\begin{align}
 Y_{l}^{(\lambda)}(1)\geq Y_{l}^{(\lambda)}(2) \geq Y_{l}^{(\lambda)}(3) \geq \dots
\end{align}
by its definition,
$\log N_{l}^{(\lambda)}(Y_{l}^{(\lambda)}(j))$ is asymptotically equal to $\log j$.
Thus, the equation
\begin{align}
 \log N_{1}^{(\lambda)}(Y_{1}^{(\lambda)}(j)) = \log N_{0}^{(\lambda)}(Y_{0}^{(\lambda)}(j))
\end{align}
holds asymptotically.
Thus, $\Delta_{\lambda}(j)= Y_0^{(\lambda)}(j)-Y_1^{(\lambda)}(j)$ satisfies
\begin{align}
 \log N_{1}^{(\lambda)}(Y_0^{(\lambda)}(j)-\Delta_{\lambda}(j)) = \log N_{0}^{(\lambda)}(Y_{0}^{(\lambda)}(j)).
\end{align}
Then, we obtain the approximation of $\Delta_{\lambda}(j)$ by the solution \eqref{solsol} of the equation
\eqref{number_eq} as
 \begin{align}
  &\sqrt{\lambda}\Delta_{\lambda}(j)\nonumber\\
  =&\sqrt{\lambda} q_{\lambda}(Y_0^{(\lambda)}(j)) + +\order{\lambda^{-\frac{1}{2}}} + \order{\lambda^{-2}\|\vb{Q}_{\lambda}\|^{2}}\nonumber\\
  =&
  -\frac{1}{2}\Lambda_{\lambda}''(1)f'(\nu)^2 (Y_0^{(\lambda)}(j))^2 \nonumber\\
 &-\varphi''(1)^{-1} (\nu-\varphi_{0,\lambda}'(1))\Lambda_{\lambda}''(1)f'(\nu)\lambda^{\frac{1}{2}} Y_0^{(\lambda)}(j)
 \nonumber\\
 &-\frac{1}{2}\varphi''(1)^{-2}\Lambda_{\lambda}''(1)(\nu-\varphi_{0,\lambda}'(1))^2\lambda
 +\frac{1}{2\varphi''(1)}\Lambda_{\lambda}''(1)
 +\order{\lambda^{-\frac{1}{2}}} + \order{\lambda^{-2}\|\vb{Q}_{\lambda}\|^{2}}.
\label{H3-14}
 \end{align}

\noindent{\bf Step 3:}\quad
In this step, we show that the solution of the equation \eqref{number_eq}
is asymptotically given by \eqref{solsol}.
From the asymptotic expansions \eqref{an_stld} and \eqref{an_stld0} in Lemma \ref{stld_thm}, the equation \eqref{number_eq} is written as
\begin{align}
 r_{2,\lambda}^{1}(a-x)^2 + r_{1,\lambda}^{1}(a-x) +r_{0,\lambda}^{1} +\order{\lambda^{-\frac{1}{2}}} + \order{\lambda^{-2}\|\vb{Q}_{\lambda}\|^{2}}
 =r_{2,\lambda}^{0} a^2 + r_{1,\lambda}^{0} a + r_{0,\lambda}^{0} +\order{\lambda^{-\frac{1}{2}}},
\end{align}
which may be deformed as
\begin{align}
 r_{2,\lambda}^{1} x^2 + (-2 r_{2,\lambda}^{1} a - r_{1,\lambda}^{1}) x
 - (r_{2,\lambda}^{0} - r_{2,\lambda}^{1}) a^2 - (r_{1,\lambda}^{0} - r_{1,\lambda}^{1}) a
 - r_{0,\lambda}^{0} - r_{0,\lambda}^{1} +\order{\lambda^{-\frac{1}{2}}} + \order{\lambda^{-2}\|\vb{Q}_{\lambda}\|^{2}} = 0.\label{eq_x}
\end{align}
 Dividing the both sides of \eqref{eq_x} by $\lambda$ and changing the variable $x$ to $y:=x/\sqrt{\lambda}$, we have the equation for $y$ as
 \begin{align}
  q_{2,\lambda} y^2 + q_{1,\lambda} y + \epsilon_{\lambda}= 0,\label{eq_y}
 \end{align}
 where
 \begin{align}
  q_{2,\lambda}&:= r_{2,\lambda}^{1}= \order{1}\\
  q_{1,\lambda}&:= - r_{1,\lambda}^{1}\lambda^{-\frac{1}{2}}+ \order{\lambda^{-\frac{1}{2}}}=\order{1}\\
 \epsilon_{\lambda}&:= \lambda^{-1}[
 - (r_{2,\lambda}^{0} - r_{2,\lambda}^{1}) a^2 - (r_{1,\lambda}^{0} - r_{1,\lambda}^{1}) a
 - r_{0,\lambda}^{0} - r_{0,\lambda}^{1} +\order{\lambda^{-\frac{1}{2}}} + \order{\lambda^{-2}\|\vb{Q}_{\lambda}\|^{2}}] =\order{\lambda^{-2}\|\vb{Q}_{\lambda}\|}.
 \end{align}
 The perturbation from $y=0$ up to $\order{\epsilon_{\lambda}}$ gives
 \begin{align}
  y= -\frac{\epsilon_{\lambda}}{q_{1,\lambda}} + \order{\epsilon_{\lambda}^2}
  = -\frac{\epsilon_{\lambda}}{q_{1,\lambda}} + \order{\lambda^{-4}\|\vb{Q}_{\lambda}\|^{2}}.
 \end{align}
 In fact, substituting it to the left hand side of \eqref{eq_y}, we have
 \begin{align}
  q_{2,\lambda} y^2 + q_{1,\lambda} y + \epsilon_{\lambda}
  =
  \order{\epsilon_{\lambda}^2}
  - \epsilon_{\lambda} + \epsilon_{\lambda}
  = \order{\epsilon_{\lambda}^2}.
 \end{align}
 Therefore, we obtain
\begin{align}
 \sqrt{\lambda} x= \lambda y
 =& -\lambda\frac{\epsilon_{\lambda}}{q_{1,\lambda}} + \order{\lambda^{-3}\|\vb{Q}_{\lambda}\|^{2}}\nonumber\\
 =&
 - \frac{\sqrt{\lambda}}{r_{1,\lambda}^1}
 [(r_{2,\lambda}^{0} - r_{2,\lambda}^{1}) a^2 + (r_{1,\lambda}^{0} - r_{1,\lambda}^{1}) a
 + r_{0,\lambda}^{0} - r_{0,\lambda}^{1}] +\order{\lambda^{-\frac{1}{2}}} + \order{\lambda^{-2}\|\vb{Q}_{\lambda}\|^{2}}.\label{pert0}
\end{align}
 Thus, substituting $r_{l,\lambda}^{(j)}$ given by \eqref{r12}-\eqref{r00} in Lemma \ref{stld_thm} to \eqref{pert0},
 we have
\begin{align}
 \sqrt{\lambda} x
 =&
 \left[
 1-(\varphi'_{0,\lambda}(1)-\nu)f'(\nu)
 -\varphi''(1)^{-1} (\nu-\varphi_{0,\lambda}'(1))\Lambda_{\lambda}''(1)f'(\nu)
 \right]^{-1}\nonumber\\
 &\times
 \left[
 -\frac{1}{2}\Lambda_{\lambda}''(1)f'(\nu)^2
  a^2
 -\varphi''(1)^{-1} (\nu-\varphi_{0,\lambda}'(1))\Lambda_{\lambda}''(1)f'(\nu)\lambda^{\frac{1}{2}}a \right.
 \nonumber\\
 &\left.\hspace{10pt}-\frac{1}{2}\varphi''(1)^{-2}\Lambda_{\lambda}''(1)(\nu-\varphi_{0,\lambda}'(1))^2\lambda
 +\frac{1}{2\varphi''(1)}\Lambda_{\lambda}''(1)
 \right]
 +\order{\lambda^{-\frac{1}{2}}} + \order{\lambda^{-2}\|\vb{Q}_{\lambda}\|^{2}}\nonumber\\
 \stackrel{(\rm a)}{=}&
 \left[
 1+ \order{\lambda^{-\frac{1}{2}}}
 \right]\nonumber\\
 &\times
 \left[
 -\frac{1}{2}\Lambda_{\lambda}''(1)f'(\nu)^2 a^2
 -\varphi''(1)^{-1} (\nu-\varphi_{0,\lambda}'(1))\Lambda_{\lambda}''(1)f'(\nu)\lambda^{\frac{1}{2}}a \right.
 \nonumber\\
 &\left.\hspace{10pt}-\frac{1}{2}\varphi''(1)^{-2}\Lambda_{\lambda}''(1)(\nu-\varphi_{0,\lambda}'(1))^2\lambda
 +\frac{1}{2\varphi''(1)}\Lambda_{\lambda}''(1)
 \right]
 +\order{\lambda^{-\frac{1}{2}}} + \order{\lambda^{-2}\|\vb{Q}_{\lambda}\|^{2}}\nonumber\\
 =&
 -\frac{1}{2}\Lambda_{\lambda}''(1)f'(\nu)^2 a^2
 -\varphi''(1)^{-1} (\nu-\varphi_{0,\lambda}'(1))\Lambda_{\lambda}''(1)f'(\nu)\lambda^{\frac{1}{2}}a 
 \nonumber\\
 &-\frac{1}{2}\varphi''(1)^{-2}\Lambda_{\lambda}''(1)(\nu-\varphi_{0,\lambda}'(1))^2\lambda
 +\frac{1}{2\varphi''(1)}\Lambda_{\lambda}''(1)
 +\order{\lambda^{-\frac{1}{2}}} + \order{\lambda^{-2}\|\vb{Q}_{\lambda}\|^{2}},\nonumber
\end{align}
 where (a) is verified by observing that
 $(\varphi'_{0,\lambda}(1)-\nu)f'(\nu)$ and $\varphi''(1)^{-1} (\nu-\varphi_{0,\lambda}'(1))\Lambda_{\lambda}''(1)f'(\nu)$
 are included in $+\order{\lambda^{-\frac{1}{2}}}+ \order{\lambda^{-2}\|\vb{Q}_{\lambda}\|^{2}}]^{-\frac{1}{2}}$ because of Lemma \ref{L1}.

\noindent{\bf Step 4:}\quad
Finally, in this step, we evaluate the relative entropy
$D(\rho_{\rm opt}^{(\lambda)}\|\tau_{\bm{\theta}_\lambda}^{(\lambda)})$
by using \eqref{H3-14}.
%First, substituting $Y_0^{(\lambda)}(j)$ to $a$,
%we obtain the desired approximation of $\Delta(j)$ as
% \begin{align}
%  &\Delta(j)\nonumber\\
%  =&Y_0^{(\lambda)}(j)-Y_1^{(\lambda)}((Y_0^{(\lambda)})^{-1}[Y_0^{(\lambda)}(j)])\nonumber\\
%  =&
%Y_0^{(\lambda)}(j)-N_1^{-1}(N_0(Y_0^{(\lambda)}(j)))\nonumber\\
%=& Q(Y_0^{(\lambda)}(j)).\label{H3}
% \end{align}
% \begin{align}
%  \Lambda_{\lambda}'(1)
%  =\lambda^{-1}\sum_{j}\theta_0^{j}K_{\lambda,j}+\order{\frac{\|\vb{Q}_{\lambda}\|^2}{\lambda^{2-\alpha}}}
%  =\order{\frac{\|\vb{Q}_{\lambda}\|^2}{\lambda^{2-\alpha}}},
% \end{align}
%%%%%%%%%%%%%%%%%%%%
% \begin{align}
%  \Lambda_{\lambda}''(1)
%  =\Lambda_{\lambda}'(1)+\sum_{i,j,k,l}\frac{\partial^3\phi}{\partial\theta^{i}\partial\theta^{k}\partial\theta^{l}}\theta^{k}\theta^lJ_{\lambda}^{ij}K_{\lambda,j}
%  =\order{\frac{\|\vb{Q}_{\lambda}\|^2}{\lambda^{2-\alpha}}}+\order{\frac{\|\vb{Q}_{\lambda}\|}{\lambda}},
  % \end{align}
 Because
 \begin{align}
  \varphi_{0,\lambda}''(1)=& \lambda^{-1} \sum_{i,j} \frac{\partial^2 \phi_{\lambda}}{\partial \theta^i \partial \theta^j} (\bm{\theta}_0) \theta_0^i \theta_0^j\nonumber\\
  =&
  \lambda^{-1}
  \tr\left[\tau_{\bm{\theta}_0}^{(\lambda)}\sum_{i}\theta_0^i(X_{i,\lambda}-\eta_{\lambda,i}(\bm{\theta}_0)) \sum_{j}\theta_0^j (X_{j,\lambda}-\eta_{\lambda,j}(\bm{\theta}_0))\right]\nonumber\\
  =&
  \lambda^{-1}
  \tr \tau_{\bm{\theta}_0}^{(\lambda)} (-\log \tau_{\bm{\theta}_0}^{(\lambda)} - S(\tau_{\bm{\theta}_{0}}^{(\lambda)}))^2\nonumber\\
  =&
  \lambda^{-1} \mathbb E_{p_{\bm{\theta}_0}^{(\lambda)}}[(-\log p^{(\lambda)}_{\bm{\theta}_0}- S(\tau_{\bm{\theta}_0}^{(\lambda)}))^2]
 \end{align}
 holds,
 we have
 \begin{align}
   \mathbb E_{p_{\bm{\theta}_0}^{(\lambda)}}[Y_0^{(\lambda)}]
   =&-\sqrt{\lambda}(\nu - \varphi'_{0,\lambda}(1)),\label{ym}\\
  \mathbb E_{p_{\bm{\theta}_0}^{(\lambda)}}[(Y_0^{(\lambda)})^2]
  =& \lambda^{-1} \mathbb E_{p_{\bm{\theta}_0}^{(\lambda)}}[(-\log p^{(\lambda)}_{\bm{\theta}_0}- S(\tau_{\bm{\theta}_0}^{(\lambda)}))^2] + (S(\tau_{\bm{\theta}_0}^{(\lambda)}) - \lambda \nu)^2\nonumber\\
  =& \varphi_{0,\lambda}''(1) +\lambda(\nu - \varphi'_{0,\lambda}(1))^2+\order{\lambda^{-\frac{1}{2}}}\nonumber\\
   =&\varphi''(1)+\lambda(\nu - \varphi'_{0,\lambda}(1))^2+\order{\lambda^{-\frac{1}{2}}},\label{yv}\\
   f'(\nu)=& \frac{1}{\varphi''(1)}\label{fff}
 \end{align}
 in view of \eqref{crucial}.
 Thus, we obtain
 \begin{align}
  &D(\rho_{\rm opt}^{(\lambda)}\|\tau_{\bm{\theta}_\lambda}^{(\lambda)})\nonumber\\
  \stackrel{(\rm a)}{=}&\sqrt{\lambda}\mathbb E_{p_{\bm{\theta}_0}^{(\lambda)}}[\Delta_{\lambda}]
  \nonumber\\
  \stackrel{(\rm b)}{=}&
  -\frac{1}{2}\Lambda_{\lambda}''(1)f'(\nu)^2 \mathbb E_{p_{\bm{\theta}_0}^{(\lambda)}}[(Y_0^{(\lambda)})^2] \nonumber\\
 &-\varphi''(1)^{-1} (\nu-\varphi_{0,\lambda}'(1))\Lambda_{\lambda}''(1)f'(\nu) \lambda^{\frac{1}{2}} \mathbb E_{p_{\bm{\theta}_0}^{(\lambda)}}[Y_0^{(\lambda)}]
 \nonumber\\
 &-\frac{1}{2}\varphi''(1)^{-2}\Lambda_{\lambda}''(1)(\nu-\varphi_{0,\lambda}'(1))^2\lambda
 +\frac{1}{2\varphi''(1)}\Lambda_{\lambda}''(1)
  +\order{\lambda^{-\frac{1}{2}}} + \order{\lambda^{-2}\|\vb{Q}_{\lambda}\|^{2}}\nonumber\\
  \stackrel{(\rm c)}{=}&
  -\frac{1}{2\varphi''(1)}\Lambda_{\lambda}''(1)+\frac{1}{2\varphi''(1)}\Lambda_{\lambda}''(1)\nonumber\\
  &+\varphi''(1)^{-2}\Lambda_{\lambda}''(1)(\nu-\varphi_{0,\lambda}'(1))^2\lambda
  -\frac{1}{2}\varphi''(1)^{-2}\Lambda_{\lambda}''(1)(\nu-\varphi_{0,\lambda}'(1))^2\lambda
  -\frac{1}{2}\varphi''(1)^{-2}\Lambda_{\lambda}''(1)(\nu-\varphi_{0,\lambda}'(1))^2\lambda
  \nonumber\\ 
  &+\order{\lambda^{-\frac{1}{2}}} + \order{\lambda^{-2}\|\vb{Q}_{\lambda}\|^{2}}\nonumber\\
  =&\order{\lambda^{-\frac{1}{2}}} + \order{\lambda^{-2}\|\vb{Q}_{\lambda}\|^{2}},
 \end{align}
 where (a) and (b) follow from (\ref{rel_ent_an}) and (\ref{H3-14}) respectively, and (c) follows from
 substituting \eqref{ym}-\eqref{fff}.
 %where we observe that $\alpha < \frac{1}{2}$ implies
% \begin{align}
%  \varphi''(1)^{-2}(\nu-\varphi_{0,\lambda}'(1))^2 \Lambda_{\lambda}''(1)^2 f'(\nu)\lambda
%  =\order{\lambda^{2\alpha-3} \|\vb{Q}_{\lambda}\|^{2}}
%  =o(\lambda^{-2}\|\vb{Q}_{\lambda}\|^{2})
% \end{align}
% for the verification of $(a)$.
 Hence, we obtain the desired estimation for $D(\rho_{\rm opt}^{(\lambda)}\|\tau_{\bm{\theta}_\lambda}^{(\lambda)})$.
 
 {\bf Non-commutative case:} For non-commutative $A_{i,\lambda}$ and $B_{i,\lambda}$ $(i=1,2)$, since $\rho_{\rm opt,nc}^{(\lambda)}$ commutes with $\tau_{\bm{\theta}_\lambda}^{(\lambda)}$ by its construction, we also have
\begin{align}
 D(\rho_{\rm opt,nc}^{(\lambda)}\|\tau_{\bm{\theta}_\lambda}^{(\lambda)})
 =
 \tr \rho_{\rm opt,nc}^{(\lambda)}(\log \rho_{\rm opt,nc}^{(\lambda)}-\log \tau^{(\lambda)}_{\bm{\theta}_{\lambda}})
 =
 \sum_{j}
 p_{\bm{\theta}_0}^{(\lambda)}(j)
 (\log p_{\bm{\theta}_0}^{(\lambda)}(j)-\log p_{\bm{\theta}_{\lambda}}^{(\lambda)}(j)).
\end{align}
 %Thus, $p_{\bm{\theta}}^{(\lambda)}(1)\geq p_{\bm{\theta}}^{(\lambda)}(2)\geq \dots$ holds.
 Thus, we can define $Y_l^{(\lambda)}(j)$ as with (\ref{def_y}).
 Therefore, Step 1 and Step 2 are completely the same as the commutative case because it is sufficient to deal with just the probability distributions composed of the eigenvalue of the thermal states.
 In Step 3, (\ref{ym}) also holds for the non-commutative case.
 In addition, (\ref{yv}) is also valid because we have
 \begin{align}
  \varphi_{0,\lambda}''(1)=& \lambda^{-1} \sum_{i,j} \frac{\partial^2 \phi_{\lambda}}{\partial \theta^i \partial \theta^j} (\bm{\theta}_0) \theta_0^i \theta_0^j\nonumber\\
  =&
  \lambda^{-1}\sum_{i,j}\int_{0}^{1}ds\;\tr\left[\left(\tau_{\bm{\theta}_0}^{(\lambda)}\right)^{1-s}(X_{i,\lambda}-\eta_{\lambda,i}(\bm{\theta}_0))\left(\tau_{\bm{\theta}_0}^{(\lambda)}\right)^s (X_{j,\lambda}-\eta_{\lambda,j}(\bm{\theta}_0))\right]
  \theta_0^i \theta_0^j\nonumber\\
  =&
  \lambda^{-1}
  \int_{0}^{1}ds\;\tr\left[\left(\tau_{\bm{\theta}_0}^{(\lambda)}\right)^{1-s}\sum_{i}\theta_0^i(X_{i,\lambda}-\eta_{\lambda,i}(\bm{\theta}_0))\left(\tau_{\bm{\theta}_0}^{(\lambda)}\right)^s \sum_{j}\theta_0^j (X_{j,\lambda}-\eta_{\lambda,j}(\bm{\theta}_0))\right]\nonumber\\
  \stackrel{(\rm a)}{=}&
  \lambda^{-1}
  \tr\left[\tau_{\bm{\theta}_0}^{(\lambda)}\sum_{i}\theta_0^i(X_{i,\lambda}-\eta_{\lambda,i}(\bm{\theta}_0)) \sum_{j}\theta_0^j (X_{j,\lambda}-\eta_{\lambda,j}(\bm{\theta}_0))\right]\nonumber\\
  =&
  \lambda^{-1}
  \tr \tau_{\bm{\theta}_0}^{(\lambda)} (-\log \tau_{\bm{\theta}_0}^{(\lambda)} - S(\tau_{\bm{\theta}_{0}}^{(\lambda)}))^2\nonumber\\
  =&
  \lambda^{-1} \mathbb E_{p_{\bm{\theta}_0}^{(\lambda)}}[(-\log p^{(\lambda)}_{\bm{\theta}_0}- S(\tau_{\bm{\theta}_0}^{(\lambda)}))^2],
 \end{align}
 where (a) follows from that $\tau_{\bm{\theta}_0}^{(\lambda)}= \exp [\sum_{i} \theta_0^i X_{i,\lambda}- \phi_{\lambda}(\bm{\theta}_0)]$ commutes with $\sum_{i} \theta_0^i X_{i,\lambda}$.
 Thus, Step 3 is also the same as the commutative case.
 Then, the proof is completed.
 \end{proof}

 \subsection{Estimation of the Fisher information (proof of \eqref{jj})}\label{apps_fish}
 Next, we prove \eqref{jj} (Sec.~\ref{sub_ach_FGCB}).
 Recall that $\vb{s}_{\lambda}(t)$ is defined as the generalized inverse temperature such that
\begin{align}
 \bm{\eta}_{\lambda}(\vb{s}_{\lambda}(t))=t\bm{\eta}_{\lambda}(\bm{\theta}_{\lambda})+(1-t)\bm{\eta}_{\lambda}(\bm{\xi}_{\lambda})
\end{align}
 by the ideal inverse temperature $\bm{\theta}_{\lambda}$ associated with a vector $\vb{Q}_{\lambda}$, and the effective inverse temperature $\xi_{\lambda}:=\tilde{\bm{\theta}}_{\lambda}(\rho_{\rm opt (,nc)}^{(\lambda)})$ of the final state $\rho_{\rm opt}^{(\lambda)}$ ($\rho_{\rm opt, nc}^{(\lambda)}$) of our protocol (for the non-commutative case).
 We show the following estimation of the Fisher information $J_{\lambda,ij}$ defined by \eqref{canonical_cor}:
  \begin{lemma}\label{fish_lem}
   Under Assumption \ref{assump2} and $\lambda^{\frac{5}{8}}\ll\|\vb{Q}_{\lambda}\|\ll\lambda$, we have
  \begin{align}
  \max_{t\in[0,1]}\|(J_{\lambda,ij}(\vb{s}_{\lambda}(t)))_{ij}\|
  =\order{\lambda}.
 \end{align}
  \end{lemma}
 \begin{proof}
 We consider the non-commutative case, which of course includes the commutative case.
 First of all, since $\|A\|\leq \|A\|_1$ holds for any matrix $A$, where $\|A\|_1=\tr|A|$ is the trace norm,
 we have
 \begin{align}
  \|(J_{\lambda,ij}(\vb{s}_{\lambda}(t)))_{ij}\|
  \leq
  \tr (J_{\lambda,ij}(\vb{s}_{\lambda}(t)))_{ij}
  =
  \sum_{i=1}^{4}
  \int_{0}^{1}ds\ \tr (X_{i,\lambda} - \eta_{\lambda,i}(\vb{s}_{\lambda}(t))) \left(\tau^{(\lambda)}_{\vb{s}_{\lambda}(t)}\right)^s (X_{i,\lambda} - \eta_{\lambda,i}(\vb{s}_{\lambda}(t)))\left(\tau^{(\lambda)}_{\vb{s}_{\lambda}(t)}\right)^{1-s},\label{c19}
 \end{align}
 where $X_{i,\lambda}=A_{i,\lambda}$, $X_{i+2,\lambda}=B_{i,\lambda}$ $(i=1,2)$.
 Furthermore, for any $0\leq s \leq 1$ we have
 \begin{align}
  \tr (X_{i,\lambda} - \eta_{\lambda,i}(\vb{s}_{\lambda}(t))) \left(\tau^{(\lambda)}_{\vb{s}_{\lambda}(t)}\right)^s (X_{i,\lambda} - \eta_{\lambda,i}(\vb{s}_{\lambda}(t)))\left(\tau^{(\lambda)}_{\vb{s}_{\lambda}(t)}\right)^{1-s}
  \leq
  \tr (X_{i,\lambda} - \eta_{\lambda,i}(\vb{s}_{\lambda}(t)))^2 \tau^{(\lambda)}_{\vb{s}_{\lambda}(t)}\label{c20}
 \end{align}
 since the Wigner-Yanase-Dyson skew information \cite{wigner1963information,wigner1964positive,1742-6596-201-1-012015}
 \begin{align}
  I_{\rho,s}(X):=\tr X^2\rho - \tr X \rho^s X \rho^{1-s}
 \end{align}
 is positive $I_{\rho,s}(X)\geq 0$ for any state $\rho$, observable $X$, and $0\leq s \leq 1$.
 The positivity follows from the fact that $I_{\rho,s}(X)$ is convex with respect to $\rho$ \cite{lieb1973convex},
 because $I_{\ket{\psi}\bra{\psi},s}(X)$ is obviously positive for any pure state $\ket{\psi}$.
 Observing that
 \begin{align}
  \tr (X_{i,\lambda} - \eta_{\lambda,i}(\vb{s}_{\lambda}(t)))^2 \tau^{(\lambda)}_{\vb{s}_{\lambda}(t)}
  =\tr X_{i,\lambda}^2 \tau^{(\lambda)}_{\vb{s}_{\lambda}(t)} - \eta_{\lambda,i}(\vb{s}_{\lambda}(t))^2
  \leq \tr X_{i,\lambda}^2 \tau^{(\lambda)}_{\vb{s}_{\lambda}(t)},
 \end{align}
 The combination of \eqref{c19} and \eqref{c20} yields that
 \begin{align}
  \|(J_{\lambda,ij}(\vb{s}_{\lambda}(t)))_{ij}\|
  \leq
  \sum_{i=1}^{4}
  \tr X_{i,\lambda}^2 \tau^{(\lambda)}_{\vb{s}_{\lambda}(t)}.\label{c23}
 \end{align}
 From the inequalities $\|AB\|_1\leq \|A\|\|B\|_1$ and $\|A^2\|\leq\|A\|^2$ for any matrices $A,B$, and the assumption $\|X_{i,\lambda}\|=\order{\lambda}$,
 we obtain
 \begin{align}
  \tr X_{i,\lambda}^2 \tau^{(\lambda)}_{\vb{s}_{\lambda}(t)}
  \leq
  \|X_{i,\lambda}^2 \tau^{(\lambda)}_{\vb{s}_{\lambda}(t)}\|_1
  \leq
  \|X_{i,\lambda}\|^2 \|\tau^{(\lambda)}_{\vb{s}_{\lambda}(t)}\|_1
  \leq
  \|X_{i,\lambda}\|^2
  =\order{\lambda^2}\label{c24}
 \end{align}
 since $\|\tau^{(\lambda)}_{\vb{s}_{\lambda}(t)}\|_1= \tr \tau^{(\lambda)}_{\vb{s}_{\lambda}(t)}=1$.
 Thus, the combination of \eqref{c23} and \eqref{c24} implies that
 \begin{align}
  \max_{t\in[0,1]}\|(J_{\lambda,ij}(\vb{s}_{\lambda}(t)))_{ij}\|=\order{\lambda^2}.\label{Jup}
 \end{align}
 Furthermore, we improve this estimation by using \eqref{optup}.
 Combining \eqref{Jup} and \eqref{optup}, we have
 \begin{align}
  &\|\bm{\eta}_{\lambda}(\bm{\theta}_{\lambda})-\bm{\eta}_{\lambda}(\bm{\xi}_{\lambda})\|\nonumber\\
  =&\order{\sqrt{D(\rho_{\rm opt, nc}^{(\lambda)}\|\tau_{\bm{\theta}_\lambda}^{(\lambda)})
  \max_{t\in[0,1]}\|(J_{\lambda,ij}(\vb{s}_{\lambda}(t)))_{ij}\|}}\label{finalopt}\nonumber\\
  =&\order{D(\rho_{\rm opt, nc}^{(\lambda)}\|\tau_{\bm{\theta}_\lambda}^{(\lambda)})^{\frac{1}{2}}\lambda}.
 \end{align}
 Since we assumed $\|\vb{Q}_{\lambda}\|=o(\lambda)$,
 the relation $D(\rho_{\rm opt, nc}^{(\lambda)}\|\tau_{\bm{\theta}_\lambda}^{(\lambda)})=o(1)$ follows from Theorem \ref{rel_thm},
 which implies $\|\bm{\eta}_{\lambda}(\bm{\theta}_{0})-\bm{\eta}_{\lambda}(\bm{\xi}_{\lambda})\|=o(\lambda)$ because of $\|\bm{\eta}_{\lambda}(\bm{\theta}_{0})-\bm{\eta}_{\lambda}(\bm{\theta}_{\lambda})\|=o(\lambda)$.
 Hence, the relation $\|\bm{\eta}_{\lambda}(\bm{\theta}_{0})-\bm{\eta}_{\lambda}(\vb{s}_{\lambda}(t))\|=o(\lambda)$ holds for any $t\in[0,1]$ since $\bm{\eta}_{\lambda}(\vb{s}_{\lambda}(t))$ is a convex combination of $\bm{\eta}_{\lambda}(\bm{\theta}_{\lambda})$ and $\bm{\eta}_{\lambda}(\bm{\xi}_{\lambda})$.
 Therefore,
 \begin{align}
  \vb{s}_{\lambda}(t)\rightarrow \bm{\theta}_{0}\label{c27}
 \end{align}
 holds.
 Because \eqref{asymp_correlation} and \eqref{c27} imply that
 \begin{align}
  J_{\lambda,ij}(\vb{s}_{\lambda}(t))=\lambda g_{ij}(\bm{\theta}_0)+o(\lambda),
 \end{align}
 we have
 \begin{align}
  \max_{t\in[0,1]}\|(J_{\lambda,ij}(\vb{s}_{\lambda}(t)))_{ij}\|
  =\order{\lambda}.
 \end{align}
 \end{proof}
  \section{Examples with i.i.d.-scaling}\label{sec_iid_example}
 
%%%%%%%%%%%%%%%%%%%%%%%%%%%%%%%%
\subsection{An ordinary heat engine with i.i.d.~particles}\label{sub_iid}
We verify that
the model of the heat engine in the previous work \cite{PhysRevE.96.012128}
is included in our theory.
In this heat engine,
the hot and the cold baths
consist of $n$ particles with Hamiltonian $H_h$ and $H_c$ respectively.
Quantity $A$ is the energy, and $B$ is empty in this case.
The scale parameter is the number $n$ of the particles.
The scale dependent observables of the hot and the cold baths are
$H_{h,n}:=\sum_{l=0}^{n} I^{l} \otimes H_{h}\otimes I^{n-l}$
and
$H_{c,n}:=\sum_{l=0}^{n} I^{l} \otimes H_{c}\otimes I^{n-l}$
respectively.
%not necessarily exactly the same number
Then, the initial thermal state is the i.i.d.~Gibbs state
$\left(\frac{e^{-\beta_{h}H_h}}{\tr e^{-\beta_{h}H_h}}\right)^{\otimes n}\otimes\left(\frac{e^{-\beta{c}H_c}}{\tr e^{-\beta{c}H_c}}\right)^{\otimes n}
=\frac{e^{-\beta_{h}H_h^{\otimes n}-\beta{c}H_c^{\otimes n}}}{(\tr e^{-\beta_{h}H_h-\beta{c}H_c})^{n}}
$ with the inverse temperatures $\beta_{h}, \beta_{c}>0$.
It is easy to check that Assumption \ref{assumption} is satisfied since
$\phi_{n}(\beta_{h},\beta_{c})=\log (\tr e^{-\beta_{h}H_h-\beta{c}H_c})^n=n\log \tr e^{-\beta_{h}H_h-\beta{c}H_c}
=n\phi(\beta_{h},\beta_{c})$.
Indeed, because $\phi(\beta_{h},\beta_{c})=\log \tr e^{-\beta_{h}H_h-\beta_{c}H_c}$ is smooth, Assumption \ref{assumption} is satisfied.
In this case, the deviation from the extensivity is exactly $0$.
Hence, the achievability in Sec.~\ref{Sub_achievability} is verified.
In fact, although the paper \cite{PhysRevE.96.012128} gives a slightly different operation as the asymptotically optimal operation by using the specific structure of i.i.d.~and gives a better estimation of the bound, the application of our general theory also gives the same estimation up to the second order as follows.
%In this case, Quantity $A$ is the energy, and no Quantity $B$ exists.
Then, FGCB becomes
\begin{align}
 \Delta W
 \leq  \left(1-\frac{\beta_{h}}{\beta_{c}}\right)\Delta Q_{h,n}
 -C_{AA}\frac{\Delta Q_{h,n}^2}{n}
 +o\left(\frac{\Delta Q_{h,n}^2}{n}\right),
\end{align}
where $\Delta W$ and $\Delta Q_{h,n}$ are the extracted work and the endothermic heat from the hot bath respectively.
Since there is no correlation between two baths,
the matrix composed of the canonical correlations is just a diagonal matrix
\begin{align}
 (g_{ij}(\beta{c},\beta_{h}))_{ij}=
 \left(
 \begin{array}{cc}
  \sigma_{L}^2 & 0\\
  0 & \sigma_{H}^2
 \end{array}
 \right),
\end{align}
where $\sigma_{L,H}^2$ is the variance of the energy of each bath at each initial inverse temperature.
Thus, its inverse is
\begin{align}
 (g^{ij}(\beta{c},\beta_{h}))_{ij}=
 \left(
 \begin{array}{cc}
  \sigma_{L}^{-2} & 0\\
  0 & \sigma_{H}^{-2}
 \end{array}
 \right).\label{gij}
\end{align}
Therefore, the coefficient $C_{AA}$ is calculated as
\begin{align}
 C_{AA}
  =\frac{1}{2}\left(
   \frac{g^{11}(\beta{c},\beta_{h})\beta_{h}^2}{(\beta{c})^3}
   +\frac{g^{22}(\beta{c},\beta_{h})}{\beta{c}}
 \right)
 =\frac{\beta_{h}^2}{2\sigma_L^2\beta{c}^3}
   +\frac{1}{2\sigma_H^2\beta{c}},
\end{align}
which indeed reproduces the second order coefficient \cite[Eq. (39)]{PhysRevE.96.012128}.
%We show that the heat baths consists of i.i.d. Gibbs state particles

\subsection{Spin-$\frac{1}{2}$ bath}\label{sub_spin}
\begin{figure}[!t]
\centering
\includegraphics[clip ,width=3.2in]{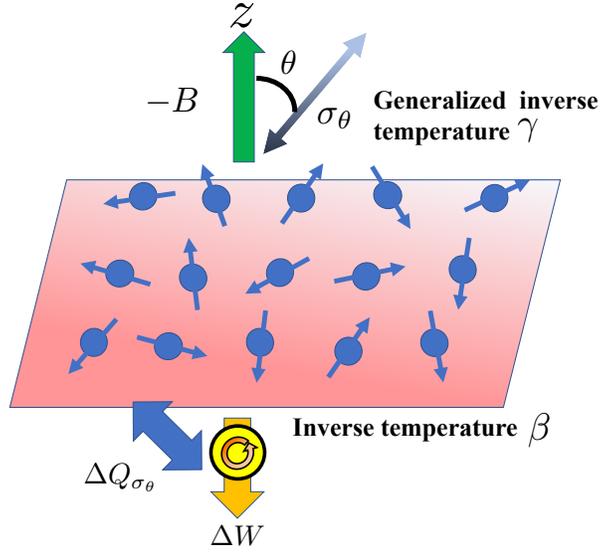}
 \caption{Schematic of spin-$\frac{1}{2}$ bath
 }
\label{figure51}
\end{figure}
The next example is a simple toy model to illustrate the explicit behavior of the coefficient of the finite-size effect in FGCB with non-commutative quantities in a two-level system, though its scaling is of i.i.d.
We consider the work extraction from just one bath composed of spin-$\frac{1}{2}$ systems without interaction (Fig.~\ref{figure51}) in the following model.
$n$ particles with spin-$\frac{1}{2}$ are placed on a lattice, so that each particle does not move.
We assume that interactions among particles are negligible.
We impose a uniform external magnetic field in $z$-direction, then the Hamiltonian of each particle is
$H=\hbar\omega \sigma_z$, where $\sigma_z=\ketbra{0}{0}-\ketbra{1}{1}$, and $2\omega$ is the cyclotron frequency.
In this example, Quantity $A$ is the energy given by this Hamiltonian.
As seen from FGCB, by using another conserved quantity $B$, we can extract work even from one bath.
As a toy model, we consider the spin in another direction as Quantity $B$.
Obviously, since this system is symmetric with respect to the rotation around $z$-axis,
it is sufficient to consider the angle $0\leq\theta\leq\frac{\pi}{2}$ between the spin and $z$-axis.
Then, we denote the $\theta$-direction spin by
\begin{align}
 \sigma_{\theta}=\cos\theta(\ketbra{0}{0}-\ketbra{1}{1})+\sin\theta(\ketbra{0}{1}+\ketbra{1}{0}).
\end{align}
We suppose that
the bath is
%$\sigma_\theta$ together with the Hamiltonian $H$ are
thermalized to the generalized thermal state
\begin{align}
 \tau^{(n)}_{\beta,\gamma}=\left(\frac{e^{-\beta H -\gamma \sigma_{\theta}}}{\mathcal Z}\right)^{\otimes n}
\end{align}
like the grand canonical ensemble with the non-commutative charges $\sigma_{\theta}$ and $H$ \cite{Yunger-Halpern:2016aa},
where $\beta>0$ is the ordinary inverse temperature of the bath, and $\gamma$ is the corresponding generalized inverse temperature for $\sigma_{\theta}$.
%The spin $\sigma_{\theta}$ is not conserved 
%does not commute with $H$
%Suppose that the spin has  .
Here, $\gamma$ is taken to be dimensionless.
Although it is unclear whether such a thermal state can be realized physically in this way,
this is one of the simplest examples of a generalized heat engine with two distinct conserved quantities.
Since it is simple to calculate the coefficient of the second order term in FGCB for this example,
we can analytically observe the behavior of the finite-size effects on its optimal performance.
The free entropy is calculated as
\begin{align}
 \phi_n(\beta,\gamma)
 =n\log \mathcal Z
 =n(\log\cosh\sqrt{(\beta\hbar\omega)^2+\gamma^2+2\gamma\beta\hbar\omega\cos\theta} +\log 2).
\end{align}
We use this state as the initial state of the bath.
Note that in the commutative case $\theta=0$,
$\sigma_{0}=\sigma_{z}$ is proportional to $H$,
which means that $\sigma_{0}$ and $H$ are essentially the same quantities.
Hence, $\sigma_{0}$ is useless for the work extraction.
%Note that $H$ and $\sigma_\theta$ is obviously dependent .
Thus, in this system, non-commutativity is needed for the work extraction.

For the work extraction $\Delta W$ under the supply $\Delta Q_{\sigma_{\theta},n}=o(n)$ of the $\theta$-direction spin,
the FGCB takes the following form:
\begin{align}
 \Delta W_n
 \leq -\frac{\gamma}{\beta}\Delta Q_{\sigma_{\theta},n}
 -C(\beta,\gamma,\omega,\theta)\frac{(\Delta Q_{\sigma_{\theta},n})^2}{n}
 +O\left(\frac{(\Delta Q_{\sigma_{\theta},n})^3}{n^2}\right).
\end{align}
The coefficient $C_{\theta}(\beta,\gamma;\omega)$ is calculated by \eqref{Cbb}
as
\begin{align}
 C(\beta,\gamma,\omega,\theta)
 =\frac{((\beta\hbar\omega)^2+\gamma^2+2\gamma\beta\hbar\omega\cos\theta)^{\frac{3}{2}}}{2\beta^3(\hbar\omega)^2\sin^2\theta\tanh\sqrt{(\beta\hbar\omega)^2+\gamma^2+2\gamma\beta\hbar\omega\cos\theta}}.
\end{align}
Note that this coefficient explicitly depends on
the full parameters:
the direction $\theta$ and cyclotron frequency $2\omega$ as well as inverse temperatures.
As already mentioned in the general theory,
the coefficient $C(\beta,\gamma,\omega,\theta)$ reflects the correlation between the Hamiltonian and $\theta$ direction spin.
%while the coefficient of the first term corresponding GCB is given just by .
Thus,
while just the ratio $\frac{\gamma}{\beta}$ between the inverse temperatures
determines the maximum work extraction
in thermodynamic limit,
the imposed field and the direction $\theta$ of the spin themselves explicitly make differences
in consideration of finite-size regime.

To extract the work as large as possible,
%(or take as small as possible cost)
we should minimize $C(\beta,\gamma,\omega,\theta)$ under the fixed $\eta:=\frac{\gamma}{\beta}$ to keep the first term.
Then, the coefficient $C(\beta,\gamma,\omega,\theta)$ is rewritten as
\begin{align}
 C(\beta,\omega,\theta;\eta)
 :=C(\beta,\beta\eta,\omega,\theta)
 =\frac{((\hbar\omega)^2+\eta^2+2\hbar\omega\eta\cos\theta)^{\frac{3}{2}}}{2(\hbar\omega)^2\sin^2\theta\tanh(\beta\sqrt{(\hbar\omega)^2+\eta^2+2\hbar\omega\eta\cos\theta})}.
\end{align}
Interestingly, it depends on not only the ratio $\eta$, but also the single inverse temperature $\beta$.
When $\eta$ (hence the first term) is fixed,
the lower the temperature is,
the smaller $C(\beta,\omega,\theta;\eta)$ becomes.
Moreover, $C(\beta,\omega,\theta;\eta)$ quite differently behaves in accordance with the sign of $\eta$ (i.e.~of $\gamma$) as follows.

At first, we consider the case when $\eta>0$.
In this case, $\Delta Q_{\sigma_{\theta},n}<0$ is needed to extract work.
The coefficient $C(\beta,\omega,\theta;\eta)$ diverges $+\infty$ as $\theta\rightarrow 0$.
The coefficient $C(\beta,\omega,\theta;\eta)$ always takes its minimum at $\theta=\frac{\pi}{2}$ ($x$-direction)
for any $\beta$ and $\omega$:
\begin{align}
 C(\beta,\omega,\frac{\pi}{2};\eta)
 =\frac{(\hbar^2\omega^2+\eta^2)^{\frac{3}{2}}}{2\hbar^2\omega^2\tanh(\beta\sqrt{\hbar^2\omega^2+\eta^2})}.\label{pih}
\end{align}
The derivative of \eqref{pih} with respect to $\omega$ is
\begin{align}
 \frac{\partial}{\partial\omega}C(\beta,\omega,\frac{\pi}{2};\eta)
 =\frac{\hbar^4\omega^4-\hbar^2\omega^2\eta^2-2\eta^4}{\hbar^3\omega^3\sqrt{\hbar^2\omega^2+\eta^2}\tanh(\beta\sqrt{\hbar^2\omega^2+\eta^2})}
 -\frac{\beta(\hbar^2\omega^2+\eta^2)}{\hbar\omega\sinh^2(\beta\sqrt{\hbar^2\omega^2+\eta^2})}.\label{der_pih}
\end{align}
Thus, further the value \eqref{pih}
takes its minimum at the $\omega_m$ such that the RHS of \eqref{der_pih} vanishes.
%%.
At large enough $\beta$, i.e.~low enough temperature $T:=k_{B}^{-1}\beta^{-1}$, where $k_B$ is the Boltzmann constant,
we have $\hbar\omega_m\approx \sqrt{2}\eta$ since the second term in \eqref{der_pih} becomes negligible.
Thus, in summary, to make the maximum work large, one should use $x$-direction spin and low temperature $T$, and tune $\omega$
to $\omega_m\approx \sqrt{2}\eta$.
As an example, we show the graph of $C(\beta,\omega,\frac{\pi}{2};\eta)$ as a function of $\theta$
at $T=1\rm{K}$, and $\eta=1\rm{J}$ in Fig.~\ref{figure_pos_om}, which indeed takes its minimum at
$\hbar\omega\approx\sqrt{2}\rm{J}=\sqrt{2}\eta$.
We also plot the graph of $C(\beta,\omega,\theta;\eta)$ as a function of $\theta$ at the same $T$ and $\eta$ with $\hbar\omega=10\rm{J}$ (solid (blue) curve) and $\hbar\omega=\sqrt{2}\rm{J}$
(dashed (red) curve)
in Fig.~\ref{figure_pos}, which shows that
$C(\beta,\omega,\theta;\eta)$ indeed becomes smaller when $\hbar\omega=\sqrt{2}\eta$.
\begin{figure}[!t]
\centering
\includegraphics[clip ,width=3.2in]{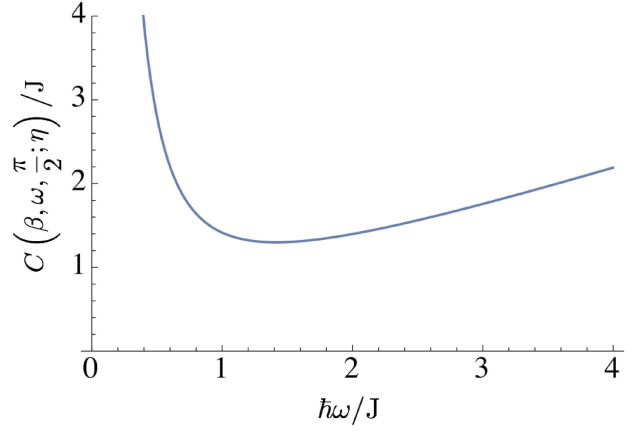}
 \caption{The graph of $C(\beta,\omega,\frac{\pi}{2};\eta)$ at $T=1\rm{K}$, and $\eta=1\rm{J}$ as a function of $\omega$.
 It takes its minimum at $\hbar\omega\approx\sqrt{2}\rm{J}=\sqrt{2}\eta$.}
\label{figure_pos_om}
\end{figure}

\begin{figure}[!t]
\centering
\includegraphics[clip ,width=3.2in]{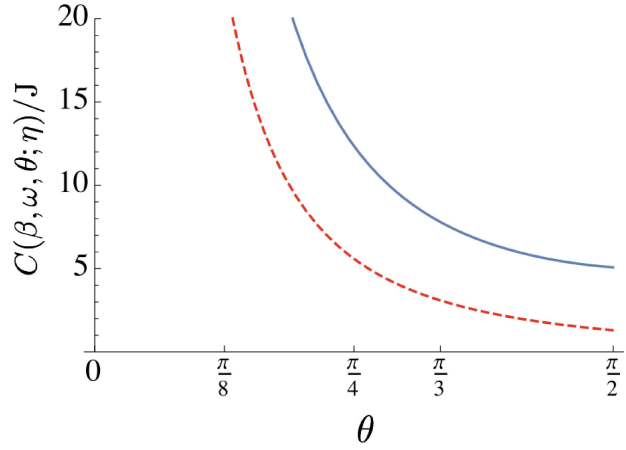}
 \caption{The graph of $C(\beta,\omega,\theta;\eta)$ at $T=1\rm{K}$, and $\eta=1\rm{J}$ with $\hbar\omega=10\rm{J}$ (solid (blue) curve) and $\hbar\omega=\sqrt{2}\rm{J}$ (dashed (red) curve).
 }
\label{figure_pos}
\end{figure}
\begin{figure}[!t]
\centering
\includegraphics[clip ,width=3.2in]{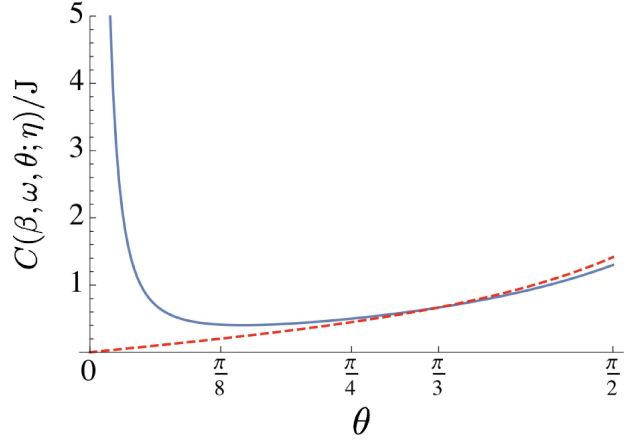}
 \caption{The graph of $C(\beta,\omega,\theta;\eta)$ as a function of $\theta$ at $T=1\rm{K}$, and $\eta=-1\rm{J}$.
 As $\theta\rightarrow 0$,
 it becomes small 
for the resonant frequency $\hbar\omega= 1\rm{J}=-\eta$ (dashed red curve), while it diverges for a non-resonant frequency $\hbar\omega=\sqrt{2}\rm{J}$ (solid blue curve).
 }
\label{figure_neg}
\end{figure}

Next, we consider the case when $\eta<0$, where $\Delta Q_{\sigma_{\theta},n}$
have to be positive to extract work.
In this case,
$\lim_{\theta\rightarrow 0}C(\beta,\omega,\theta;\eta)=(2\beta)^{-1}$
only when $\hbar\omega=-\eta$, otherwise it diverges to $+\infty$.
Thus, a kind of resonance occurs.
Since $C(\beta,\omega,\theta;\eta)>(2\beta)^{-1}$ holds in general,
that gives the infimum of the drawback.
Note that, however, $(2\beta)^{-1}$ is not the {\it minimum} since $\sigma_{\theta}$ with $\theta=0$ can no longer be used for the work extraction.
That is because $H$ is proportional to $\sigma_{0}=\sigma_{z}$.
Thus, in summary,
to make the maximum work large,
one should
tune $\omega$ to $-\eta$,
and to use low temperature,
small but non-zero $\theta$.
As an example, we show the graph of $C(\beta,\omega,\theta;\eta)$ as a function of $\theta$
at $T=1\rm{K}$, and $\eta=-1\rm{J}$
with resonant $\hbar\omega=1\rm{J}=-\eta$ (dashed (red) curve)
and non-resonant $\hbar\omega=\sqrt{2}\rm{J}$ (solid (blue) curve)
in Fig.~\ref{figure_neg}.
It shows that the coefficient indeed becomes small in $\theta\rightarrow 0$
for the resonant $\omega$.

  \bibliographystyle{apsrev4-1}
\bibliography{papers}
\end{document}